\documentclass[12pt]{article}
\usepackage[utf8]{inputenc}
\usepackage[parfill]{parskip}
\usepackage[a4paper, top=1in,bottom=1in,left=0.75in, right =0.75in]{geometry}
\usepackage{amsmath,amssymb}
\usepackage{enumitem}
\usepackage{amsthm,bm}
\usepackage[ruled,vlined]{algorithm2e}
\usepackage{setspace}
\usepackage{multirow,makecell}
\DeclareMathOperator{\E}{\mathbb{E}}
\newcommand{\bs}[1]{\boldsymbol{#1}}

\DeclareMathOperator*{\argmin}{arg\,min}

\newcommand{\indep}{\!\perp\!\!\!\perp}
\DeclareMathOperator{\vect}{vec}

\DeclareMathOperator{\rank}{rank}

\usepackage{tikz}
\usetikzlibrary{shapes.geometric, arrows, positioning}
\tikzstyle{process} = [rectangle, rounded corners, minimum width=3cm, minimum height=1cm, text centered, draw=black, fill=gray!10]
\tikzstyle{arrow} = [thick,->,>=stealth]

\newcommand{\EE}[2][]{\mathbb{E}_{#1}\left[#2\right]}
\newcommand{\PP}[2][]{\mathbb{P}_{#1}\left(#2\right)}
\newcommand{\Var}[2][]{\operatorname{Var}_{#1}\left(#2\right)}

\newcommand{\Cov}[2][]{\operatorname{Cov}_{#1}\left(#2\right)}
\newcommand{\Corr}[2][]{\operatorname{Corr}_{#1}\left[#2\right]}
\newcommand{\rr}[1]{#1^\text{r}}

\usepackage[small]{titlesec}
\titlelabel{\thetitle.\quad} 
\titleformat*{\section}{\bf\large}
\titleformat*{\subsection}{\bf\normalsize}
\titlespacing*{\section}
{0pt}{0ex}{0ex}
\titlespacing*{\subsection}
{0pt}{0ex}{0ex}

\usepackage{hyperref}
\hypersetup{
    colorlinks=true,
    linkcolor=blue,
    citecolor=blue,
    filecolor=magenta,      
    urlcolor=cyan,
}

\usepackage{cite}
\usepackage{natbib}
\usepackage{graphicx}
\graphicspath{ {images/} }
\usepackage{authblk}

\usepackage{setspace}
\newtheorem{theorem}{Theorem}
\newtheorem{remark}{Remark}
\newtheorem{assumption}{Assumption}
\newtheorem{corollary}{Corollary}

\newtheorem{lemma}{Lemma}
\newtheorem{definition}{Definition}
\newtheorem*{example}{Example}

\definecolor{WowColor}{rgb}{.75,0,.75}
\definecolor{SubtleColor}{rgb}{0,0,.50}

\newcounter{margincounter}

\title{Statistical Inference for Cell Type Deconvolution}

\date{}
\author[1]{Dongyue Xie}
\author[1]{Lin Gui}
\author[1]{Jingshu Wang \thanks{Corresponding author. Email address: jingshuw@uchicago.edu. }}
\affil[1]{Department of Statistics, The University of Chicago, Chicago, IL, USA}

\begin{document}

\doublespacing

\maketitle

\begin{abstract}

Integrating heterogeneous datasets across different measurement platforms is a fundamental challenge in many scientific applications. A common example arises in deconvolution problems, such as cell type deconvolution, where one aims to estimate the composition of latent subpopulations using reference data from a different source.
However, this task is complicated by systematic platform-specific scaling effects, measurement noise, and differences between data sources. For the problem of cell type deconvolution, existing methods often neglect the correlation and uncertainty in cell type proportion estimates, possibly leading to an additional concern of false positives in downstream comparisons across multiple individuals. We introduce MEAD, a statistical framework that provides both accurate estimation and valid statistical inference on the estimates. One of our key contributions is the identifiability result, which establishes the conditions under which cell type compositions are identifiable under arbitrary gene-specific scaling differences across platforms. MEAD also supports the comparison of cell type proportions across individuals after deconvolution, accounting for gene-gene correlations and biological variability. Through simulations and real-data analysis, MEAD demonstrates superior reliability for inferring cell type compositions in complex biological systems.

\end{abstract}


\newpage

\section{Introduction}\label{sec:intro}

Integrating data from diverse sources is a common strategy in modern data analysis, especially when direct measurements of necessary features are limited or unavailable. Leveraging external datasets, often collected from different individuals or using different technologies, can provide a cost-effective solution to fill information gaps. However, such integration introduces additional biases and variability. Ignoring the heterogeneity across datasets may increase the risk of false positives in downstream statistical analyses.

Solutions to these integration challenges are typically model- and context-specific. In this paper, we focus on a key example in genetics: estimating individual-level cell-type proportions through cell-type deconvolution. This task is represented by the model (Figure~\ref{fig:MEAD}):
\begin{equation}\label{eq:general_model}
    \bs y_i = \tilde{\bs X}_i \bs p_i + \bs \epsilon_i,
\end{equation}
where the goal is to estimate $\bs p_i$, the vector of cell-type proportions for individual $i$, despite the unavailability of the design matrix $\tilde{\bs X}_i$ which must be approximated using external data.

Cell-type deconvolution is a widely used computational approach to estimate $\bs p_i =\left(p_{i1},\cdots,p_{iK}\right)\in [0,1]^K$ with $\sum_{k=1}^K p_{ik} = 1$, the relative abundances of $K$ cell types in a bulk tissue sample $i$ \citep{newman2019determining,wang2019bulk,menden2020deep,dong2021scdc}. These proportions provide insights into tissue composition and are often associated with disease development \citep{fridman2012immune, mendizabal2019cell}. However, it is challenging for current experimental technologies to directly measure cell-type composition across large cohorts \citep{jew2020accurate,o2019complementary}. Instead, bulk RNA-seq provides a noisy gene expression vector $\bs y_i \in \mathbb{R}^G$ for the average gene expressions within a target tissue (individual), where $G$ is the number of genes, without access to the cell-type-specific expression matrix $\tilde{\bs X}_i\in \mathbb{R}^{G\times K}$. To address this, deconvolution methods approximate $\tilde{\bs X}_i$ using reference datasets, often single-cell RNA-seq (scRNA-seq) from other reference individuals. Cell type deconvolution is also popular for spatial transcriptomics \citep{gaspard2025cell}, and similar problems arise in other domains, such as admixture estimation in population genetics \citep{alexander2009fast}.

\begin{figure}[t]
\centering
    \includegraphics[width = 0.95\textwidth]{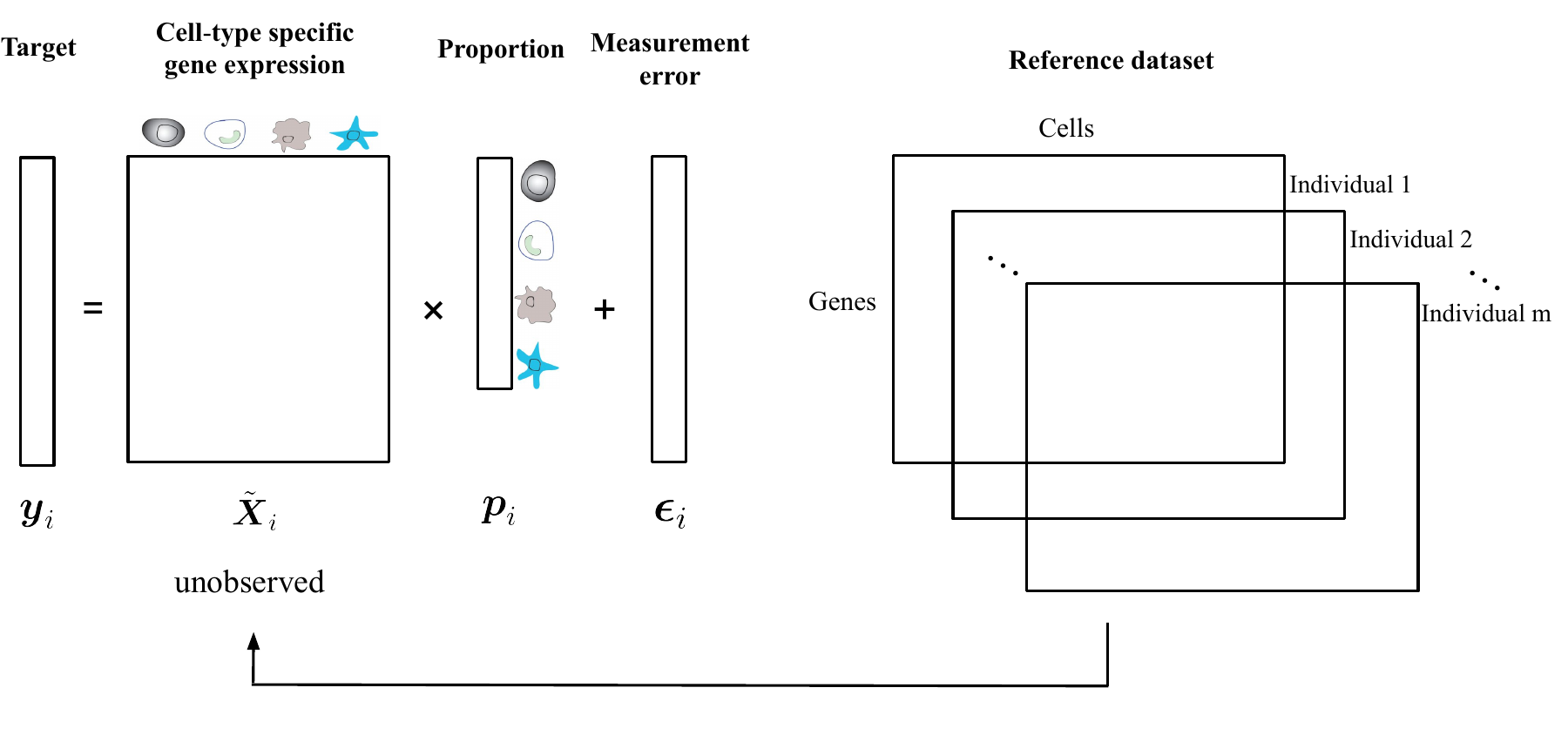}
    \caption{Overview of cell type deconvolution}
    \label{fig:MEAD}
\end{figure}

Despite the apparent simplicity of the linear model in \eqref{eq:general_model}, several challenges complicate statistical inference in practice:

First, the approximation of $\tilde{\bs X}_i$ may differ substantially from the true expression matrix due to both biological variation between target and reference individuals and platform-specific measurement biases.
 Such differences can lead to biased or unidentifiable estimates of $\tilde{\bs X}_i$. Recent methods \citep{jew2020accurate,cable2022robust} allow for unknown gene-specific scaling differences across the target and reference data platforms, but identifiability of cell-type proportions may be lost if such differences are entirely unrestricted.

 Second, inference is further complicated by gene-gene correlations and heterogeneous noise levels. To solve model \eqref{eq:general_model} for any target individual $i$, although genes can be treated as “samples” in the model, they are not independent and often vary in scale. Moreover, preprocessing steps such as normalization can induce additional dependencies among genes.


  Third, the cell-type proportion vector $\bs p_i$ must lie in the simplex, satisfying non-negativity and summing to one. These constraints complicate estimation but also help mitigate scaling differences between true and approximated expression levels.

  Finally, in downstream analyses, researchers often treat estimated proportions as known when comparing groups of target individuals and evaluate whether the cell proportions associate with any features (such as disease status) of the target individuals.
 However, this ignores uncertainty in $\widehat{\bs p}_i$ and the dependence across target individuals $i$ induced by shared reference data. Whether such simplifications affect downstream inference remains unclear.

In this paper, we introduce MEAD (Measurement Error Adjusted Deconvolution), a new method that addresses these challenges and provides valid inference for both individual-level cell-type proportions and cross-individual comparisons. MEAD improves robustness by avoiding strong distributional assumptions and explicitly correcting for measurement error, platform scaling differences, and gene-gene correlation.

One key contribution is showing that the common practice of treating $\widehat{\bs p}_i$ as truth in downstream analyses is justified when the number of target individuals $N$ is small relative to the number of genes $G$. In this regime, estimation error in $\widehat {\bs p}_i$ is negligible compared to noise across samples. However, when $N$ grows such that $N/G \not\to 0$,
the naive approach can inflate false positives, except when testing the global null hypothesis that none of the proportions change with any features of interest, for which we show it remains valid. While our theory focuses on MEAD, simulations suggest these insights generalize to other deconvolution methods.

A second contribution is establishing necessary and sufficient conditions for identifying $\bs p_i$ under arbitrary gene-specific cross-platform scaling differences. We show that simply increasing the number of target individuals is insufficient for identifying the cell type proportions, and 
additional structure is required in the cell-type-specific expression matrix of the selected genes. Our findings provide a comprehensive understanding of statistical inference in cell type deconvolution.







\section{Model Setup and Identification}\label{sec:model}

To address the challenges described in Section~\ref{sec:intro}, we first specify assumptions for both the reference and target datasets, focusing on their shared structure and key differences. We begin by introducing the data and defining the estimand of cell type proportions without major assumptions. We then present the main assumptions and identifiability conditions, and conclude with a simplified model for estimation and inference in Sections~\ref{sec:estimation}--\ref{sec:infer_multi}.

Let $K$ denote the number of cell types and $G$ be 
the number of genes measured in both the target and reference data. 
 Let $N$ and $M$ be the number of individuals in the target and reference datasets, respectively. While $N$ is often large, potentially including hundreds of individuals, $M$ is typically small, as most scRNA-seq experiments sample only a few subjects. We use unbolded lowercase letters for scalars, bold lowercase for vectors, and bold uppercase for matrices.

For a target individual $i$ measured in bulk RNA-seq, we model the observed gene expression as:
\begin{equation}\label{bulk0}
    \bm{y}_i = \gamma_i\text{diag}(\bm \alpha)\bm X_i\bm p_i + \bm\epsilon_i, \quad \E\left(\bm\epsilon_i\mid \bm X_i\right) = \bm 0.
\end{equation}
Here, matrix $\bm X_i\in \mathbb{R}^{G\times K}$ is the true cell-type specific gene expression matrix for individual $i$, averaged across cells within each cell type. The gene-specific factors $\bm \alpha = (\alpha_1, \cdots, \alpha_G)$ account for technical variation across genes, such as gene length and GC content \citep{benjamini2012summarizing}. The individual-specific scalar  $\gamma_i$ captures sequencing depth and tissue size \citep{wang2018gene}. The effective expression matrix $\tilde{\bm X}_i = \gamma_i \mathrm{diag}(\bm \alpha) \bm X_i$ as in model \eqref{eq:general_model} thus incorporates both biological and technical variation.

For a reference individual $j$ in scRNA-seq, the corresponding cell-type level pseudo-bulk expression can be modeled as (see Section~\ref{sec:cell_model} for derivation):
\begin{equation}\label{ref0}
    \rr{\bm{Z}}_j = \rr{\gamma_j}\text{diag}(\rr{\bm \alpha})\rr{\bm X_j} + \rr{\bm E_j},\quad \E\left(\rr{\bm E_j}\mid \rr{\bm X_j}\right) = \bm 0.
\end{equation}
Here, $\rr{\bm{Z}}_j \in \mathbb{R}^{G\times K}$ contains observed average expression within each cell type, serving as input to our framework. The true expression matrix $\rr{\bm X}_j\in \mathbb{R}^{G\times K}$ and scaling factors  $\rr{\gamma_j}$ and $\rr{\bm \alpha} = (\rr{\alpha}_1, \cdots, \rr{\alpha}_G)$ represent the true expression levels and technical factors in the reference data, in parallel to $\bs X_i, \gamma_i$ and $\bs \alpha$ in the target data. However, due to differences in sequencing platforms, the gene-specific factors $\rr{\alpha}_g$ in the scRNA-seq data may differ substantially from $\alpha_g$ for many genes \citep{jew2020accurate,cable2022robust}.

\subsection{Identifiability allowing arbitrary cross-platform differences}

We now introduce assumptions necessary to identify the cell-type proportions $\bs p_i$, allowing arbitrary differences in gene-specific scaling factors $\bm \alpha$ and $\rr{\bm \alpha}$ across platforms. These factors are treated as fixed, while true cell-type-specific gene expressions $\bm X_i$ and $\rr{\bm X_j}$ are modeled as random across individuals.

Our first key assumption links the gene expression distributions of the target and reference individuals by assuming they are drawn from a common population:

\begin{assumption}[Homogeneous population]\label{asp:homo_popu}
Both the target and reference individuals are independently sampled from the same population:
$$\bm X_1, \cdots, \bm X_N \overset{i.i.d.}{\sim} F_{\bm X},\quad \rr{\bm X_1}, \cdots, \rr{\bm X_M} \overset{i.i.d.}{\sim} F_{\bm X}.$$
Additionally, the noise terms $\bs\epsilon_1, \cdots, \bs \epsilon_N$ and $\rr{\bs E}_1, \cdots, \rr{\bs E}_M$ are mutually independent. 
\end{assumption}

\begin{remark}
    Assumption \ref{asp:homo_popu} implies unbiased sampling from the population, which might not always hold in practice. Potential relaxations are discussed in Section \ref{sec:discuss}. 
\end{remark}

Under Assumption~\ref{asp:homo_popu}, define the rescaled population-level cell-type true gene expression matrix  $\bs U  = \text{diag}(\rr{\bm \alpha})\EE{\bm X_i} \in \mathbb{R}^{G\times K}$, with entries $\mu_{gk}$, and let $\bs \Lambda = \text{diag}(\lambda_1, \cdots,$ $ \lambda_G)$ where $\lambda_g = \alpha_g/\rr{\alpha}_g$ captures the gene-specific cross-platform scaling ratios. Then, models \eqref{bulk0} and \eqref{ref0} become
\begin{equation}\label{bulk}
\begin{aligned}
 & \text{Target data: } \quad \bm{y}_i = \gamma_i\bm \Lambda\bm U\bm p_i + \bm\epsilon_i',\quad i = 1,\cdots,N \\
  & \text{Reference data: } \quad \rr{\bm{Z}}_j = \rr{\gamma_j}\bm U + \rr{\tilde{\bm E}_j},\quad j = 1,\cdots,M  
\end{aligned}
\end{equation}
where the error terms $\bm\epsilon_i'=\gamma_i\bs\Lambda\left(\text{diag}(\bm \alpha^r)\bm X_i - {\bm U}\right)\bm p_{i}+\bm\epsilon_i$ and $\rr{\tilde{\bm E}_j}=\rr{\gamma_j}\left(\text{diag}(\bm \alpha^r)\bm X_i - {\bm U}\right)+\rr{\bm E}_j$ both have zero means.


 Because the proportion vector $\bs p_i$ must satisfy $\bs p_i^\top \bs 1 = 1$ and the scalars $\gamma_i$ and $\gamma_j^r$ absorb overall scaling for each individual, we impose the following constraints without loss of generality:

\begin{assumption} [Scaling constraints]\label{asp:scaling_constraints}
    $\bs U$ and $\bs \Lambda$ satisfy the following scaling constraints:
    $$\sum_{k = 1}^K\sum_{g = 1}^G\mu_{gk} = KG, \quad \sum_{g = 1}^G \lambda_g = G.$$
\end{assumption}

Given the reference data, $\bs U$ is identifiable (Corollary~\ref{cor:identification1}). 
Identifying $\bs P = [\bs p_1, \cdots, \bs p_N]\in \mathbb{R}^{K \times N}$ then reduces to identifying $\bs B =(\bs \beta_1,\cdots, \bs \beta_N)$ and $\bs \Lambda$ where $\bs\beta_i = \gamma_i\bs p_i$, from known $\bs \Lambda \bs U \bs B$ and $\bs U$.
The following theorem gives necessary and sufficient conditions:

\begin{theorem}\label{thm:identify}
Under Assumption~\ref{asp:homo_popu} and $\rank(\bm P) = K$, the proportion matrix $\bm P$ in model \eqref{bulk} is identifiable if and only if:
\begin{enumerate}
    \item[a.] $\bs U$ has full rank $K$; 
    \item[b.] For any partition $\{I_1, I_2, \cdots, I_t\}$ of the genes indices
    with $t \geq 2$, the sum of the ranks of the corresponding submatrices $\bs U_{I_s}$ of $\bs U$ satisfies $\sum_{s = 1}^t \rank(\bs U_{I_s}) > K$.
\end{enumerate} 
\end{theorem}

Condition (a) ensures sufficient informative genes available to estimate the cell-type proportions, as in linear regression.
Condition (b) prevents the decomposition of the signal into disjoint gene subsets, which would confound the estimation of  $\bs \Lambda$ and $\bs P$. As a counter-example, $\bs P$ cannot be identified if all genes are perfect marker genes for the involved cell types.

\begin{example} [counter-example: perfect marker genes] 
Suppose only perfect marker genes that only express in a particular cell type are used: for each gene $g \in I_k$,  $u_{gk} > 0$ while $u_{gk'} = 0$ for $k' \neq k$. Many methods recommend such gene selection for deconvolution \citep{chen2018profiling,newman2019determining}. Define any non-negative vector $\bs \delta = (\delta_1, \cdots, \delta_K)$, and let 
$$\tilde\lambda_{g}=\lambda_g/\delta_k \text{ if } g\in I_k, \quad \tilde p_{ik}=\delta_kp_{ik}/\sum_{l=1}^K\delta_lp_{il}, \quad \tilde \gamma_i = \gamma_i\sum_{l=1}^K\delta_lp_{il}.$$
Then, for all $i$, we have:
$$\tilde \gamma_i\tilde{\bs \Lambda}\bs U\tilde{\bs p}_i = \gamma_i\bs\Lambda \bs U \bs p_i, $$
implying $\bs P$ is not identifiable.
\end{example}

Finally, Theorem~\ref{thm:identify} requires $\text{rank}(\bm P) = K$, which implies $N \geq K$, a condition highlighted in \cite{jew2020accurate} and \cite{cable2022robust}. However, this is not sufficient on its own; gene selection also plays a critical role in ensuring identifiability under arbitrary cross-platform differences.

\subsection{The final model with non-informative gene-specific scaling ratios}
While Theorem~\ref{thm:identify} allows arbitrary gene-specific scaling ratios $\{\lambda_g, g = 1, 2, \cdots, G\}$, estimating these ratios and accounting for their uncertainty is difficult, particularly when $G$ is large. To simplify inference, we follow \cite{cable2022robust} to assume non-informative $\lambda_g$, which is more flexible than assuming a constant ratio across genes, as done in many existing methods.

\begin{assumption}[Non-informative gene-specific scaling ratios]\label{asp:bias}
The scaling ratios $\{\lambda_g = \alpha_g/\alpha_g^r\}$ are independently distributed with mean $1$ and variance $\sigma_0^2$: $\lambda_g\overset{i.i.d.}{\sim }[1, \sigma_0^2]$.
\end{assumption}

Under Assumption~\ref{asp:bias}, $\lambda_g$ is independent of the expression matrix $\bs U$, allowing us to absorb its effect into the error term and simplify model~\eqref{bulk} as follows:
\begin{equation}\label{eq:final_bulk}
\begin{aligned}
   &\text{Target data: } \quad \bs y_i = \bs U \bs \beta_i + \bs e_i, \quad \beta_{ik}\geq 0, \quad \bs p_i = \bs \beta_i/|\bs \beta_i|_1,\quad i = 1,\cdots,N  \\
    & \text{Reference data: } \quad \rr{\bm{Z}}_j = \rr{\gamma_j}\bm U + \rr{\tilde{\bm E}_j}, \quad j = 1,\cdots,M
\end{aligned}
\end{equation}
where $\bs e_i =(\bs \Lambda - \bs I)\bs U\bs \beta_i + \bs\epsilon_i'$ captures both biological variation and measurement error in the target data, and satisfies $\EE{\bs e_i} = \bs 0$.

In this model, the cross-platform scaling ratios $\lambda_g$ do not cause systematic bias in estimating $\bs \beta_i$, but they do increase uncertainty and introduce dependence across individuals. In particular, the errors $\bs e_i$ are no longer independent across target individuals, complicating the inference when comparing across target individuals. 
Furthermore, the model remains flexible and does not require specific distributional assumptions, allowing for heterogeneity and correlation across genes or cell types within each individual.



\section{Model Estimation} \label{sec:estimation}

We first estimate $\bs U$ from the reference data. Estimating the coefficients $\bs \beta_i$ in model~\eqref{eq:final_bulk} then becomes an error-in-variable linear regression problem with non-negativity constraints and heteroscedastic, correlated noise. We develop a new procedure for estimating the cell-type proportions $\bs p_i$, adapting the general structure of existing deconvolution methods but incorporating corrections for the measurement errors in $\widehat{\bs U}$, as in classical error-in-variable models.

\subsection{Estimation of $ \mathbf U$}\label{sec:estimation_U}

We estimate $\bs U$ directly from the normalized reference data. Let $\rr z_{jgk}$ be the $(g,k)$-th entry of matrix $\rr{\bs Z}_j$, and $\rr{\bm z}_{jg}$ denote the $g$-th row vector of $\rr{\bs Z}_j$. Also, denote $\bs \mu_g = (\mu_{g1}, \cdots, \mu_{gK})$ as the $g$-th row vector of $\bs U$. Then we estimate the individual-specific scaling factor $\rr \gamma_j$ and compute the sample average as:
$$\widehat \gamma_j = \frac{\sum_{k = 1}^K\sum_{g = 1}^G \rr z_{jgk}}{KG}, \quad  \widehat {\bs{\mu}}_{g} = \frac{1}{M}\sum_{j = 1}^M \frac{\rr{\bs{z}}_{jg}}{\widehat \gamma_j}.$$
Let $\widehat{\bs U} =(\widehat{\bs \mu}_1, \cdots, \widehat{\bs \mu}_G)^\top$. To quantify uncertainty, define $\bm V_g = \text{Cov}\left[\widehat{\bm\mu}_g\right]$ and estimate it by 
the sample covariance: 
$$\widehat{\bm V}_g = \frac{1}{M(M-1)}\sum_{j = 1}^M \left(\frac{\rr{\bm z}_{jg}}{\widehat \gamma_j} - \widehat{\bm\mu}_g\right)\left(\frac{\rr{\bm z}_{jg}}{\widehat \gamma_j} - \widehat{\bm\mu}_g\right)^\top.$$

\begin{remark}
While $\widehat{\bs \mu}_g$ and $\widehat{\bs V}_g$ are not unbiased due to unknown $\gamma_j$, we show in Section \ref{sec:inference_one} that  $\widehat\gamma_j$ is a consistent estimator of $\rr\gamma_j$. As a result, individual $\widehat{\bs \mu}_g$ and $\widehat{\bs V}_g$ become asymptotically unbiased as $G \to \infty$.  
\end{remark}

\subsection{Estimation of cell type proportions for each individual}\label{sec:est_2}  
For the target data, model~\eqref{eq:final_bulk} 
 takes the form of a linear regression with genes as ``samples''. Since $\widehat {\bs U}$ is a noisy estimate of $\bs U$, estimating $\bs \beta_i$ becomes an error-in-variables regression problem \citep{fuller2009measurement}.
Unlike the classical setting, however, gene-level variability and correlation are substantial, requiring gene-specific weighting for efficient estimation \citep{wang2019bulk}. 

We introduce a diagonal weight matrix $\bm W = \text{diag}(\bs w)$, where $\bs w =(w_1, \cdots, w_G)$. Many existing deconvolution methods focus on choosing $\bs W$ appropriately \citep{wang2019bulk,newman2019determining}. In Section \ref{sec:weight_matrix}, we present our empirical approach to selecting $\bs W$ and compare it with existing strategies. For now, we treat $\bs W$ as fixed.

Adapting classical bias correction for errors-in-variables regression, we estimate $\bm \beta_i$ using the following equation:
\begin{equation}\label{est_eq}
\begin{split}
    \bm \phi(\bm \beta_i) &= \widehat{\bm U}^\top \bm W \bm y_i - (\widehat{\bm U}^\top \bm W \widehat{\bm U} - \widehat{\bs V}) \bm{\beta}_i = \bm 0, 
\end{split}
\end{equation}
where $\widehat{\bs V}  = \sum_{g=1}^G w_g\widehat{\bs V}_g$. 
If each $\widehat{\bs \mu}_g$ and $\widehat{\bs V}_g$ are asymptotically unbiased for $\bs \mu_g$ and $\bs V_g$ as $G \to \infty$, then $\EE{\bs \phi(\bm \beta_i)} \to \bs 0$ at the true value of $\bs \beta_i$. Notably, this remains asymptotically valid under gene-level heterogeneity and correlation.

The unconstrained solution of Equation~\eqref{est_eq} is:
\begin{equation*}\label{adjwls}
     \widehat{\bs \beta}_i = (\widehat{\bs U}^\top \bs W \widehat{\bs U} - \widehat{\bs V})^{-1}\widehat{\bs U}^\top \bs W \bs y_i.
\end{equation*}
To enforce non-negativity, we define $\widehat {\bs \beta}_i^\star$ either by truncating negative entires, i.e., $\widehat {\bs \beta}_i^\star = \widehat {\bs \beta}_i^{\text{trunc}} = \widehat {\bs \beta}_i \vee \bs 0$, or solving a contrained problem:
$$\widehat {\bs \beta}_i^\star = \widehat {\bs \beta}_i^{\text{constr}} = \argmin_{{\bs \beta}_i \succeq \bs 0}(\bs y_i - \widehat{\bs U}\bs \beta_i)^\top \bs W \bs (\bs y_i - \widehat{\bs U}\bs \beta_i) - \bs \beta_i^\top\widehat{\bs V} \bs \beta_i.$$
In either case, $\widehat {\bs \beta}_i^\star \neq \widehat {\bs \beta}_i$ only when any element $\widehat\beta_{ik} < 0$ in $\widehat {\bs \beta}_i$. The final estimate of the cell-type proportions is given by
 $  \widehat {\bs p}_i = \bs g(\widehat {\bs \beta}_i^\star) \overset{\Delta}{=} \widehat {\bs \beta}_i^\star /\|\widehat {\bs \beta}_i^\star\|_1.$

\section{Statistical inference for a single target individual}\label{sec:inference_one}

We analyze the theoretical properties of the estimator $\widehat{\bs p}_i$ and construct confidence intervals for each component $p_{ik}$, focusing on a single target individual $i$. Inference across multiple individuals is discussed in Section~\ref{sec:infer_multi}. Our analysis considers the common setting where the number of genes $G$ is large, while the number of reference individuals $M$ and cell types $K$ are relatively small. We therefore work in the asymptotic regime $G \to \infty$ with fixed $M$ and $K$.

Although model~\eqref{eq:final_bulk} allows for gene-gene dependence, formal inference requires structural assumptions on this dependence. We adopt the notion of an ``almost sparse'' gene co-expression network (GCN) \citep{langfelder2008wgcna,zhang2012weighted,russo2018cemitool}, formalized through the following concept of a dependency graph:

\begin{definition}[Dependency graph, \cite{chen2004normal}]
Let $\{X_i,i \in \mathcal V\}$ be a set of variables indexed by the vertices of a graph $\mathcal G = (\mathcal V, \mathcal E)$. Then $\mathcal G$ is a dependency graph if, for any disjoint subsets $\Gamma_1$ and $\Gamma_2$ in $\mathcal V$ with no edge connecting them, the collections $\{X_i,i \in {\Gamma}_1\}$ and $\{X_i,i \in \Gamma_2\}$ are independent.
\end{definition}

Let $\bs e_i = (e_{i1}, \cdots, e_{iG})$ and $\rr{\tilde{\bm E}_j}=\left(\rr{\bm \epsilon}_{j1}, \cdots, \rr{\bm \epsilon}_{jG}\right)^\top$ denote the gene-level error vectors in the target and reference models, respectively. We assume that the noise terms follow a structured but sparse dependency pattern across most genes, while allowing a small subset to exhibit arbitrary dependencies. This is formalized as follows: 
\begin{assumption}[Dependence structure across genes]\label{assmp:dependence}
There exists a subset $\mathcal V \subset \{1, 2, \cdots, G\}$ such that the noise terms $\left\{\left(\{e_{ig}\}, \{\rr{\bm{\epsilon}}_{jg}\}\right), g \in \mathcal V\right\}$, form a dependency graph $\mathcal G$. with maximum degree $D \leq s$ for some constant $s$. The complement satisfies $|\mathcal V^c|/\sqrt G \to 0$ as $G \to \infty$. 
\end{assumption}

\begin{remark}
 We believe that Assumption~\ref{assmp:dependence} reasonably approximates the complex gene–gene dependence in real data. Prior work often assumes sparse GCN, or at least sparsity in strong correlations \citep{langfelder2008wgcna, iacono2019single}, supported by empirical findings such as those in Figure 3 of \citet{agarwal2020data}, where most pairwise gene–gene sample correlations are close to zero. 
\end{remark}


 


\subsection{Consistency}

To establish the consistency of $\widehat {\bs p}_i$, we avoid imposing parametric distributional assumptions and instead assume that the observed data have uniformly bounded moments across genes. This ensures that a small subset of genes does not dominate the variability in gene expression.

\begin{assumption}[Bounded moments]\label{assumption_consistency}
For model~\eqref{eq:final_bulk}, assume 
\begin{itemize}
    \item[a.] As $G\to\infty$, $\frac{1}{G}\bs U^\top\bs W\bs U \to \bs\Omega$, where $\bs\Omega \succ  0$ is  positive definite and $\bs W$ is fixed. 
    \item[b.] There exists $\delta >0$ and a constant $C$ such that 
    $$\max_{i, g}\E\left[y_{ig}^{4 + \delta}\right]\leq C, \quad\max_{j, g, k}\E\left[(\rr z_{jgk})^{4 + \delta}\right]\leq C, \quad\max_{g}\E\left[\lambda_g^{4 + \delta}\right]\leq C,$$
    and row vectors of $\bs U$ are bounded: $\max_g\|\bs\mu_g\|_1\leq C$. 
    \item[c.] The weights $w_g$ are uniformly bounded, with $0\leq w_g \leq C$ for all $g$.
\end{itemize}
\end{assumption} 
Under these conditions, we can establish the following consistency result:

\begin{theorem}\label{thm:consistency}
Under Assumptions~\ref{asp:homo_popu}-\ref{assumption_consistency}, with $M$ and $K$ fixed and $G \to \infty$, we have:
\begin{equation*} \label{eq:consistency_omega}
    \widehat {\bs \Omega} \overset{\Delta}{=} \frac{1}{G}(\widehat{\bs U}^\top\bs W\widehat{\bs U}-\widehat{\bs V})\overset{p}{\to} \bs\Omega, \quad \widehat{\gamma}_j\overset{p}{\to}\rr\gamma_j, \quad \text{for each } j = 1, \dots, M,
\end{equation*}
and for any target individual $i$, $ \widehat{\bs p_i}\overset{p}{\to}\bs p_i.$
   
\end{theorem}




\subsection{Asymptotic normality}
To analyze the asymptotic distribution of $\widehat{\bs p_i}$, we require an additional condition ensuring that the variance $\Cov{\sqrt G \widehat{\bs p}_i}$ grows with $G$, even when the gene-gene correlations are present.  
To ensure this, define the average reference noise $\rr{\bar{\bs \epsilon}}_g=\frac{1}{M}\sum_{j=1}^M (\rr{\bs \epsilon}_{jg}/\rr{\gamma}_j)$ and the following items: 
$$\rr{\bm H} = \sum_{g=1}^G \left[w_g\left(\rr{\bar{\bs \epsilon}}_g(\rr{\bar{\bs \epsilon}}_g+\bs \mu_g)^\top - \text{Cov}_M\left(\rr{\bar{\bs \epsilon}}_g\right)\right) - \frac{\bm 1_K^\top\rr{\bar{\bs \epsilon}}_g}{K}\bs\Omega\right],\quad \bs s_i = \sum_{g = 1}^G w_g e_{ig}\rr{\bar{\bs \epsilon}}_g,$$
where
$\text{Cov}_M\left(\rr{\bar{\bs \epsilon}}_g\right) \overset{\Delta}{=}\left(\sum_{j=1}^M\left(\rr{{\bs \epsilon}}_{jg} - \rr{\bar{\bs \epsilon}}_g\right)\left(\rr{{\bs \epsilon}}_{jg} - \rr{\bar{\bs \epsilon}}_g\right)^\top\right)/[M(M-1)]$ and $\bm 1_K = (1, \cdots, 1)\in \mathbb{R}^K$.

We introduce the following minor technical assumption:
\begin{assumption}[Non-collapsing variance]\label{assmp:variance}
As $G \to \infty$, we assume 
$\lim_{G\to \infty}\Cov{\vect\left(\rr{\bs H}\right)}/G$
and  
$\lim_{G\to \infty}\Cov{\bs s_i}/G$
exist, and at least one of the limits is positive definite.
\end{assumption}

Empirical studies suggest gene-gene correlations are mostly positive, which increases overall variance, making Assumption~\ref{assmp:variance} a practically reasonable assumption. 
Then, using the central limit theorem for weakly dependent variables with local dependence \citep{chen2004normal}, we obtain the asymptotic distribution of $\widehat {\bs p}_i$:
\begin{theorem}\label{thm:clt}
Under Assumptions~\ref{asp:homo_popu}-\ref{assmp:variance}, for each target individual $i$, if $p_{ik} > 0$ for all $k$, then for each target individual $i$, as $G \to \infty$:
$$\sqrt G(\widehat {\bs \beta}_i^\star - \bs \beta_i) \overset{d}{\to} \mathcal{N}(\bs 0, \bs \Omega^{-1}\bs \Sigma_i \bs \Omega^{-1}),$$
where $\bm\Sigma_i \overset{\Delta}{=} \lim_{G\to \infty}\Cov{\bs \phi(\bs \beta_i)}/G   \succ 0 $. As a result,
\begin{equation*} 
    \sqrt{G}(\widehat {\bs p}_i - \bs p_i)\overset{d}{\to} \mathcal N\left(\bs 0, \nabla\bs g( \bs\beta_i)\bs\Omega^{-1}\bs\Sigma_i\bs\Omega^{-1}\nabla\bs g(\bs\beta_i)^\top\right),
\end{equation*}
where $\nabla\bs g(\bs x)$ is the Jacobian matrix of the standardizing function $\bs g(\bs x) = \bs x/|\bs x|_1$.
\end{theorem}
Theorem~\ref{thm:clt} shows that the variance of $\widehat{\bs p}_i$ increases with gene-gene correlations (through $\bs \Sigma_i$) and with homogeneity across cell types (through $\bs\Omega^{-1}$). 
To construct confidence intervals for each  $p_{ik}$, we estimate the asymptotic covariance by:
$$\widehat{\text{Cov}}\left[\sqrt G\widehat{\bm p}_i\right] = \nabla\bs g( \widehat{\bs\beta}^\star_i)\widehat{\bs\Omega}^{-1}\widehat{\bs\Sigma}_i\widehat{\bs\Omega}^{-1}\nabla\bs g(\widehat{\bs\beta}^\star_i)^\top,$$
where consistency of $\widehat{\bs\beta}^\star_i$ and $\widehat{\bs\Omega}$ is guaranteed by Theorem~\ref{thm:consistency}.
Accurate estimation of ${\bs\Sigma}_i$ remains the main challenge in the presence of gene-gene dependence. As $\bs \phi(\bs \beta) = \sum_{g = 1}^G \bs \phi_g(\bs \beta_i)$ where
\begin{equation}\label{eq:phi_g}
    \bs \phi_g(\bs \beta_i) = w_g \widehat{\bs \mu}_g^\top y_{ig} - (w_g \widehat{\bs \mu}_g^\top\widehat{\bs \mu}_g - w_g\widehat{\bs V}_g)\bs \beta_i, 
\end{equation}
    we define the sandwich-type estimator: $$\widehat{\bs\Sigma}_i= \frac{1}{G}\left(\sum_{g = 1}^G \bs\phi_g(\widehat {\bs\beta}_i^\star)\bs\phi_g(\widehat{\bs{\beta}}_i^\star)^\top + \sum_{(g_1,g_2) \in \mathcal A}\bs\phi_{g_1}(\widehat {\bs\beta}_i^\star)\bs\phi_{g_2}(\widehat {\bs\beta}_i^\star)^\top\right),$$
where the set $\mathcal A = \{(g_1, g_2): (e_{ig_1}, \rr{\bm \epsilon}_{jg_1}) \text{ and } (e_{ig_2}, \rr{\bm \epsilon}_{jg_2}) \text{ are not independent}\}$. 

In Section \ref{sec:practical}, we will discuss how we estimate $\mathcal A$ (Section \ref{sec:est_cor}) and apply finite-sample corrections (Section \ref{sec:finite_correct}) to obtain good coverage for our confidence intervals in practice. 

\subsection{Softplus transformation}

Theorem~\ref{thm:clt} requires that all $p_{ik}> 0$ to ensure that the estimator $\widehat {\bs\beta}_i^\star$ that satisfies the non-negativity constraints is asymptotically well-behaved. However, in practice, some cell types may be completely absent in a given individual, leading to a point mass at zero for $\widehat {\bs\beta}_i^\star$ (thus $\widehat{\bs p}_i$) and violating the regularity conditions for its asymptotic normality.

To address this, we apply a Softplus transformation of $\widehat {\bs\beta}_i$ instead of directly using $\widehat {\bs\beta}_i^\star$ to smooth the non-negativity constraint:
\begin{equation*}
    \widehat \beta_{ik}^{(a)} = h_a(\widehat\beta_{ik})\overset{\Delta}{=}\frac{1}{a}\log(1 + e^{a\widehat\beta_{ik}}),\quad
\widehat{\bm p}_i^{(a)} = \frac{\widehat{\bm \beta}_i^{(a)}}{|\widehat{\bm \beta}_i^{(a)}|_1},
\end{equation*}
where $\widehat{\bm \beta}_i^{(a)} =(\widehat{\beta}_{i1}^{(a)}\cdots, \widehat\beta_{iK}^{(a)})$, and $a>0$ is a tuning parameter. Let $\bm \beta_i^{(a)}$ with each element $\beta_{ik}^{(a)} = h(\beta_{ik})$, and $\bm p_{i}^{(a)} = \bm g(\bm \beta_i^{(a)})$.  As $a \to \infty$, we recover the original quantities: $\bm \beta_i^{(a)} \overset{a\to \infty}{\to}\bs \beta_i$, and hence $\bm p_i^{(a)} \to\bs p_i$.

Since $h_a(\cdot)$ is smooth, we can derive the asymptotic distribution of $\widehat {\bm p}_i^{(a)}$:
\begin{corollary}\label{cor:CLT_softplus}
    Under Assumptions~\ref{asp:homo_popu}-\ref{assmp:variance}, the Softplus-transformed estimator $\widehat{\bs p}_i^{(a)}$ is asymptotically normal as $G \to \infty$:
    \begin{equation*} 
    \sqrt{G}(\widehat {\bm p}_i^{(a)} - \bm p_i^{(a)})\overset{d}{\to} \mathcal N\left(\bm 0, \nabla\bm g(\bm\beta_i^{(a)})\bm \Gamma\bm\Omega^{-1}\bm\Sigma_i\bm\Omega^{-1}\bm \Gamma\nabla\bm g(\bm\beta_i^{(a)})^\top\right),
\end{equation*}
 where $\bs \Gamma = \text{diag}(\gamma_{11}, \cdots, \gamma_{KK}) \in \mathbb{R}^{K\times K}$ and 
$\gamma_{ii} = h_a'(\beta_{ik})=e^{a\beta_{ik}}/(1 + e^{a\beta_{ik}}).$
\end{corollary}
For any finite $a$, $\widehat{\bm p}_i^{(a)}$ is not a consistent estimator of $\bs p_i$, but the bias diminishes with $a$ since $\bm p_i^{(a)} \overset{a\to \infty}{\to}\bs p_i$. Specifically, as $a \to \infty$, $\widehat{\bs \beta}_i^{(a)}\to \widehat {\bs \beta}_i^{\text{trunc}} =\widehat{\bs \beta}_i \vee 0$, thus  $\widehat{\bs p}_i^{(a)}$ is close to $\widehat{\bs p}_i$ using the truncation estimator when $a$ is sufficiently large.
In practice, we set $a=10$, which empirically yields results close to the original truncation-based estimator. 

 We can estimate the asymptotic covariance of $\widehat{\bm p}_i^{(a)}$ as
$$\widehat{\text{Cov}}\left[\sqrt G\widehat{\bm p}_i^{(a)}\right] = \nabla\bs g( \widehat{\bs\beta}^{(a)}_i)\widehat{\bs \Gamma}\widehat{\bs\Omega}^{-1}\widehat{\bs\Sigma}_i\widehat{\bs\Omega}^{-1}\widehat{\bs \Gamma}\nabla\bs g(\widehat{\bs\beta}^{(a)}_i)^\top$$
to construct CI for each $p_{ik}$. As we will show in simulations and real data studies, CIs based on $\widehat{\bm p}_i^{(a)}$ tend to be shorter than those from the original truncation estimator when some proportion estimates are close to $0$, while still maintaining reasonable coverage.

\section{Statistical Inference across multiple individuals}\label{sec:infer_multi}

In many applications, estimating individual-level cell type proportions serves as an intermediate step. A common downstream goal is to assess how these proportions relate to covariates of interest, such as disease status, treatment assignment, age, or genetic factors \citep{fadista2014global}.

We model the true cell type proportions $\bm p_i$ using a generalized linear model (GLM):
\begin{equation} \label{eq:downstream_lm}
    \mathbb{E}\left(\bs p_{i}\right) = \bs h(\bs b_0 + {\bs A_0}^\top\bs f_i), \quad i = 1, 2, \cdots, N.
\end{equation}
where $\bm f_i \in \mathbb{R}^S$ represents individual-level covariates, $\bm b_0$ is an intercept vector, and $\bm A_0 \in \mathbb{R}^{S \times K}$ encodes how cell type proportions vary with these covariates. The function $\bs h: \mathbb{R}^K \rightarrow [0,1]^K$ is a known link function, such as the identity or softmax functions. Without loss of generality, we assume that $\bs f_i$ is already centered, satisfying $\sum_{i=1}^N \bm f_i = \bm 0$, so that the intercept and covariates are orthogonal.

\begin{remark}
    In contrast to Theorem~\ref{thm:clt}, 
which treats the cell type proportions $\bm p_i$ as fixed, the GLM framework assumes that they are random. These perspectives are reconciled by interpreting our earlier inference as conditional on $\bm p_i$. 
\end{remark}

We further assume that $\bs p_i$ are randomly drawn from the population:
\begin{assumption}[Independence]
\label{assmp:random_prop}
The vectors $\bs p_i$ are mutually independent and independent of the error $\bs \epsilon_i'$ and scaling matrix $\bs \Lambda$ in model~\eqref{bulk}. 
\end{assumption}

If the true proportions ${\bm p_1, \dots, \bm p_N}$ were observed, we may estimate $\bm A_0$ by solving the following estimating equation:
\begin{equation}
\label{eq:estimating-equation-for-A}
\bs L_N(\bs A, \bs b;\bs P):=\frac{1}{N}\sum_{i=1}^N\left\{\bs p_{i,1:(K-1)}-\bs h(\bs b+\bs A^\top \bs f_i)_{1:(K-1)}\right\}\tilde{\bs f}_i^\top=\bs 0,
\end{equation}
where $\tilde{\bs f}_i=\left(1,\bs f_i\right)$ and we drop the last entry of each composition vector since both $\bm p_i$ and $\bs h(\cdot)$ lie on the simplex. Let $(\bm A_N, \bm b_N)$ denote the solution of equation \eqref{eq:estimating-equation-for-A}.

Under standard regularity conditions, the estimator $\bm A_N$ is asymptotically normal as $N\to \infty$ (See Theorem~\ref{lemma:An-clt} for a formal proof):
 $$\sqrt N\vect{({\bs A}_N^\top - \bs A_0^\top)} \overset{d}{\to} \mathcal{N}\left(\bs 0,(\bs L_{\bs B_0}^{-1}\bs D \bs L_{\bs B_0}^{-\top})_{I_{\bs A}\times I_{\bs A}}\right),$$
 where
 \begin{equation}\label{eq:L_B0}
     \bs D=\lim_{N\to\infty}\frac{1}{N}\sum_{i=1}^N\tilde{\bs f}_i\tilde{\bs f}_i^\top\otimes \bs D_i, \quad \bs L_{\bs B_0}=\lim_{N\to\infty}\frac{1}{N}\sum_{i=1}^N\tilde{\bs f}_i\tilde{\bs f}_i^\top\otimes \dot{\bs h}(\bs b_0+\bs A_0^\top\bs f_i)_{1:(K-1)},
 \end{equation}
and $\bs D_i=\mathrm{Cov}(\bs p_{i, 1:(K-1)})$ and  $I_{\bs A}=\{2,\dots,S+1\}$ indexes the parameters in $\bm A$. 

In practice, the true proportions are unknown and we instead use the estimated proportions $\widehat{\bm p}_i$. A natural plug-in estimator $\widehat{\bm A}$ is obtained by solving:
\begin{equation} \label{eq:estimating-equation-for-A-p-hat}
\bs L_N(\bm A, \bm b;\widehat{\bm P}) = \bs 0.
\end{equation} 
To construct valid confidence intervals based on $\widehat{\bm A}$, it is crucial to understand the discrepancy $\widehat{\bm A} - \bm A_0$. Given that $\bm A_N$ is asymptotically normal, it suffices to assess whether the additional estimation error $\widehat{\bm A} - \bm A_N$ is asymptotically negligible or not. We analyze this error under the following conditions:
\begin{assumption}\label{assmp:group}
We assume the following conditions hold:
\begin{enumerate}[label=\alph*.]
\item For any $\bs A$ and $\bs b$, it holds that $\bs L_N(\bs A, \bs b;\bs P) \overset{p}{\to}\bs  L(\bs A,\bs b)$
and $(\bs A_0, \bs b_0)$ is its unique root such that $\bs L(\bs A_0, \bs b_0)=\bs 0$.
\item The equation $\bs L_N(\bs A, \bs b; \bs P) = \bs 0$ has a unique solution for any $\bs P$.  
\item $\lim_{N\to\infty}\frac{1}{N}\sum_{i=1}^N\bs f_i\bs f_i^\top\succ0$, and $\max_i \|\bs f_i\|_2 \leq C_1$ for some constant $C_1$.
\item The function $\dot{\bs h}(\cdot)_{1:(K-1)}$ is continuously differentiable. Additionally, $\bs L_{\bs B_0}$ as defined in \eqref{eq:L_B0} is invertible and its smallest singular value is larger than some constant $C_2$.
\item The scaling parameters $\gamma_i$ satisfy $C_3 \leq \gamma_i \leq C_4$ for constants $C_3, C_4 > 0$, and $\min_{i,k} p_{ik} \geq C_5$ for some $C_5 > 0$.
\end{enumerate}
\end{assumption}

Assumption~\ref{assmp:group}a-d are standard regularity conditions in GLM theory. The bounds on $\gamma_i$ in Assumption~\ref{assmp:group}e reflects typical quality control steps in RNA-seq analyses. The lower bound on $\min_{i,k}p_{ik}$ ensures that our non-negative refinement $\widehat{\bs\beta}_i^\star$ is asymptotically equivalent to $\widehat{\bs\beta}_i$, simplifying our inference procedure as in Theorem~\ref{thm:clt}.

Under these conditions, we establish sufficient conditions for when the plug-in estimator $\widehat{\bm A}$ achieves the same asymptotic distribution as $\bm A_N$:
\begin{theorem}\label{thm:two_group}
Under Assumptions~\ref{asp:homo_popu}-\ref{assmp:group}, the estimator $\widehat{\bs A}$ satisfies
$$\widehat {\bs A} - \bs A_N = o_p(\bs A_N-\bs A_0),~\text{as}~N\to\infty,$$
if either of the following holds:
\begin{enumerate}[label=(\roman*)]
    \item $N=o(G)$,
    \item $N=o(G^2)$ and the global null hypothesis $H_0: \bm A_0 = \bm 0$ holds, with $\bm p_i$ share the same mean and variance.
\end{enumerate}
\end{theorem}

Theorem~\ref{thm:two_group} demonstrates that the naive approach that ignores estimating error in $\widehat{\bs p}_i$ in downstream analyses remains valid under certain conditions. The condition $N = o(G)$ is typically satisfied in bulk RNA-seq applications, where $G$ often exceeds $10{,}000$ and $N$ is in the hundreds or fewer. However, when $N$ becomes comparable to $G$, as for large-scale transcriptomics data, or for spatial transcriptomics where each spot is a target individual, 
estimation errors in $\widehat{\bm p}_i$ may accumulate and affect inference.  The only scenario where naive confidence intervals remain valid even when $N$ is large is under the global null hypothesis $\bm A_0 = \bm 0$, where the cell type proportions do not associate with any features of the target individuals. For more general alternative hypotheses, failure to account for estimation uncertainty may lead to under-coverage of confidence intervals and inflated false positive rates.

\section{Practical considerations}\label{sec:practical}

\subsection{Choice of the weight matrix $ \mathbf W$}\label{sec:weight_matrix}
Due to the variability in gene expressions, assigning equal weights to all genes can lead to inefficient estimators. 
A common approach is to select marker genes based on differential expression in the reference data \citep{chen2018profiling}, under the intuition that such genes are more informative. However, from a linear regression perspective, removing genes (samples) does not provide efficiency gains unless noise levels differ across genes.    

A more effective strategy is to weight genes inversely by their noise variance. For instance, MuSiC \citep{wang2019bulk} estimates the variance $\sigma_g^2 = \Var{y_{ig} - \bs\mu_g^\top\bs \beta_i}$ and sets $w_g = 1/\widehat\sigma_g^2$.
MuSiC relies on residuals from the observed target data $y_{ig}$, which can severely bias our estimating equation~\eqref{est_eq}.

Instead, we estimate $\sigma_g^2$ using only the reference data. Because biological variation typically dominates technical noise in bulk RNA-seq, we approximate $\sigma_g^2$ using the reference-based variance matrix $\bs V_g$. A natural choice is $s^2_g = \bs 1^\top\widehat{\bs V}_g\bs 1$, representing average biological variability across cell types, and set $w_g = 1/s^2_g$.


To stabilize the weights, we assume $s_g^2 \overset{ind}{\sim}\sigma_g^2\chi^2_{d}/d$ with degrees of freedom $d = M - 1$ and apply the empirical Bayes method Vash \citep{lu2016variance},  which assumes a mixture of inverse-Gamma priors on $\sigma_g^2$. 
We then obtain shrinkage estimates $\tilde s_g^2$ towards the mean across genes. The final gene weights are then set to $w_g = 1/\tilde s_g^2$, reducing the effect of extreme values.

\subsection{Estimation of gene-gene dependence set $\mathcal A$}\label{sec:est_cor}

The gene-gene dependence set $\mathcal{A}$ is generally unknown. 
Since reference data often lack sufficient samples to estimate cell-type-specific gene correlations, we infer  $\mathcal{A}$ based on the sample covariance matrix of the target data.


Theoretically, under Assumption~\ref{assmp:dependence},  non-zero entries in the gene-gene covariance matrix can be identified via thresholding \citep{cai2011adaptive}, but this requires a large number of samples $N$. In practice, when $N\ll G$, the thresholding methods can be ineffective, and we instead adopt a multiple testing approach proposed by \citet{cai2016large}.

Specifically, 
let $\bs R= (\rho_{g_1g_2})_{G \times G}$ denote the gene-gene correlation matrix. We test $H_{0,g_1g_2}: \rho_{g_1g_2} = 0$ for all gene pairs using test statistics 
$$T_{g_1g_2} = \left(\sum_{i=1}^N (y_{ig_1} - \bar y_{g_1})(y_{ig_2} - \bar y_{g_2})\right)/\sqrt{N\hat\theta_{g_1g_2}},$$
where $\bar y_{g}$ is the sample mean of gene $g$, and $$\hat\theta_{g_1g_2} = \sum_{i=1}^N\left((y_{ig_1} - \bar y_{g_1})(y_{ig_2} - \bar y_{g_2})-\hat\sigma^2_{g_1g_2}\right)^2/N$$ with $\hat\sigma^2_{g_1g_2}$ being the sample covariance between genes $g_1$ and $g_2$. 
We reject $H_{0,g_1g_2}$ if $|T_{g_1g_2}|\geq \hat t$ where
\[\hat t = \inf\left\{0\leq t\leq b_G:\frac{(2-2\Phi(t))(G^2-G)/2}{\max\left(\sum_{1\leq g_1<g_2\leq G}I(|T_{g_1g_2}|\geq t),1\right)}\leq\alpha\right\},\]
with $b_G = \sqrt{4\log G - 2\log(\log G)}$, and we set $\alpha = 0.5$ to guarantee enough power and $\Phi(\cdot)$ denoting the standard normal cumulative distribution function.

If $N$ is too small, we optionally use public bulk RNA-seq datasets (e.g., GTEx \citep{lonsdale2013genotype}) to identify the top gene-gene pairs with non-zero correlations.


\subsection{Finite-sample correction}\label{sec:finite_correct}

The plug-in sandwich estimator $\widehat{\bs \Sigma}_i$ tends to underestimate the variance of $\widehat{\bs p}_i$, especially when $\widehat{\bs \beta}_i^\star$ is substituted for the true $\bs \beta_i$. 
Despite a large  $G$, high gene-gene correlations greatly reduce the effective sample size, making finite-sample corrections necessary \citep{long2000using}.  

A standard correction is the HC3 method \citep{mackinnon1985some}, which computes jackknife estimators  $\widehat {\bs \beta}_{ig}^\star$ by omitting gene $g$ and estimating:
\begin{equation}\label{eq:CV}
    \widehat{\bs \Sigma}_i^\star = \frac{1}{G}\left(\sum_{g = 1}^G \bs\phi_g(\widehat {\bs\beta}_{ig}^\star)\bs\phi_g(\widehat{\bs{\beta}}_{ig}^\star)^\top + \sum_{(g_1,g_2) \in \mathcal A}\bs\phi_{g_1}(\widehat {\bs\beta}_{ig_1}^\star)\bs\phi_{g_2}(\widehat {\bs\beta}_{ig_2}^\star)^\top\right).
\end{equation}
where $\bs \phi_g(\bs \beta_i)$ is defined in \eqref{eq:phi_g}.
However, this HC3 correction will be ineffective when genes are correlated.


To address this, we introduce a clustering-based $C$-fold cross-validation approach (default $C=10$). 
We first apply k-medoids clustering \citep{rdusseeun1987clustering} to a dissimilarity matrix  $\bs 1\bs 1^\top - \bs A$, where $\bs A \in \{0, 1\}^{G\times G}$  indicates gene-gene dependencies (membership in $\mathcal{A}$). The clustering ensures highly correlated genes are grouped into the same fold.

For each fold  $s$, we exclude it and estimate $\widehat{\bs \beta}_{is}^\star$ from the remaining data, further excluding any genes correlated with those in fold $s$ based on $\mathcal{A}$.
We then compute $\widehat{\bs \Sigma}_i^\star$ as in equation~\eqref{eq:CV}, using $\widehat{\bs \beta}_{is}^\star$ in place of $\widehat{\bs \beta}_{ig}^\star$ to for genes in fold $s$.

Finally, Figure~\ref{fig:flowchart} summarizes the MEAD pipeline. While the full procedure is complex and not fully tractable analytically, its components are well motivated and practically sound. In particular, while we use weights estimated from the reference data, they are still independent from the target data, thus should not severely impact the asymptotic validity of our estimating equation~\eqref{est_eq}, as will be shown in our simulations and real data studies.

\begin{figure}[ht!]
    \centering
    \resizebox{15cm}{!}{
\begin{tikzpicture}

\node (ref) [process, text width=9cm, yshift=0cm] {\textbf{Reference data preparation} 
\vspace{0.2cm}
\begin{enumerate}
    \item[A.] Obtain cell-type specific gene expression reference matrix, $\widehat{\boldsymbol{U}}$, and the sample covariance $\widehat{\boldsymbol{V}}_g$ for each gene, as described in Section~\ref{sec:estimation_U};
    \item[B.] Compute the weight matrix, $\boldsymbol{W}$, following the procedure in Section~\ref{sec:weight_matrix}.
\end{enumerate}
};
\node (bulk) [process, right of=ref, text width=8cm, xshift=8cm, yshift = -0.25cm] {\textbf{Target data processing} \\
\vspace{0.2cm}
    Perform multiple testing on target samples (or from other sources) to identify correlated gene pairs, $\mathcal{A}$, as described in Section~\ref{sec:est_cor}.};

\node (estimation) [process, below of=ref, text width=8cm, xshift=4.5cm, yshift = -3cm] {\textbf{Estimation and inference} \\
\vspace{0.2cm}
    For each target individual $i$, estimate $\bs p_i$ and their corresponding standard errors, and apply finite sample correction (Section~\ref{sec:finite_correct}).};

\node (downstream) [process, below of=estimation, text width=9cm, yshift=-2cm] {\textbf{Downstream tasks} \\
\vspace{0.2cm}
    Compare the change of cell types across multiple individuals, as described in Section~\ref{sec:infer_multi}.}; 

\draw [arrow] (ref.south) -- (estimation.north);
\draw [arrow] (bulk.south) -- (estimation.north);
\draw [arrow] (estimation.south) -- (downstream.north);

\end{tikzpicture}
}
    \caption{Flowchart illustrating the overall procedure of MEAD.}
    \label{fig:flowchart}
\end{figure}

 \section{Simulations}\label{sec:simu1}

We benchmark the performance of MEAD against ordinary least squares (OLS), CIBERSORT \citep{newman2015robust}, and MuSiC \citep{wang2019bulk}. We also evaluate the effect of weighting, estimating $\mathcal{A}$, finite-sample correction, and Softplus transformation on MEAD's accuracy and inference quality.

We generate synthetic data using parameters estimated from the scRNA-seq dataset of \citet{xin2016rna}. After excluding individuals with missing cell types, we use the remaining 14 individuals to compute average gene expressions per cell type and individual. The population mean $\bm U$ is set as the average of these means, and the gene-wise cell-type covariance $\bs V_g =\text{diag}(\sigma_{g1}^2, \cdots, \sigma_{gK}^2)$ is set based on empirical variances. To obtain cleaner simulation data, 
we retain 9,496 genes after filtering out the top $5\%$ lowly and overly expressed genes, following \cite{li2016comprehensive}. This helps reduce noise from extremely lowly expressed genes and avoid dominance by highly expressed genes. 
To introduce gene-gene dependence, we simulate each individual’s true expression $\bs X_j$ from a multivariate log-normal distribution. For each cell type $k$, $\bs x_{jk}$ satisfies that $\EE{\bs x_{jk}}=\bs \mu_i$, $\Var{x_{jgk}}=\sigma_{gk}^2$ and $\Corr{\log{\bs x_{jk}}}=\bs R$, where $\bs R$ is a banded correlation matrix
with entries $\rho_{g_1g_2} = \max\left(1- \frac{|g_1-g_2|}{d}, 0\right)$ and bandwidth $d= 500$. 

For reference individuals, we generate $50$ cells per cell type, simulate single-cell counts from a Negative Binomial distribution $\text{NegBinomial}\left(\mu_{igk} = x_{jgk},\theta = 5\right)$, and compute pseudo-bulk expressions by averaging across cells as each $\rr z_{jgk}$. Each simulated cell exhibits roughly $50\%$ dropout, reflecting the sparsity in scRNA-seq data. 
For target individuals, the observed counts $y_{ig}$ are generated from a Negative Binomial distribution: 
\begin{equation}\label{eq:simu_bulk}
    y_{ig}\sim \text{NegBinomial}\left(\mu_{ig} = s_i \frac{\sum_k \lambda_g x_{igk} p_{ik}}{\sum_{g'} \lambda_{g'}\sum_k x_{ig'k} p_{ik}},\theta = 10\right),
\end{equation}
with $s_i = 500\times G$, the Gamma-distributed cross-platform scaling ratios $\lambda_g \overset{i.i.d.}{\sim} \mathrm{Gamma}(1/0.3, 1/0.3)$, satisfying $\E(\lambda_g) = 1$ and $\text{Var}[\lambda_{g}] = 0.3$.
The cell type proportions $\bs{p}_i$ are drawn from a Dirichlet distribution with parameters $\bs{\alpha} = a\bs{p}_0$, where $\bs{p}_0 = (0.5, 0.3, 0.1, 0.1)$  and $a = 10$. We set $N = 50$ target and $M = 10$ reference individuals, repeating each simulation $100$ times. 

We compare MEAD using three different weight choices: (i) equal weights ($w_g = 1$), (ii) marker gene weights based on \citet{newman2019determining} (same as CIBERSORT), and (iii) the proposed weighting in Section~\ref{sec:weight_matrix}.  Table~\ref{table:rmse_simu} reports the root mean square errors (RMSE) for each method. 
MuSiC yields the lowest RMSE, but MEAD with similar weighting performs comparably. Marker gene weighting provides little improvement for either MEAD or CIBERSORT. In contrast, MEAD’s proposed weighting substantially reduces RMSE from $0.096$ to $0.073$.

We also evaluate CI coverage using MEAD with equal weights and with proposed weighting, and compare it to OLS with HC3 correction. (Table \ref{table:coverage_equal}). Without accounting for gene-gene correlations, both MEAD and OLS show severe under-coverage. We incorporate gene-gene correlation adjustment and finite-sample correction in Section~\ref{sec:finite_correct}, estimating the dependency set $\mathcal{A}$ using 100 additional synthetic target samples. The coverage improves substantially with gene-gene correlation and finite-sample corrections. Figure~\ref{fig:ci_plot_simulation} shows CIs for each cell type and target individual in one simulation, comparing MEAD with and without Softplus transformation. The Softplus transformation slightly reduces coverage, as the CIs become slightly shorter when the proportion estimates are near $0$.   

\begin{figure}[ht!]
    \centering
    \includegraphics[width=\textwidth]{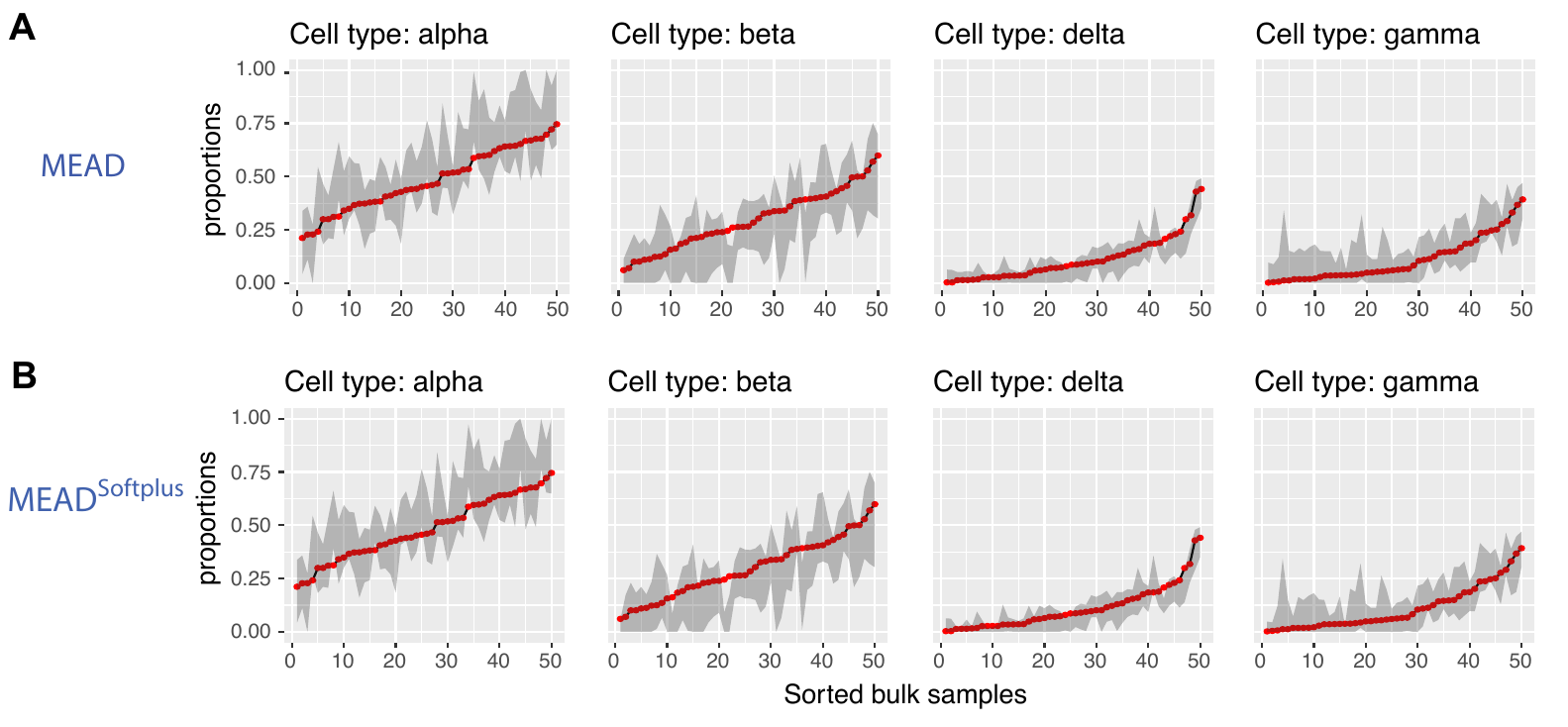}
    \caption{95\% CIs (gray shared areas) for simulated target individuals from one randomly selected simulation replicate using A) MEAD and B) MEAD with Softplus transformation. For each cell type, target individuals are sorted using their true cell-type proportions (red dots) in ascending order. }
    \label{fig:ci_plot_simulation}
\end{figure}





\section{Real data}\label{sec:simu2}

\subsection{Cross-platform pseudo-bulk data deconvolution}\label{sec:additional_analysis}

We perform deconvolution across two different sequencing platforms for human pancreatic islets, following \cite{wang2019bulk}. Pseudo-bulk target samples are constructed by averaging cell-specific gene expression from \cite{xin2016rna}, where the data is generated from the Fluidigm C1 platform (18 individuals: 12 healthy and 6 with Type 2 diabetes (T2D)). The reference dataset, from \cite{segerstolpe2016single}, includes scRNA-seq data from 10 individuals (6 healthy, 4 T2D) using the Smart-seq2 protocol. Consistent with \cite{wang2019bulk}, we use only the 6 healthy individuals as the reference.

After preprocessing, we retain 17,858 genes and 4 cell types (alpha, beta, delta, gamma). We compared MEAD to CIBERSORT, MuSiC and NNLS (Non-negative Least Squares, implemented in the MuSiC package) for estimation accuracy.  
As shown in Table~\ref{tab:platform_metrics} and Figure~\ref{fig:heatmap_platform}a, MEAD exhibits slightly worse point estimation accuracy than MuSiC for individual-level cell type proportions. However, it provides more accurate estimates of the mean differences in cell type proportions between healthy and T2D groups, whereas MuSiC tends to overestimate the differences (Table~\ref{table:cross_platform_betweengroups}). Figure~\ref{fig:heatmap_platform}bc illustrates the $95\%$ confidence intervals for each cell type. Among all $18\times 4 = 72$ cell type proportions, MEAD achieves $92\%$ coverage and the Softplus version achieves $87\%$, with substantially shorter intervals for proportions near $0$.

\begin{figure}[ht!]
    \centering
    \includegraphics[width=\textwidth]{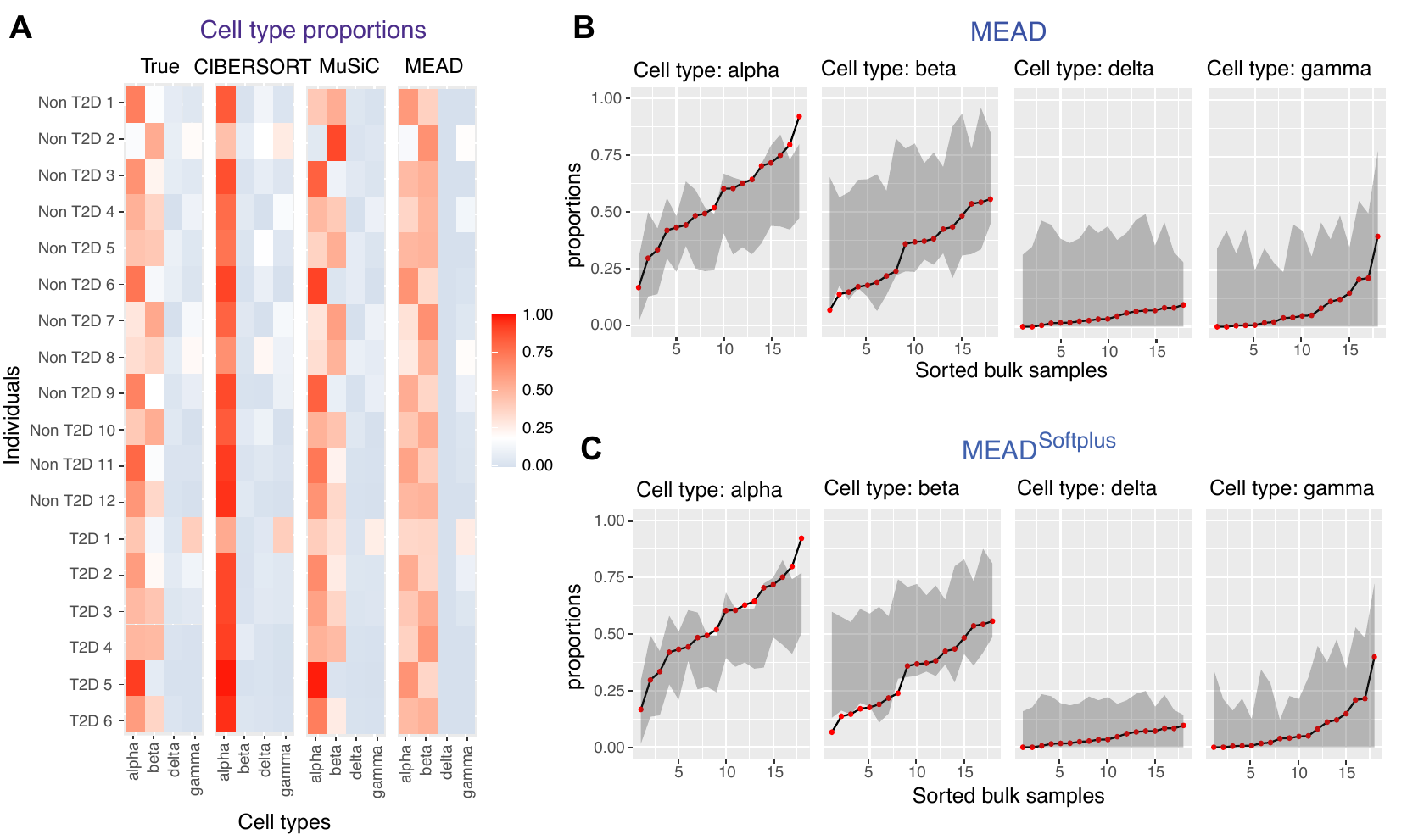}
    \caption{Comparison of estimation and inference results in the pancreas case study. A) Heatmap of estimated cell type proportions for each target individual. B) 95\% CIs (gray shared areas) for each cell type and individual using MEAD. C) 95\% CIs using MEAD with Softplus transformation. For each cell type, target individuals are sorted by their true cell-type proportions (red dots) in ascending order.}
    \label{fig:heatmap_platform}
\end{figure}


\subsection{Compare cell type proportion changes across multiple individuals}

Next, we benchmark MEAD using a population-scale scRNA-seq dataset from \citet{jerber2021population}, which profiles neuron development across 175 individuals and over 250,000 cells. 
While the original study identified seven cell types, we merged two unknown neuron subtypes (one present in $\geq 10$ cells in only 16 individuals), resulting in six cell types for analysis.  

\begin{figure}[ht!]
    \centering
    \includegraphics[width=0.9\textwidth]{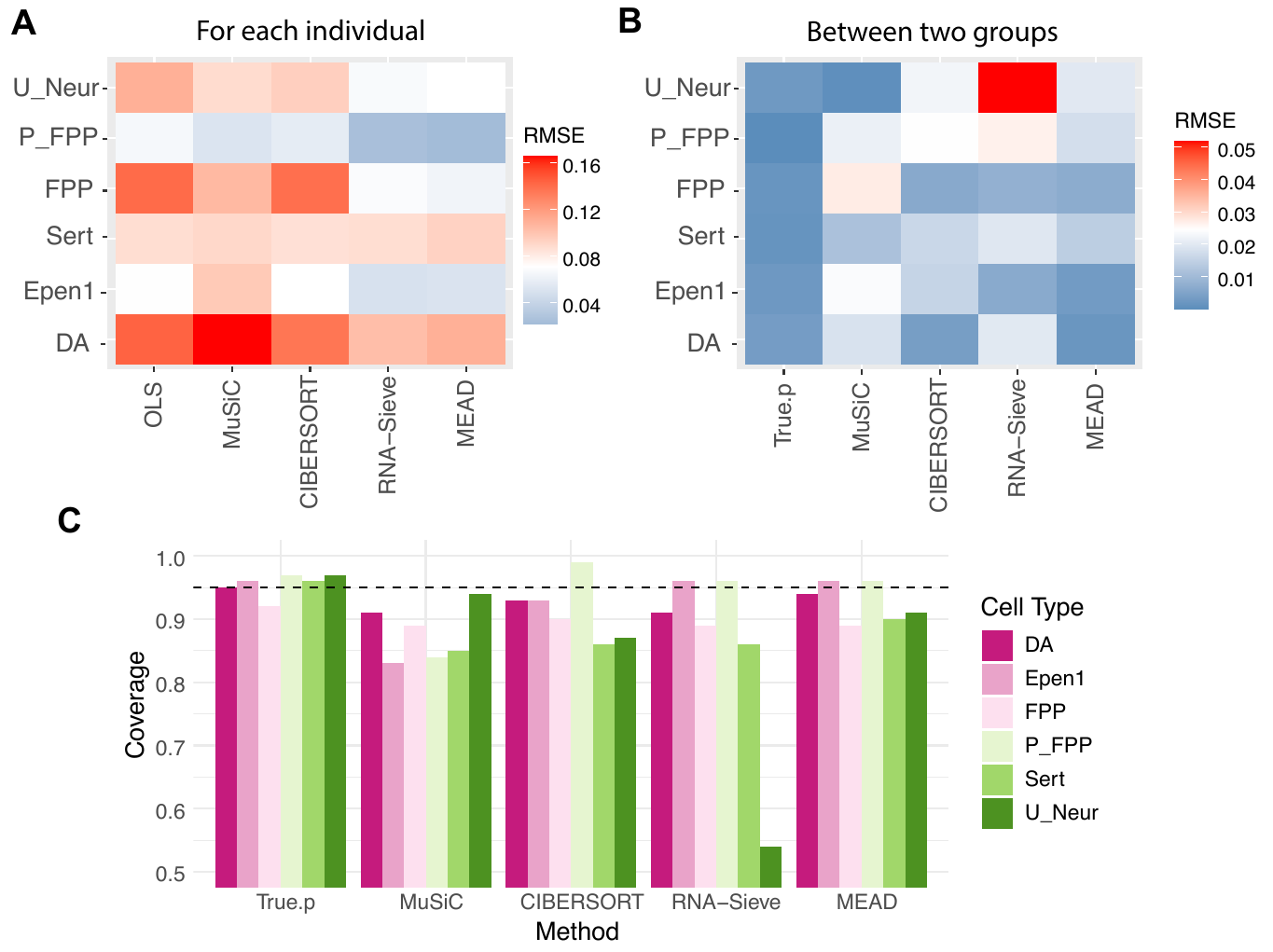}
    \caption{Comparison of RMSE and coverage across methods for the neuron development dataset with $a = 5$. A) RMSE for estimating individual-level cell type proportions under the global null. B) RMSE for estimating differences in average cell type proportions between two groups under the alternative. C) Empirical coverage of $95\%$ CIs for the between-group differences in cell type proportions under the alternative.}
    \label{fig:rmse}
\end{figure}

Each experiment is repeated $B = 100$ times.
In each round, $11$ individuals are randomly selected as reference, while the remaining $86$ individuals are used to generate $N$ target samples (with replacement if $N > 86$). The targets are split into two equal groups. Under the global null, both groups share the same mean proportions $\bs p_1=\bs p_2=(0.3,0.2,0.15,0.15,0.1,0.1)$. Under the alternative, $\bs p_1 = (0.15,0.15,0.1,0.1,0.2,0.3)$ and $\bs p_2 = (0.1,0.1,0.2,0.3,0.15,0.15)$.
The true cell type proportions for each target individual $i$ sampled from a Dirichlet distribution with concentration parameter $=\bs \alpha = a\bs p_1$ for group 1 and $\bs \alpha = a\bs p_2$ for group 2, where the scaling parameter $a=5$ (high variation) or $20$ (low variation) controls across-individual heterogeneity. Target data are generated as pseudo-bulk samples from scRNA-seq using the generated proportions, with added Poisson noise and cross-platform scaling factors $\lambda_g \overset{i.i.d.}{\sim} \mathcal{N}(1, 0.1)$. The gene-gene dependency set $\mathcal{A}$ is estimated using pseudo-bulk data from $97$ individuals.

We compare MEAD to OLS, MuSiC, CIBERSORT, and RNA-Sieve \citep{erdmann2021likelihood}, which also provides  CIs for the proportions.
Figure \ref{fig:rmse}a shows the RMSE for estimated individual-level proportions under the global null, while Figure~\ref{fig:rmse}b shows RMSE for group mean differences under the alternative (similar results for $a=20$ in Figure~\ref{fig:rmse20}), when $N = 86$.
 MEAD consistently achieves the lowest RMSE. 
Table~\ref{table:coverage_real2} and Table~\ref{table:coverage_real2_appendix} report $95\%$ CI coverage for individual-level proportions. MEAD achieves near-nominal coverage, outperforming RNA-Sieve, which does not account for inter-individual variations or gene-gene correlations. Figure~\ref{fig:neuron_ci} illustrates CI widths from one random round of the experiment.

We also evaluate coverage of mean proportion differences using a naive two-sample t-test. Figure~\ref{fig:rmse}c and Figure~\ref{fig:twogroup_coverage_bycelltype20} show CI coverage by cell type, including an oracle that uses true proportions. MEAD performs comparably to the oracle, while other methods perform well when
$a=5$, while their performance degrade when $a=20$, where true cell type proportions have lower cross-individual heterogeneity.


Finally, Table~\ref{table:more_two_group} examines CI coverage as $N$ increases from 86 to 1000. Under the global null, where there is no difference between groups, all methods maintain good coverage. Under the alternative, coverage declines with increasing $N$, consistent with Theorem~\ref{thm:two_group}. 

\section{Discussion}\label{sec:discuss}

In this paper, we introduced MEAD, a method for estimating and inferring cell type proportions that, unlike many existing approaches, provides asymptotically valid confidence intervals for both individual proportions and between-group comparisons. We also revisited the common practice of marker gene selection. Both our theoretical analysis and simulations suggest that restricting to marker genes does not necessarily improve estimation, which aligns with findings in recent empirical studies \citep{wang2019bulk,tsoucas2019accurate,cobos2020benchmarking}.

Our framework assumes target and reference individuals are drawn from the same population, but this can be relaxed by conditioning on covariates (e.g., disease status). This allows consistent estimation of cell-type-specific gene expression within subpopulations for use in deconvolution. Such generalizations preserve the validity of our inference procedure.

MEAD accommodates gene-specific cross-platform scaling factors and establishes general identifiability conditions; however, our inference procedure relies on a simplifying assumption that the scaling ratios $\lambda_g$ are non-informative, i.e., independently distributed and uncorrelated with population true gene expression levels. This assumption enables tractable estimation and asymptotically valid inference, but may not hold when the scaling ratios are correlated with mean expressions. Addressing this limitation represents a challenging direction for future research.

While our approach is frequentist, a hierarchical Bayesian model could offer an alternative by jointly modeling proportions, gene-specific biases, and downstream regression analyses. However, fully specifying the high-dimensional, non-Gaussian, and correlated gene expression distribution is difficult in practice. Existing Bayesian methods for deconvolution, such as RCTD \citep{cable2022robust}, ST-assign \citep{geras2023joint}, and BayesTME \citep{zhang2023bayestme}, typically ignore gene–gene dependence, which is acceptable for point estimation but inadequate for uncertainty quantification. Moreover, deconvolution and downstream analysis are often conducted as separate steps in practice, with the estimated cell type proportions used for various exploratory and inferential purposes beyond regression \citep{gaspard2025cell}. A joint Bayesian model may not align with this flexible workflow. Nonetheless, future work could explore a well-designed Bayesian framework to relax key assumptions in our method, such as the non-informative scaling ratio assumption, while preserving valid inference.


 Finally, we discuss the effects of missing or over-partitioned cell types. If a cell type is missing in the reference data, only a projection of its contribution can be estimated, with bias depending on its abundance and similarity to other types. Decomposing a cell type into subtypes may reduce within-type gene–gene correlation, potentially improving estimation, but may also increase similarity between cell types, making them harder to distinguish. The trade-off between these effects is beyond the scope of this paper.

\newpage

\begin{table}[h]
\centering
\scalebox{0.9}{
\begin{tabular}{@{}lccccccc@{}}
\hline
     & \multirow{2}{*}{OLS}  & \multirow{2}{*}{\makecell{MEAD \\(equal weights)}} & \multirow{2}{*}{\makecell{MEAD \\(marker)}} & \multirow{2}{*}{MEAD} & \multirow{2}{*}{MEAD\textsuperscript{Softplus}} & \multirow{2}{*}{MuSiC} & \multirow{2}{*}{CIBERSORT} \\ 
 &&&&&&&\\
 \hline
RMSE              &   0.080   
         &   0.093    
   &   0.089   
  &   0.073    &   0.073 
   &   0.059    & 0.094 \\ \hline
\end{tabular}
}
\caption{RMSE comparisons on simulated data.}
\label{table:rmse_simu}
\end{table}

\begin{table}[h]
\centering
\scalebox{0.9}{
\begin{tabular}{@{}lclll@{}}
\hline
 \textbf{Method} & \textbf{Correlation Considered} & \multicolumn{2}{c}{\textbf{Coverage}}     \\ \hline
OLS             & No & \multicolumn{2}{c}{0.45(0.065)}  \\ 
MEAD (equal weights)         & No & \multicolumn{2}{c}{0.54(0.085)}  \\ 
MEAD   & No & \multicolumn{2}{c}{0.58(0.130)}  \\ 
       &            & \multicolumn{1}{l}{\textbf{Known Cor}} & \textbf{Estimated Cor} \\
MEAD& Yes& \multicolumn{1}{l}{0.90(0.032) } & 0.92(0.051)   \\ 
MEAD\textsuperscript{Softplus} & Yes&\multicolumn{1}{l}{0.88(0.032) } & 0.91(0.049)   \\ 
MEAD+cv & Yes & \multicolumn{1}{l}{0.95(0.022)} & 0.95(0.035) \\ 
MEAD\textsuperscript{Softplus}+cv & Yes & \multicolumn{1}{l}{0.93(0.026)} & 0.94(0.034)\\ \hline
\end{tabular}}
\caption{CI Coverage in simulation. The Coverage reported is averaged over all target individuals. We report both the mean coverage over repeated simulations and its standard deviation in parentheses.}
\label{table:coverage_equal}
\end{table}

\begin{table}[ht!]
    \centering
    \scalebox{0.9}{
    \begin{tabular}{@{}lcccc@{}}
        \hline
         & MEAD  & MuSiC & NNLS & CIBERSORT \\
        \hline
     Individual   RMSE   & 0.113 & 0.099 & 0.172 & 0.246 \\
       Group mean difference RMSE   & 0.0246 &  0.0370 & 0.0254 & 0.0299 \\
        \hline
    \end{tabular}
    }
    \caption{RMSE comparisons in the cross-platform deconvolution study. The top row shows the RMSE for estimating individual-level cell type proportions, while the bottom row shows the RMSE for estimating the mean difference in cell type proportions between the $12$ healthy and $6$ T2D individuals.}
    \label{tab:platform_metrics}
\end{table}

\begin{table}[ht!]
\centering
\scalebox{0.9}{
\begin{tabular}{@{}lcccccc@{}}
\hline
     & DA   & Epen1 & Sert  &  FPP &P FPP& U$\_$Neur\\ \hline
MEAD  &0.93 & 0.89 &0.88 &0.91 & 0.94 &  0.89 \\ 

RNA-Sieve   &0.11  &0.18 &0.12 &0.16  &0.23   &0.13 \\ \hline

\end{tabular}}
\caption{Mean coverage of $95\%$ CIs for each individual's proportions and each cell type under the global null with $a=5$.}
\label{table:coverage_real2}
\end{table}

\begin{table}[ht!]
\centering
\scalebox{0.9}{
\begin{tabular}{@{}lcccccc@{}}
\hline

\multicolumn{1}{l}{\multirow{2}{*}{}} & \multicolumn{3}{c}{Equal group means}                                                     & \multicolumn{3}{c}{Different group means}                            \\ 
\multicolumn{1}{l}{}                  & \multicolumn{1}{l}{$N = 86$} & \multicolumn{1}{l}{$N = 500$} & \multicolumn{1}{l}{$N = 1000$} & \multicolumn{1}{l}{$N = 86$} & \multicolumn{1}{l}{$N = 500$} & $N = 1000$ \\ \hline
\multicolumn{1}{@{}l}{True p}            & \multicolumn{1}{c}{0.940}   & \multicolumn{1}{c}{0.958}    & \multicolumn{1}{c}{0.943}  & \multicolumn{1}{c}{0.955}   & \multicolumn{1}{c}{0.972}    &  0.957 \\ 
\multicolumn{1}{@{}l}{MuSiC}             & \multicolumn{1}{c}{0.963}   & \multicolumn{1}{c}{0.961}    & \multicolumn{1}{c}{0.938}  & \multicolumn{1}{c}{0.877}   & \multicolumn{1}{c}{0.600}    &  0.468 \\ 
\multicolumn{1}{@{}l}{CIBERSORT}         & \multicolumn{1}{c}{0.955}   & \multicolumn{1}{c}{0.961}    & \multicolumn{1}{c}{0.945}  & \multicolumn{1}{c}{0.913}   & \multicolumn{1}{c}{0.711}    &   0.610\\ 
\multicolumn{1}{@{}l}{RNA-Sieve}          & \multicolumn{1}{c}{0.958}   & \multicolumn{1}{c}{--}    & \multicolumn{1}{c}{--}     & \multicolumn{1}{c}{0.853}   & \multicolumn{1}{c}{--}       & --     \\ 
\multicolumn{1}{@{}l}{MEAD}           & \multicolumn{1}{c}{0.968}   & \multicolumn{1}{c}{0.957}    & \multicolumn{1}{c}{0.937}  & \multicolumn{1}{c}{0.927}   & \multicolumn{1}{c}{0.861}    &  0.810 \\ \hline
\end{tabular}}
\caption{Mean coverage of the $95\%$ CIs of the group difference with growing $N$ when $a = 5$. RNA-Sieve is not performed for larger samples due to its high computational cost.}
\label{table:more_two_group}
\end{table}

\newpage

\section*{Data and Code Availability}
All data used are publicly available. The scRNA-seq pancreas dataset from \cite{xin2016rna} is available at \url{https://www.ncbi.nlm.nih.gov/geo/query/acc.cgi?acc=GSE81608}. The pancreas dataset from the \cite{segerstolpe2016single} is available at \url{https://www.ebi.ac.uk/biostudies/arrayexpress/studies/E-MTAB-5061}.  The scRNA-seq data from \cite{jerber2021population} is available at \url{https://zenodo.org/record/4333872}.

The code for reproducing results in this paper is accessible at \url{https://github.com/DongyueXie/MEAD-paper}, and the R package is available at \url{https://github.com/DongyueXie/MEAD}.



\section*{Acknowledgments and Funding}
This work was supported by the National Science Foundation under grants DMS-2113646 and DMS-2238656.

\vspace{0.2in}

\begin{small}
\begingroup
\setstretch{1}

\bibliography{manuscript}
\endgroup
\end{small}

\newpage

\newpage

\renewcommand{\thesection}{S\arabic{section}}   
\renewcommand{\thetable}{S\arabic{table}}   
\renewcommand{\thefigure}{S\arabic{figure}}

\setcounter{section}{0}
\setcounter{figure}{0}
\setcounter{table}{0}

\renewcommand{\thetheorem}{S\arabic{theorem}}
\renewcommand{\theassumption}{S\arabic{assumption}}
\renewcommand{\thecorollary}{S\arabic{corollary}}
\renewcommand{\theremark}{S\arabic{remark}}
\renewcommand{\thelemma}{S\arabic{lemma}}
\renewcommand{\theequation}{S\arabic{equation}}

\setcounter{equation}{0}
\setcounter{theorem}{0}
\setcounter{assumption}{0}
\setcounter{corollary}{0}
\setcounter{remark}{0}
\setcounter{lemma}{0}

\renewcommand {\thepage} {S\arabic{page}}
\setcounter{page}{1}


\newpage

\begin{center}
    \large\textbf{\MakeUppercase{Supplementary Materials for ``Statistical Inference for Cell Type Deconvolution''}}
\end{center}

\section{Supplementary tables and figures}

\begin{table}[ht!]
\centering
\begin{tabular}{@{}lcccccr@{}}
\hline
     & DA   & Epen1 & Sert  &  FPP &P FPP& U$\_$Neur\\ \hline
MEAD  &0.92 & 0.88 &0.85 &0.89 & 0.92 &  0.86 \\ 

RNA-Sieve   &0.06  &0.14 &0.07 &0.11  &0.24   &0.09 \\ \hline

\end{tabular}
\caption{Mean coverage of $95\%$ CIs for each individual's proportions and each cell type under the global null with $a=20$.}
\label{table:coverage_real2_appendix}
\end{table}

\begin{table}[h!]
\centering
\begin{tabular}{@{}lccccc@{}}
\hline
 & True & MEAD & MuSiC & NNLS & CIBERSORT \\
\hline
Alpha   & -0.0550 & -0.0221 & -0.1084 & -0.0705 & -0.0581 \\
Beta    &  0.0565 &  0.0285 &  0.1075 &  0.0266 &  0.0103 \\
Delta   &  0.0189 &  0.0000 &  0.0170 &  0.0273 &  0.0542 \\
Gamma   & -0.0204 & -0.0064 & -0.0161 &  0.0167 & -0.0064 \\
\hline
\end{tabular}
\caption{Mean proportion difference between the healthy and T2D individuals for each cell type.}
\label{table:cross_platform_betweengroups}
\end{table}

\newpage 
\begin{figure}[ht!]
    \centering
    \includegraphics[scale=0.6]{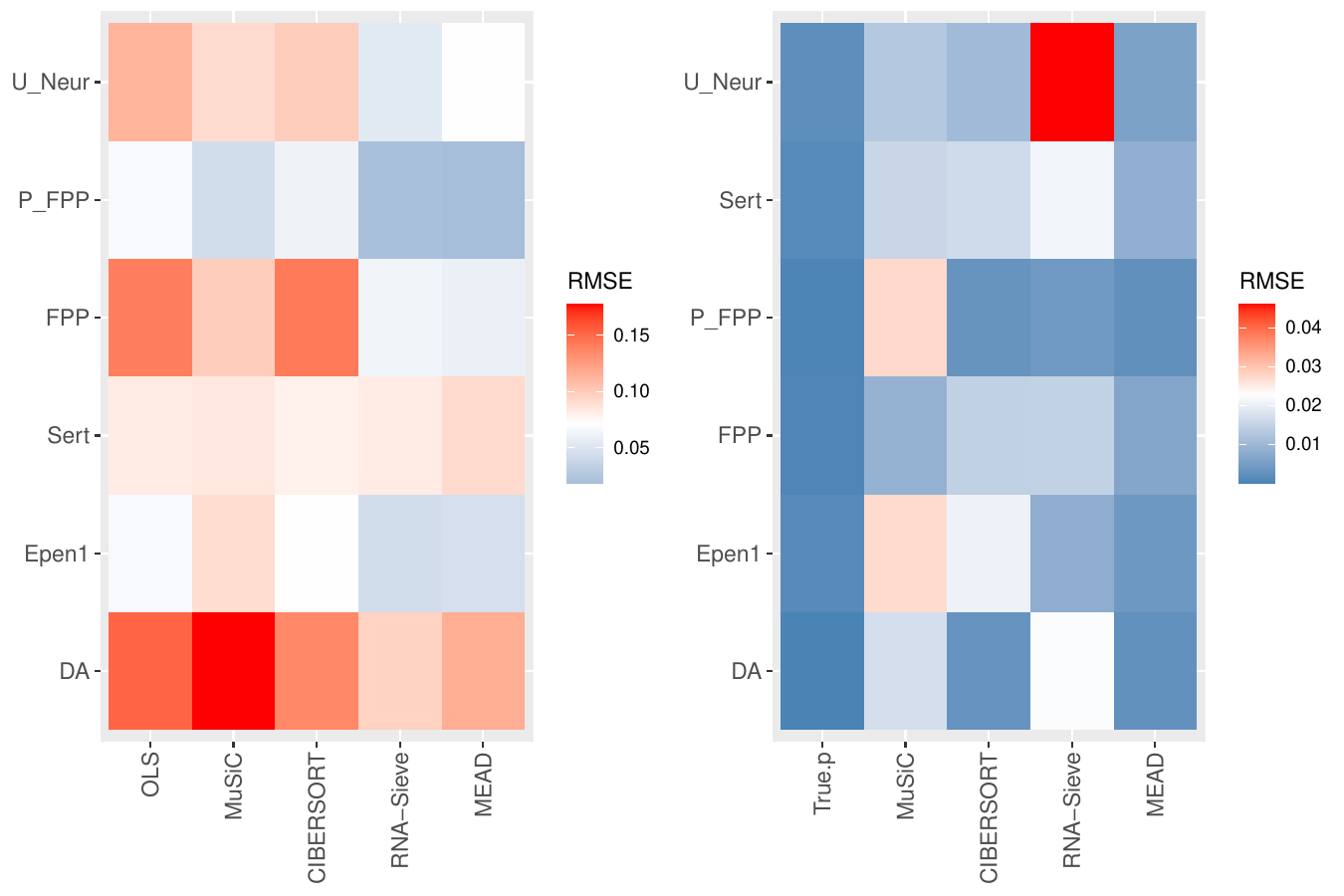}
    \caption{Comparison of RMSE across methods for the neuron development dataset with $a = 20$. Left: RMSE for estimating individual-level cell type proportions under the global null. Right: RMSE for estimating differences in average cell type proportions between two groups under the alternative.}
    \label{fig:rmse20}
\end{figure}

\begin{figure}[ht!]
    \centering
    \includegraphics[scale=0.8]{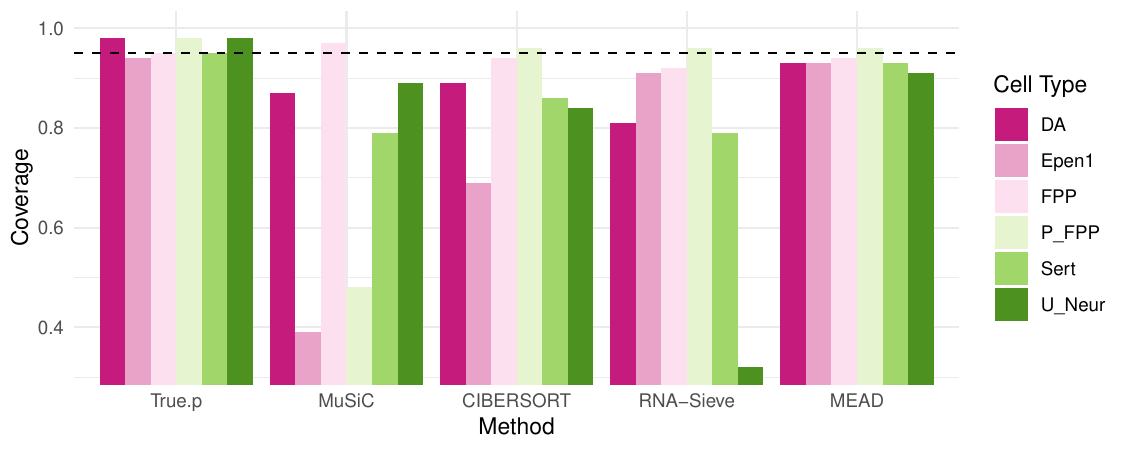}
    \caption{Comparison of the $95\%$ CIs for each individual's cell type proportions for the two group difference by cell types under the alternative for the neuron development dataset with $a = 20$. }
    \label{fig:twogroup_coverage_bycelltype20}
\end{figure}
\newpage

\begin{figure}[ht!]
    \centering
    \includegraphics[scale=0.6]{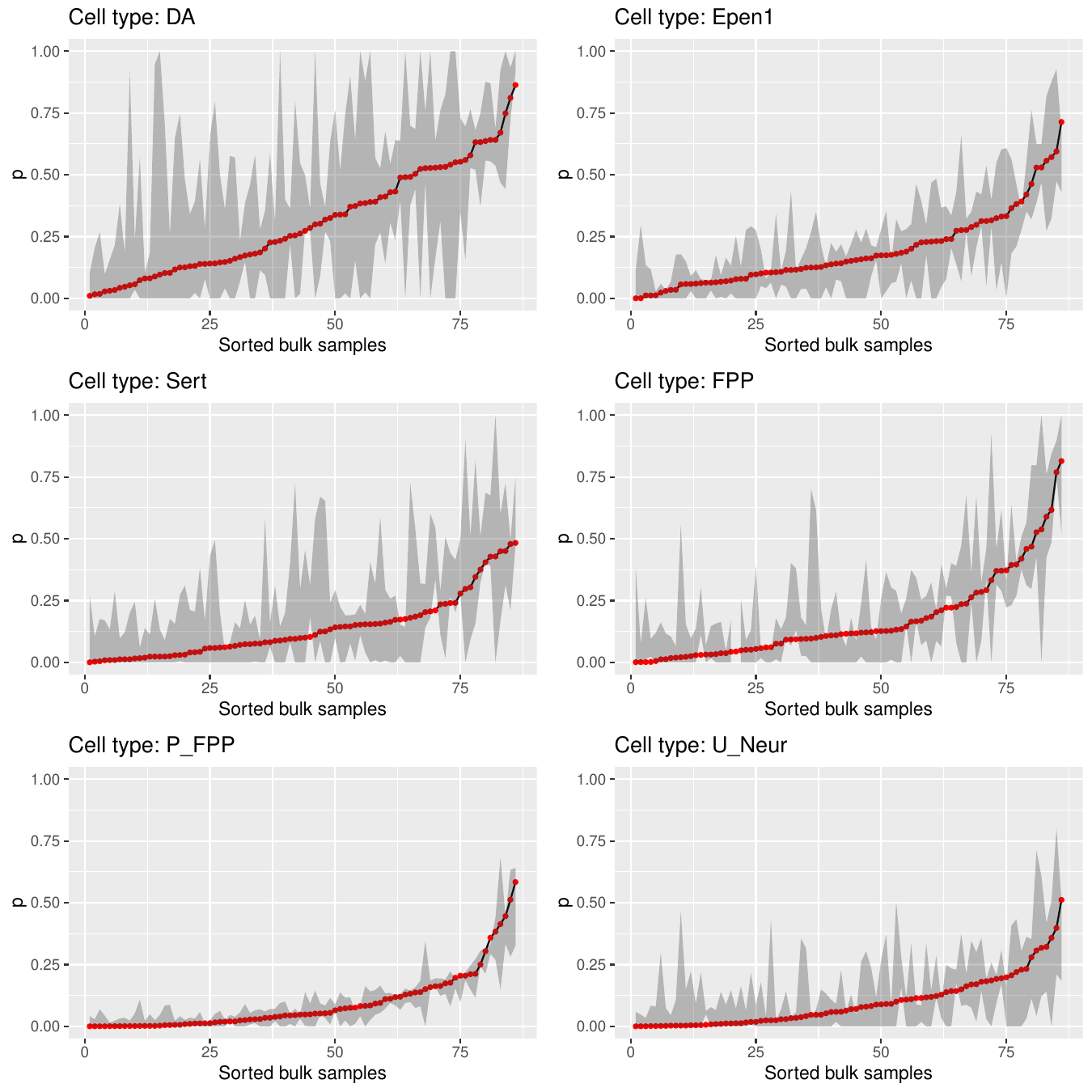}
    \caption{Confidence intervals of cell-type proportions in $86$ target individuals from one random split of the individuals. The target samples are sorted for each cell type in ascending order according to the true cell type proportions. The grey shaded areas represent the confidence intervals, and read dotted lines indicate the true proportions.}
    \label{fig:neuron_ci}
\end{figure}

\newpage

\section{Supplementary text}
\subsection{Cell level model for the reference data}\label{sec:cell_model}

In Section~\ref{sec:model} of the main text, for a reference individual $j$, we started with the cell-type level model:

\begin{equation*}
    \rr{\bm{Z}}_j = \rr{\gamma_j}\text{diag}(\rr{\bm \alpha})\rr{\bm X_j} + \rr{\bm E_j},\quad \E\left(\rr{\bm E_j}\mid \rr{\bm X_j}\right) = \bm 0.
\end{equation*}
As scRNA-seq measures gene expressions for individual cells, we now provide how the cell-type level model is derived from the raw scRNA-seq measurements. 

Specifically, denote $\rr{\bs y}_{jc}\in \mathbb{R}^{G}$ as the observed scRNA-seq counts for cell $c$ in individual $j$. Then we have
\begin{equation}\label{ref_sup}
\begin{split}
        \rr{\bs y}_{jc}  =  \rr{\gamma_{jc}}\text{diag}(\rr{\bm \alpha})\rr{\tilde{\bm x}_{jc}} + \rr{\bm e_{jc}},\quad \E\left(\rr{\bm e_{jc}}\mid \rr{\tilde{\bm x}_{jc}}, \rr{\bs X}_j\right) = \bm 0
\end{split}
\end{equation}
Here, ${\rr{\tilde{\bs x}}}_{jc}$ is the true gene expression level in individual $j$ and cell $c$ and the term $\rr{\gamma}_{jc}$ represents cell-specific measurement. The gene-specific scaling factors $\rr{\bm \alpha}$ are the same in model~\eqref{ref0} in the main text. The term $\rr{\bs e}_{jc}$ represents the centered measurement errors. 

Let $\rr{\bm{Z}}_j = (\rr{\bm z}_{j1}, \cdots, \rr{\bm z}_{jK})$ where each $\rr{\bm z}_{jk}$ represents the observed average gene expression within cell type $k$. Also, let $\rr{\bm{X}}_j = (\rr{\bm x}_{j1}, \cdots, \rr{\bm x}_{jK})$ where each $\rr{\bm x}_{jk}$ represents the true average gene expression within cell type $k$. Define the set of cells belonging to cell type $k$ in individual $j$ as $\mathcal{C}_{jk}$. Now we require the following assumptions to derive model~\eqref{ref0} in the main text from model~\eqref{ref_sup}.

\begin{assumption}
    [Unbiased sampling of the cells]\label{asp:unbiased_sampling}
scRNA-seq provides an unbiased sampling of cells within each cell type. Specifically, for a captured cell $c$ from reference individual $j$, it satisfies $\EE{\rr{\tilde{\bm x}_{jc}}\mid \rr{\bm X}_{j}} = \rr{\bm x}_{jk}$ if $ c\in \mathcal{C}_{jk}$.
\end{assumption}

\begin{assumption}[Random cell-specific efficiencies]\label{asp:scaling}
For any reference individual $j$ and any cell $c$, $$\rr{\gamma_{jc}} \indep \rr{\tilde{\bm x}_{jc}} \mid \rr{\bs X}_j.$$ Also, regardless of the cell type of cell $c$, $\EE{\rr{\gamma_{jc}}} = \rr\gamma_j$ always holds where $\rr\gamma_j$ is the subject-specific scaling factor for individual $j$, as defined in model~\eqref{ref0} of the main text.
\end{assumption}

\begin{remark}
In practice, the scaling factors may distribute differently across cell types. Then, one may relax the assumption to $\rr{\gamma_{jc}} = \rr\gamma_j/\delta_k$ for some $\delta_k$ to allow heterogeneity across cell types. However, we can then only identify a rescaled cell type proportions defined as $\tilde p_{ik} = \delta_kp_{ik}/\sum_k'\delta_{k'}p_{ik'}$. See also a similar discussion in \cite{wang2019bulk}.
\end{remark}

Given Assumptions~\ref{asp:unbiased_sampling}-\ref{asp:scaling}, if cell $, c\in \mathcal{C}_{jk}$, we can rewrite the model for $\rr{\bs y}_{jc}$ as
\begin{equation*}
\begin{split}
        \rr{\bs y}_{jc}  =  \rr{\gamma_{j}}\text{diag}(\rr{\bm \alpha})\rr{\bm x_{jk}} + \rr{\tilde{\bm e}_{jc}}
\end{split}
\end{equation*}
where $\EE{\rr{\tilde{\bm e}_{jc}}\mid \rr{\bm X_{j}}} = \EE{(\rr{\gamma_{jc}}-\rr{\gamma_{j}})\text{diag}(\rr{\bm \alpha})\rr{\tilde{\bm x}_{jc}} + \rr\gamma_j\text{diag}(\rr{\bm \alpha})(\rr{\tilde{\bm x}_{jc}}-\rr{\bm x_{jk}}) + \rr{\bm e}_{jc}\mid \rr{\bm X_{j}}}=\bm 0$.
Now we define the observed average gene expression within cell type $k$ as
$$\rr{\bm z}_{jk}=\frac{1}{|\mathcal{C}_{jk}|}\sum_{c\in \mathcal{C}_{jk}}\rr{\bs y}_{jc}= \rr{\gamma_{j}}\text{diag}(\rr{\bm \alpha})\rr{\bm x_{jk}} + \frac{1}{|\mathcal{C}_{jk}|}\sum_{c\in \mathcal{C}_{jk}}\rr{\tilde{\bm e}_{jc}}.$$

\begin{remark}
In the paper, we estimate an individual-level common scaling factor $\widehat \gamma_j$ for $\rr{\gamma_{j}}$ while a more common scaling approach for scRNA-seq data is to work with the normalized gene expressions  $\rr{\tilde{\bs y}}_{jc} = \rr{\bs y}_{jc}/\widehat{\gamma}_{jc}$ which applies different scaling
factors (library size) $\widehat{\gamma}_{jc} = \left(\rr{\bs y}_{jc}\right)^\top \bs 1$ to different cells. 
We avoid using cell-specific scaling factors for easier theoretical analysis and to account for differences in cell sizes across cell types \citep{wang2019bulk}.
\end{remark}

\newpage

\subsection{Proofs}

In model~\eqref{eq:final_bulk} of the main text, the noise term $\bs e_i$ contains the cross-platform scaling ratios $\bm \Lambda$ that are shared across all the target individuals, thus these noise terms are not independent across $i$. When we are comparing multiple target individuals as discussed in Section~\ref{sec:infer_multi}, we need to specifically account for such dependence. Thus in the proof, to simplify the description we decompose $\bs y_i$ following model~\eqref{bulk} in the main text where 
\begin{equation}\label{bulk1}
  \bm{y}_i = \bm \Lambda\bm U\bm \beta_i + \bm\epsilon_i'
\end{equation}

\begin{corollary}\label{cor:identification1}
Under Assumptions~\ref{asp:homo_popu}-\ref{asp:scaling_constraints}, the scaling factors $\rr{\tilde\gamma_j}$ for all reference individuals and $\bm U$ are both identifiable. Specifically, if there exists another set of parameters $\{\rr{\tilde\gamma_j}, j = 1, \cdots M\}$ and $\tilde{\bm U}$ that yield the same distribution of the observed reference data, then  
$$\tilde{\bm U} = \bm U, \quad \rr{\tilde\gamma_j} = \rr{\gamma_j}, \forall j = 1, \cdots, M.$$
\end{corollary}

\begin{proof}

Notice that model~\eqref{bulk} for the reference data
 is similar to a two-way ANOVA model. To show the identification of $\rr{\gamma_j}$, take an average across all entries of $\rr{\bm{Z}}_j$ in model~\eqref{bulk}:
$$\rr{\bar{z}}_{j\cdot\cdot} = \gamma_j^r\bar \mu_{\cdot\cdot} + \rr{\bar\epsilon}_{j\cdot\cdot}$$
where $\bar \mu_{\cdot\cdot} = \sum_{j,k}\mu_{jk}/KG =1$ following Assumption~\ref{asp:scaling_constraints}.
Since $\E(\rr{\bar\epsilon}_{j\cdot\cdot}) = 0$, if there exists another set of parameters $\{\rr{\tilde\gamma}_j, j =1, \cdots, M\}$ and $\tilde U$ that result in the same distribution of $\{\rr{\bm{Z}}_j, j =1, \cdots, M\}$, then we have
$\rr{\gamma}_j = \rr{\gamma}_j\bar \mu_{\cdot\cdot} = \rr{\tilde\gamma}_j\bar{\tilde \mu}_{\cdot\cdot}=\rr{\tilde\gamma}_j$.
As a consequence, we also have 
$\tilde {\bs U} =  \bs U$
which completes the proof.

\end{proof}

\subsubsection{Proof of Theorem~\ref{thm:identify}}

Similar to the proof of Corollary~\ref{cor:identification1}, in model~\eqref{bulk}, the matrix $\bs \Lambda\bs U \bs P$ is identifiable. 
If $\rank(\bs \Lambda\bs U) < K$, then $\bs P$ is not identifiable even when $\bs \Lambda\bs U$ is identifiable, so $\rank(\bs \Lambda\bs U) = K$ is a necessary condition for the identifiability of $\bs P$. As a result, we assume that $G \geq K$ and for any $g$, $\alpha_g \neq 0$ and $\mu_{gk} \neq 0$ for at least one $k$. Since under Assumptions~\ref{asp:homo_popu}-\ref{asp:scaling_constraints}, $\bs U$ is identifiable, to prove Theorem~\ref{thm:identify}, we only need to show that if $\bs\Lambda_1\bs U\bs P_1 = \bs\Lambda_2\bs U\bs P_2$ and $\text{rank}(\bs \Lambda_1\bs U\bs P_1) = K$ then there exists some constant $c$ that $\bs P_2=c\bs P_1$ if and only if for any disjoint partition $\{I_1, I_2, \cdots, I_t\}$ of $\{1, 2, \cdots, G\} = \bigcup_{s=1}^t I_s$ with $t\geq 2$, we have $\sum_{s=1}^t \rank(\bs U_{I_s})>K$. Matrix $\bs P_1$ and $\bs P_2$ do not need to satisfy the constraint that each column has sum $1$ as we can always rescale the unobserved $\gamma_i$ in each target individual so that $\bs p_i^\top \bs 1 = 1$ in model~\eqref{bulk}.

To show this, notice that since $\bs \Lambda_1\bs U,\bs \Lambda_2\bs U, \bs P_1, \bs P_2$ all have full rank, we have $\text{span}(\bs P_1) = \text{span}(\bs P_2)$. Hence there exists an invertiable matrix $\bs V\in \mathbb R^{K\times K}$ such that $\bs P_2^\top = \bs P_1^\top\bs V$. Now we only need to prove that there does not exist $\bs V \neq c_0\bs I$ for any constant $c_0$ if and only if the inequality $\sum_{s=1}^t \rank(\bs U_{I_s})>K$ holds for any disjoint partition $\{1, 2, \cdots, G\} = \bigcup_{s=1}^t I_s$ with $t\geq 2$.

Denote $\bs \mu_g$ as each row vector of the matrix $\bs U$, and let the $g$th diagonal element of $\bs \Lambda_1$ and $\bs \Lambda_2$ be $\lambda_{g1}$ and $\lambda_{g2}$. Then $\bs\Lambda_1\bs U\bs P_1 = \bs\Lambda_2\bs U\bs P_2$ is equivalent to $$\lambda_{g1}\bs P_1^\top\bs \mu_g = \lambda_{g2}\bs P_2^\top\bs \mu_g = \lambda_{g, 2}\bs P_1^\top\bs V\bs \mu_g,$$ which leads to 
$$\bs P_1^\top\left(\bs V\bs \mu_g - \frac{\lambda_{g1}}{\lambda_{g2}}\bs \mu_g\right) = 0.$$
Since $\bs P_1$ has full rank, we have $\bs V\bs \mu_g - \frac{\lambda_{g1}}{\lambda_{g2}}\bs \mu_g = \bs 0$, so $\{\lambda_{g1}/\lambda_{g2}\}$ and $\{\bs\mu_g\}$ are eigenvalues and eigenvectors of $\bs V$.

``if'': If the condition on partitions holds and there exists $\bs V \neq c_0\bs I$, then $\lambda_{g1}/\lambda_{g2}$ also takes different value $s_1,...,s_D$ where $D\geq 2$, Denote the submatrix $\bs U_d = \bs U_{\{g: \alpha_g=s_d\}}$, then we will have $\sum_{d = 1}^D \text{rank}(\bs U_d) \leq K$. To see this, notice that the matrix $\bs V$ is similar to its Jordan canonical form $\bs J$. Also, $\rank(\bs U_d)$ is at most the geometric dimension of eigenvalue $s_d$, which equals to the number of Jordan blocks corresponding to $s_d$. Let $\bs J_i$ be the $i$th Jordan block, then $\sum_{d = 1}^D \text{rank}(\bs U_d) \leq \sum_{i} \text{rank}(\bs J_i)= K$. This contradicts with the condition on the partitions.

``only if'': assume that there exists some partition with $\sum_{s = 1}^t \text{rank}(\bs U_{I_s})\leq K$. As $\sum_s \text{rank}(\bs U_{I_s}) \geq \rank(\bs U) = K$, we actually have $\sum_{s = 1}^t \text{rank}(\bs U_{I_s})= K$. Let $\tilde{\bs U}_s \in \mathbb R^{K \times n_s}$ be the matrix whose columns form the orthogonal basis of $\{\bs \mu_g: g\in I_s\}$. Then $\sum_s n_s = K$ and we can construct a rank $K$ matrix $\tilde{\bs U}_0 = (\tilde{\bs U}_1,...,\tilde{\bs U}_t) \in \mathbb R^{K \times K}$. 
Let $\bs D = \text{diag}(d_1, \cdots, d_1,  \cdots, d_t, \cdots, d_t)$ be a K-dimensional diagonal matrix where each $d_s$ replicates $n_s$ times and $d_1, \cdots, d_t$ are not all equal. Then we can construct $ \bs V= \tilde{\bs U}\bs D \tilde{\bs U}^{-1}$. As columns of $\tilde{\bs U}$ are eigenvectors of $\bs V$, and columns of each $\tilde{\bs U}_s$ share the same eigenvalue, each $\mu_g$ is also an eigenvector of $\bs V$. So if the condition on the partitions does not hold, for any values of $\{d_1, \cdots, d_t\}$ we can construct a matrix $\bs V$ satisfying $\bs V\bs \mu_g - d_s\bs \mu_g = \bs 0$ if $\bs \mu_g$ is a row of $\bs U_{I_s}$. We also have $\bs V \neq c_0 \bs I$ for any $c_0$ as long as $d_1, \cdots, d_t$ are not all equal.

\subsubsection{Proof of Theorem~\ref{thm:consistency}}

Notice that by definition in model~\eqref{bulk1},
\begin{align*}
    \bm \phi(\bm \beta_i) & = \widehat{\bm U}^\top \bm W \bm y_i - (\widehat{\bm U}^\top \bm W \widehat{\bm U} - \sum_{g=1}^G w_g \widehat {\bm V}_g) \bm{\beta}_i \\
   & = \widehat{\bm U}^\top \bm W \bs \epsilon_i' - \left(\widehat{\bm U}^\top \bm W (\widehat{\bm U} - \bs \Lambda\bm U) - \sum_{g=1}^G w_g \widehat {\bm V}_g\right) \bm{\beta}_i
\end{align*}
Define $\bs H = \widehat{\bm U}^\top \bm W (\widehat{\bm U} - \bs \Lambda\bm U) - \sum_{g=1}^G w_g \widehat {\bm V}_g$, then by definition
$$\bm \phi(\bm \beta_i)  = \widehat{\bm U}^\top \bm W \bm \epsilon_i' - \bs H \bm{\beta}_i.$$

At the same time, we define oracle ``estimators'' with know scaling factors:
$$\widehat {\bm U}^\star = (\widehat {\bm \mu}_1^\star, \cdots, \widehat {\bm \mu}_G^\star)^\top$$
where $\widehat{\bs \mu}_{g}^\star = \frac{1}{M}\sum_{j = 1}^M \rr{\bs z}_{jg}/\rr\gamma_j$.
Additionally, we denote
$$\widehat{\bm V}_g^\star = \frac{1}{M(M-1)}\sum_{j = 1}^M \left(\frac{\rr{\bs z}_{jg}}{\rr\gamma_j} - \widehat{\bm\mu}_g^\star\right)\left(\frac{\rr{\bs z}_{jg}}{\rr\gamma_j} - \widehat{\bm\mu}_g^\star\right)^\top$$
$$\widehat{\bm \Omega}^\star = \frac{1}{G}\left(\sum_g w_g \widehat{\bm\mu}_g^\star\widehat{\bm\mu}_g^{\star \top} -\sum_g w_g \widehat{\bm V}_g^\star\right)$$
$$\bs H^\star = \widehat{\bm U}^{\star \top} \bm W (\widehat{\bm U}^\star - \bs \Lambda\bm U) - \sum_{g=1}^G w_g \widehat {\bm V}_g^\star.$$

\begin{lemma}\label{lem:var_est}
Under Assumptions~\ref{asp:homo_popu}-\ref{asp:bias}, we have $\EE{\widehat{\bs V}_g^\star} = \mathrm{Cov}(\widehat{\bm\mu}_g^\star)$ for each $g = 1, 2, \cdots, G$.
\end{lemma}


 
\begin{lemma}\label{lem:con1}
Under the assumptions of Theorem~\ref{thm:consistency}, we have
$$\widehat{\bm \Omega}^\star - \bm \Omega = O_p\left(\frac{1}{\sqrt G}\right), \quad \widehat{\bm U}^{\star \top} \bm W \bm \epsilon_i' = O_p(\sqrt G), \quad \bm H^\star = O_p(\sqrt G)$$

\end{lemma}

\begin{proof}[Proof of Lemma~\ref{lem:con1}]

By definition and using Lemma~\ref{lem:var_est}, it is straightforward that 
$$\EE{\widehat{\bm \Omega}^\star} = \bm \Omega, \quad \EE{\widehat{\bm U}^{\star \top} \bm W \bm \epsilon_i'} =\bm 0, \quad\EE{\bm H^\star} = \bm 0.$$
Denote $\bs \epsilon_i'=(\epsilon_{i1}',\cdots, \epsilon_{iG}')$. Then we can rewrite as 
$$ \widehat{\bm U}^{\star \top} \bm W \bm \epsilon_i'= \sum_{g = 1}^G w_g\widehat{\bs \mu}_g^\star\epsilon_{ig}'$$
$$\bm H^\star = \sum_gw_g\left(\widehat{\bs \mu}_g^\star(\widehat{\bs \mu}_g^\star - \lambda_g\bs \mu_g)^\top - \widehat {\bm V}_g^\star\right).$$

Under Assumption~\ref{assumption_consistency}b of bounded moments, for any $k_1\leq K$ and $k_2\leq K$,
$\Var{\widehat{\mu}_{gk_1}^\star\widehat{\mu}_{gk_2}^\star}$, $\Var{\widehat V_{g, k_1k_2}^\star}$ and $\Var{\lambda_g\widehat\mu_{gk}^\star}$ are all uniformly bounded across $g$ (note: $V_{g, k_1k_2}$ denotes the $(k_1k_2)$th element of $\bs V_g$). Thus 
$$\max_g\Cov{\vect\left(\widehat{\bm\mu}_g^\star\widehat{\bm\mu}_g^{\star\top}\right)} = O(1), \quad \max_g\Cov{\vect\left(\widehat {\bs V}_g^\star\right)} = O(1)$$
This indicates that
$$\max_g\Cov{\vect\left(\widehat{\bm\mu}_g^\star\widehat{\bm\mu}_g^{\star\top} -\widehat {\bs V}_g^\star \right)} = O(1)$$
and 
$$ \max_g\Cov{\vect\left(\widehat{\bs \mu}_g^\star(\widehat{\bs \mu}_g^\star - \lambda_g\bs \mu_g)^\top - \widehat {\bm V}_g^\star \right)} = O(1).$$
In addition, as $\bs \epsilon_i' \indep \widehat {\bs \mu}_g^\star$, under Assumption~\ref{assumption_consistency}b
$$\max_g\Cov{\widehat{\bs \mu}_g^\star\epsilon_{ig}'} = \max_g\EE{\epsilon_{ig}'^2}\EE{\widehat{\bm\mu}_g^\star\widehat{\bm\mu}_g^{\star\top}} = O(1).$$
Since by Assumption~\ref{assmp:dependence}, the maximal degree of $\mathcal V$ is a constant, the number of dependent gene pairs in $\mathcal V$ is $O(G)$. Accordingly, combining with Assumption~\ref{assumption_consistency}c,
$$\Cov{\sum_{g\in\mathcal{V}} w_g\widehat{\bs \mu}_g^\star\epsilon_{ig}'} = O(G), \quad\Cov{\vect\left(\sum_{g \in \mathcal{V}}w_g\left(\widehat{\bs \mu}_g^\star\widehat{\bs \mu}_g^{\star\top} - \widehat {\bm V}_g^\star\right)\right)} = O(G)$$
$$\Cov{\vect\left(\sum_{g \in \mathcal{V}}w_g\left(\widehat{\bs \mu}_g^\star(\widehat{\bs \mu}_g^\star - \lambda_g\bs \mu_g)^\top - \widehat {\bm V}_g^\star\right)\right)} = O(G).$$
On the other hand, as $|\mathcal{V}^c| = o(\sqrt G)$, thus
$$\Cov{\sum_{g\in\mathcal{V}^c} w_g\widehat{\bs \mu}_g^\star\epsilon_{ig}'} = o(G), \quad\Cov{\vect\left(\sum_{g \in \mathcal{V}^c}w_g\left(\widehat{\bs \mu}_g^\star \widehat{\bs \mu}_g^{\star\top} - \widehat {\bm V}_g^\star\right)\right)} = o(G)$$
$$\Cov{\vect\left(\sum_{g \in \mathcal{V}^c}w_g\left(\widehat{\bs \mu}_g^\star(\widehat{\bs \mu}_g^\star - \lambda_g\bs \mu_g)^\top - \widehat {\bm V}_g^\star\right)\right)} = o(G).$$
Thus,
$$\Cov{\widehat{\bm \Omega}^\star} = O\left(\frac{1}{G}\right), \quad \Cov{\vect\left(\widehat{\bm U}^{\star \top} \bm W \bm \epsilon_i'\right)} = O(G), \quad \Cov{\vect\left(\bm H^\star\right)} = O(G)$$
So using Chebyshev's inequality,
$$\widehat{\bm \Omega}^\star - \bm \Omega = O_p\left(\frac{1}{\sqrt G}\right), \quad \widehat{\bm U}^{\star \top} \bm W \bm \epsilon_i' = O_p(\sqrt G), \quad \bm H^\star = O_p(\sqrt G).$$

\end{proof}

\begin{lemma}\label{lem:gap}
Under the assumptions of Theorem~\ref{thm:consistency}, we have 
$$\widehat{\bm \Omega} - \widehat{\bm \Omega}^\star = O_p\left(\frac{1}{\sqrt G}\right), \quad \widehat{\bm U}^{\top} \bm W \bm \epsilon_i' - \widehat{\bm U}^{\star \top} \bm W \bm \epsilon_i' = O_p(1)$$
$$\bs H - \bs H^\star = -\frac{1}{M}\sum_{j=1}^M\frac{1}{\rr\gamma_j}(\widehat \gamma_j - \rr\gamma_j)\bs U^\top\bs W\bs U + O_p(1)$$
with 
$$\widehat \gamma_j = \rr\gamma_j + O_p\left(\frac{1}{\sqrt G}\right)$$

\end{lemma}

\begin{proof}[Proof of Lemma~\ref{lem:gap}]
First, we show that for each reference individual $j$, the estimate $\widehat\gamma_j - \rr\gamma_j = O_p(1/\sqrt G)$ when $G \to \infty$. 
Notice that Under Assumption~\ref{assumption_consistency}b and 
Assumption~\ref{assmp:dependence}
$$\text{Var}\left(\widehat \gamma_j\right) =  \frac{1}{G^2}\Var{\sum_{g = 1}^G\frac{\sum_{k = 1}^K\rr z_{jgk}}{K}} = O\left(\frac{1}{G}\right).$$
As $\EE{\widehat \gamma_j} = \rr\gamma_j$, we have $\widehat \gamma_j - \rr\gamma_j = O_p(1/\sqrt G)$ by Chebyshev's inequality.

Next, by definition
\begin{align*}
& \widehat{\bs U}^\top\bs W\widehat{\bs U}-\widehat{\bs U}^{\star\top}\bs W\widehat{\bs U}^\star
=\sum_g w_g \widehat{\bs \mu}_g\widehat{\bs \mu}_g^\top - \sum_g w_g \widehat{\bm\mu}_g^\star \widehat{\bm\mu}_g^{\star\top} \\
= & \frac{1}{M^2}\sum_{j_1, j_2 = 1}^M\left(\frac{\rr\gamma_{j_1}\rr\gamma_{j_2}}{\widehat\gamma_{j_1}\widehat\gamma_{j_2}}- 1\right)\sum_{g = 1}^G w_g \frac{\rr{\bm z}_{j_1g}}{\rr\gamma_{j_1}}\frac{(\rr{\bm z}_{j_2g})^\top}{\rr\gamma_{j_2}} \\
= & \frac{1}{M^2}\sum_{j_1, j_2 = 1}^M\left(\frac{\rr\gamma_{j_1}\rr\gamma_{j_2}}{\widehat\gamma_{j_1}\widehat\gamma_{j_2}}- 1\right)\left(\sum_{g = 1}^G w_g \frac{\rr{\bm z}_{j_1g}}{\rr\gamma_{j_1}}\frac{(\rr{\bm z}_{j_2g})^\top}{\gamma_{j_2}} - \sum_g w_g\bs \mu_g\bs \mu_g^\top\right)\\ 
&+ \frac{1}{M^2}\sum_{j_1, j_2 = 1}^M\left(\frac{\rr\gamma_{j_1}\rr\gamma_{j_2}}{\widehat\gamma_{j_1}\widehat\gamma_{j_2}}- 1\right)\bs U^\top\bs W\bs U.
\end{align*}
Using the same logic as in the proof of Lemma~\ref{lem:con1}, we have 
$$\sum_{g = 1}^G w_g \frac{\rr{\bm z}_{j_1g}}{\rr\gamma_{j_1}}\frac{(\rr{\bm z}_{j_2g})^\top}{\gamma_{j_2}} - \sum_{g = 1}^G w_g\EE{ \frac{\rr{\bm z}_{j_1g}}{\rr\gamma_{j_1}}\frac{(\rr{\bm z}_{j_2g})^\top}{\gamma_{j_2}}} = O_p(\sqrt G)$$
with $\EE{\frac{\rr{\bm z}_{j_1g}}{\rr\gamma_{j_1}}\frac{(\rr{\bm z}_{j_2g})^\top}{\gamma_{j_2}}} = \bs \mu_g\bs \mu_g^\top$. As $\frac{\rr\gamma_{j_1}\rr\gamma_{j_2}}{\widehat\gamma_{j_1}\widehat\gamma_{j_2}}- 1 = O_p(1/\sqrt G)$ and $M$ is fixed,
\begin{align*}
&\widehat{\bs U}^\top\bs W\widehat{\bs U}-\widehat{\bs U}^{\star\top}\bs W\widehat{\bs U}^\star
=\sum_g w_g \widehat{\bs \mu}_g\widehat{\bs \mu}_g^\top - \sum_g w_g \widehat{\bm\mu}_g^\star\left(\widehat{\bm\mu}_g^\star\right)^\top\\
=& \frac{1}{M^2}\sum_{j_1, j_2 = 1}^M\left(\frac{\rr\gamma_{j_1}\rr\gamma_{j_2}}{\widehat\gamma_{j_1}\widehat\gamma_{j_2}}- 1\right)\bs U^\top\bs W\bs U + O_p(1).    
\end{align*}
Similarly, under Assumption~\ref{assumption_consistency}b of bounded moments,
\begin{align*}
& \sum_{j = 1}^M\left(\frac{(\rr\gamma_j)^2}{\widehat \gamma_j^2}-1\right)\sum_g w_g \frac{\rr{\bm z}_{jg}}{\rr\gamma_{j}}\frac{(\rr{\bm z}_{jg})^\top}{\gamma_{j}}\\
= & \sum_{j = 1}^M\left(\frac{(\rr\gamma_j)^2}{\widehat \gamma_j^2}-1\right)\left(\sum_g w_g \frac{\rr{\bm z}_{jg}}{\rr\gamma_{j}}\frac{(\rr{\bm z}_{jg})^\top}{\gamma_{j}} - \sum_g w_g\bs \mu_g \bs \mu_g^\top\right) + \sum_{j = 1}^M\left(\frac{(\rr\gamma_j)^2}{\widehat \gamma_j^2}-1\right)\bs U^\top\bs W\bs U\\
= & \sum_{j = 1}^M\left(\frac{(\rr\gamma_j)^2}{\widehat \gamma_j^2}-1\right)\bs U^\top\bs W\bs U + O_p(1).
\end{align*}
Then, it holds that
\begin{align*}
& \widehat {\bm V}-\widehat {\bm V}^\star
=\sum_g w_g\widehat {\bm V}_g - \sum_g w_g\widehat{\bm V}_g^\star \\
= & \frac{1}{M(M-1)}\sum_{j = 1}^M\left(\frac{(\rr\gamma_j)^2}{\widehat \gamma_j^2}-1\right)\sum_g w_g
\frac{\rr{\bm z}_{jg}}{\rr\gamma_{j}}\frac{(\rr{\bm z}_{jg})^\top}{\gamma_{j}} - \frac{1}{M-1}\sum_g w_g \left(\widehat{\bm\mu}_g\widehat{\bm\mu}_g^\top -\widehat{\bm\mu}_g^\star\widehat{\bm\mu}_g^{\star \top}\right)\\
=&\frac{1}{M(M-1)}\left[\sum_{j = 1}^M\left(\frac{(\rr\gamma_j)^2}{\widehat \gamma_j^2}-1\right)-\frac{1}{M}\sum_{j_1, j_2 = 1}^M\left(\frac{\rr\gamma_{j_1}\rr\gamma_{j_2}}{\widehat\gamma_{j_1}\widehat\gamma_{j_2}}- 1\right)\right]\bs U^\top\bs W\bs U + O_p(1).
\end{align*}
As a result, with Assumption~\ref{assumption_consistency}a,
\begin{align*}
&\widehat{\bm \Omega} - \widehat{\bm \Omega}^\star
=\frac{1}{G}\left((\widehat{\bm U}^{\top} \bm W \widehat{\bm U} -  \widehat{\bm U}^{\star \top} \bm W \widehat{\bm U}^\star) -(\sum_{g=1}^G w_g \widehat {\bm V}_g-\sum_{g=1}^G w_g \widehat {\bm V}_g^\star)\right)\\
= & \frac{1}{M(M-1)}\left[\sum_{j_1, j_2 = 1}^M\left(\frac{\rr\gamma_{j_1}\rr\gamma_{j_2}}{\widehat\gamma_{j_1}\widehat\gamma_{j_2}}- 1\right) - \sum_{j = 1}^M\left(\frac{(\rr\gamma_j)^2}{\widehat \gamma_j^2}-1\right)\right]\frac{\bs U^\top\bs W\bs U}{G} + O_p(1/G) \\
=& O_p(1/\sqrt{G})
\end{align*}
On the other hand,
\begin{align*}
& \bs H - \bs H^\star \\
= & (\widehat{\bm U}^\top \bm W \widehat{\bm U} -  \widehat{\bm U}^{\star \top} \bm W \widehat{\bm U}^\star) -  (\widehat{\bm U}^\top \bm W \bs \Lambda\bm U - \widehat{\bm U}^{\star \top} \bm W \bs \Lambda\bm U) - (\sum_{g=1}^G w_g \widehat {\bm V}_g-\sum_{g=1}^G w_g \widehat {\bm V}_g^\star)\\
=& \frac{1}{M(M-1)}\left\{\sum_{j_1, j_2 = 1}^M\left(\frac{\rr\gamma_{j_1}\rr\gamma_{j_2}}{\widehat\gamma_{j_1}\widehat\gamma_{j_2}}- 1\right) - \sum_{j = 1}^M\left(\frac{(\rr\gamma_j)^2}{\widehat \gamma_j^2}-1\right)\right\}{\bs U^\top\bs W\bs U}\\
&- (\widehat{\bm U}^\top \bm W \bs \Lambda\bm U - \widehat{\bm U}^{\star \top} \bm W \bs \Lambda\bm U) + O_p(1).
\end{align*}
Since
\begin{align*}
  &\widehat{\bm U}^\top \bm W \bs \Lambda\bm U - \widehat{\bm U}^{\star \top} \bm W \bs \Lambda\bm U\\
= & \frac{1}{M}\sum_{j = 1}^M\left(\frac{\rr\gamma_{j}}{\widehat\gamma_{j}}- 1\right)\sum_{g = 1}^G w_g \lambda_g\frac{\rr{\bm z}_{jg}}{\gamma_{j}}\bs\mu_{g}^\top \\
= & \frac{1}{M}\sum_{j = 1}^M\left(\frac{\rr\gamma_{j}}{\widehat\gamma_{j}}- 1\right)\left(\sum_{g = 1}^G w_g\lambda_g \frac{\rr{\bm z}_{jg}}{\gamma_{j}}\bs\mu_{g}^\top -\sum_g w_g\bs \mu_g \bs \mu_g^\top \right) + \frac{1}{M}\sum_{j = 1}^M\left(\frac{\rr\gamma_{j}}{\widehat\gamma_{j}}- 1\right)\bs U^\top\bs W\bs U \\
=&  \frac{1}{M}\sum_{j = 1}^M\left(\frac{\rr\gamma_{j}}{\widehat\gamma_{j}}- 1\right)\bs U^\top\bs W\bs U + O_p(1).
\end{align*}
Combining all above, we 
get 
$$\bs H - \bs H^\star = h(\widehat{\gamma}_1, \cdots, \widehat{\gamma}_m)\bs U^\top\bs W\bs U + O_p(1)$$
where the function
$$h(x_1, \cdots, x_M) = \frac{1}{M(M-1)}\sum_{j_1\neq j_2}\frac{\rr\gamma_{j_1}\rr\gamma_{j_2}}{x_{j_1}x_{j_2}} - \frac{1}{M}\sum_j \frac{\rr\gamma_j}{x_j} $$
Taking the derivative, we find that for this function we have 
$$\frac{\partial h}{\partial x_j}{(\rr\gamma_1, \cdots, \rr\gamma_M)} = -\frac{1}{M\rr\gamma_j}.$$
So if we take Taylor expansion of $h(\cdot)$ at the true value $(\rr\gamma_1, \cdots, \rr\gamma_m)$, then we have
$$h(\widehat{\gamma}_1, \cdots, \widehat{\gamma}_M) = h(\rr\gamma_1, \cdots, \rr\gamma_M)  - \frac{1}{M}\sum_{j = 1}^M \frac{1}{\rr\gamma_j}(\widehat \gamma_j - \rr\gamma_j) + O_p\left(\frac{1}{G}\right)$$
As $h(\rr\gamma_1, \cdots, \rr\gamma_M) = 0$, we further have
$$\bs H - \bs H^\star = -\frac{1}{M} \sum_{j = 1}^M \frac{1}{\rr\gamma_j}(\widehat \gamma_j - \rr\gamma_j)\bs U^\top\bs W\bs U + O_p(1).$$

Finally, 
\begin{align*}
  \widehat{\bm U}^\top \bm W \bm \epsilon_i' - \widehat{\bm U}^{\star \top} \bm W \bm \epsilon_i' 
   =  \frac{1}{M}\sum_{j = 1}^M\left(\frac{\rr\gamma_{j}}{\widehat\gamma_{j}}- 1\right)\sum_{g = 1}^G w_g \epsilon_{ig}'\frac{\rr{\bm z}_{jg}}{\rr\gamma_{j}}.
\end{align*}
As $\epsilon_{ig}' \indep \rr{\bm z}_{jg}$ (as they come from two different individuals), we have $\EE{\epsilon_{ig}'\frac{\rr{\bm z}_{jg}}{\rr\gamma_{j}}} = \bs 0$. So $\sum_{g = 1}^G w_g \epsilon_{ig}'\frac{\rr{\bm z}_{jg}}{\rr\gamma_{j}} = O_p(\sqrt G)$ and 
$$\widehat{\bm U}^\top \bm W \bm \epsilon_i' - \widehat{\bm U}^{\star \top} \bm W \bm \epsilon_i'  = O_p(1).$$

\end{proof}

\begin{proof}[Proof of Theorem~\ref{thm:consistency}]
Lemma~\ref{lem:gap} guarantees that $\widehat{\gamma}_j\overset{p}{\to}\rr\gamma_j$ for any $j$ when $G \to \infty$. 
Using Lemma~\ref{lem:con1} and Lemma~\ref{lem:gap}, we also have 
$$\widehat{\bm \Omega} = \widehat{\bm \Omega} - \widehat{\bm \Omega}^\star + \widehat{\bm \Omega}^\star \overset{p}{\to} \bm \Omega \succ 0$$
By the Continuous mapping theorem, additionally we have $\widehat{\bm \Omega}^{-1} \overset{p}{\to} \bm \Omega^{-1}$.

Now we show that $\widehat{\bs \beta}^\star_i \overset{p}{\to} \bs \beta_i$ where $\widehat{\bs \beta}^\star_i$ is either the truncation estimator $\widehat{\bs \beta}^\star_i = \widehat {\bs\beta}_i \vee \bs 0 $ or the constrained estimator from non-negative least squares. 
 For the truncation estimator $\widehat{\bs \beta}^\star_i = \widehat {\bs\beta}_i \vee \bs 0$, 
Since we have 
$$\frac{1}{G}\bm \phi(\bs \beta_i) =\frac{1}{G}\widehat{\bm U}^\top \bm W \bm \epsilon_i' - \frac{1}{G}\bs H \bm{\beta}_i  \overset{p}{\to} \bm 0,$$
then
 $$\widehat {\bs\beta}_i - \bs \beta_i = \frac{1}{G}\widehat{\bm \Omega}^{-1}\bm \phi(\bs \beta_i) \overset{p}{\to} \bm 0.$$
Thus,
 $\widehat{\bs \beta}^\star_i = \widehat {\bs\beta}_i \vee \bs 0 \overset{p}{\to} \bs\beta_i \vee \bs 0 = \bs \beta_i$. 

 For the constrained estimator from non-negative least squares where  
 $$\widehat {\bs \beta}_i^\star = \argmin_{{\bs \beta}_i \succeq \bs 0}(\bs y_i - \widehat{\bs U}\bs \beta_i)^\top \bs W \bs (\bs y_i - \widehat{\bs U}\bs \beta_i) - \bs \beta_i^\top\widehat{\bs V} \bs \beta_i \overset{\Delta}{=}\argmin_{{\bs \beta}_i \succeq \bs 0}l(\bs \beta_i),$$
plug in model~\eqref{bulk1} for $\bs y_i$ and denote the true $\bs \beta_i$ as $\bs \beta_{0i}$, we have 
\begin{align*}
    l(\bs \beta_i) 
    = \tilde l(\bs \beta_i) + 2 (\bs \Lambda\bs U \bs \beta_{0i}-\widehat{\bs U}\bs \beta_i)^\top \bs W \bs \epsilon_i' + \text{const} 
    & =  \tilde l(\bs \beta_i) + 2(\bs \beta_{0i}-\bs \beta_i)^\top \widehat{\bs U}^\top\bs W \bs \epsilon_i' + \text{const}\\
    & = \tilde l(\bs \beta_i) + O_p(\sqrt G)\|\bs \beta_{0i}-\bs \beta_i\|_2  + \text{const}
\end{align*}
where $\tilde l(\bs \beta_i)=(\bs \Lambda\bs U \bs \beta_{0i} - \widehat{\bs U}\bs \beta_i)^\top \bs W \bs (\bs \Lambda\bs U \bs \beta_{0i} - \widehat{\bs U}\bs \beta_i) - \bs \beta_i^\top\widehat{\bs V} \bs \beta_i$.
Additionally, expand\\ 
$\bs \Lambda\bs U \bs \beta_{0i} - \widehat{\bs U}\bs \beta_i = \bs \Lambda\bs U (\bs\beta_{0i}-\bs \beta_i) + (\bs \Lambda\bs U -\widehat{\bs U})\bs\beta_i$, we have
\begin{align*}
\tilde l(\bs \beta_i) = & (\bs \beta_{0i}-\bs \beta_i)^\top\bs U^\top\bs \Lambda \bs W\bs \Lambda \bs U(\bs \beta_{0i}-\bs \beta_i) + 2(\bs \beta_{0i}-\bs \beta_i)^\top\bs U^\top\bs \Lambda \bs W (\bs \Lambda\bs U -\widehat{\bs U})\bs\beta_i \\
& + \bs\beta_i^\top\left[(\bs \Lambda\bs U -\widehat{\bs U})^\top \bs W (\bs \Lambda\bs U -\widehat{\bs U})-\bs {\widehat V}\right]\bs\beta_i
\end{align*}
For the 2nd and 3rd terms, using results in Lemma~\ref{lem:con1} and Lemma~\ref{lem:gap}, we have
$$\bs U^\top\bs \Lambda \bs W (\bs \Lambda\bs U -\widehat{\bs U}) = O_p(\sqrt G), \quad (\bs \Lambda\bs U -\widehat{\bs U})^\top \bs W (\bs \Lambda\bs U -\widehat{\bs U})-\bs {\widehat V}=O_p(\sqrt G).$$
Thus, 
$$ (\bs \beta_{0i}-\bs \beta_i)^\top\bs U^\top\bs \Lambda \bs W (\bs \Lambda\bs U -\widehat{\bs U})\bs\beta_i = O_p(\sqrt G)\|\bs \beta_{0i}-\bs \beta_i\|_2 + O_p(\sqrt G)\|\bs \beta_{0i}-\bs \beta_i\|_2^2$$
$$  \bs\beta_i^\top\left[(\bs \Lambda\bs U -\widehat{\bs U})^\top \bs W (\bs \Lambda\bs U -\widehat{\bs U})-\bs {\widehat V}\right]\bs\beta_i = O_p(\sqrt G)\|\bs \beta_{0i}-\bs \beta_i\|_2 + O_p(\sqrt G)\|\bs \beta_{0i}-\bs \beta_i\|_2^2 + \mathrm{const}.$$
Additionally, as $\lambda_g\overset{i.i.d.}{\sim}[1,\sigma_0^2]$ across $g$ with bounded 4th moments by Assumption~\ref{assumption_consistency}b, it is easy to show by the Chebyshev inequality that 
$$\bs U^\top\bs \Lambda \bs W\bs \Lambda \bs U = (\sigma_0^2+1)\bs U^\top\bs W\bs U + O_p(\sqrt G).$$
Thus, we have
\[
\tilde l(\bs \beta_i)= (\sigma_0^2+1)(\bs \beta_{0i}-\bs \beta_i)^\top\bs U^\top\bs W\bs U(\bs \beta_{0i}-\bs \beta_i) + O_p(\sqrt G)\|\bs \beta_{0i}-\bs \beta_i\|_2 + O_p(\sqrt G)\|\bs \beta_{0i}-\bs \beta_i\|_2^2 + \mathrm{const},
\]
indicating that
$$l(\bs \beta_i)= (\sigma_0^2+1)(\bs \beta_{0i}-\bs \beta_i)^\top\bs U^\top\bs W\bs U(\bs \beta_{0i}-\bs \beta_i) + O_p(\sqrt G)\|\bs \beta_{0i}-\bs \beta_i\|_2 + O_p(\sqrt G)\|\bs \beta_{0i}-\bs \beta_i\|_2^2 + \mathrm{const}.$$
Thus, as $\bs U^\top\bs W\bs U^\top/G \overset{G\to \infty}{\to} \bs \Omega \succ 0$ under Assumption~\ref{assumption_consistency}a, for any $\epsilon>0$, when $G \to \infty$ we have
\begin{align*}
    &\PP{\|\widehat {\bs \beta}_i^\star - \bs \beta_{0i}\|_2\leq \epsilon} \geq \PP{l(\bs \beta_{0i})< \min_{\|\bs\beta_i -\bs \beta_{0i}\|_2 > \epsilon}l(\bs \beta_{i})} \\
    = &\PP{\min_{\|\bs\beta_i -\bs \beta_{0i}\|_2 > \epsilon}\left\{(\bs \beta_{0i}-\bs \beta_i)^\top\bs U^\top\bs W\bs U(\bs \beta_{0i}-\bs \beta_i) + O_p(\sqrt G)\|\bs \beta_{0i}-\bs \beta_i\|_2 + O_p(\sqrt G)\|\bs \beta_{0i}-\bs \beta_i\|_2^2\right\} > 0}\\
    = &  \PP{\min_{\|\bs\beta_i -\bs \beta_{0i}\|_2 > \epsilon}(\bs \beta_{0i}-\bs \beta_i)^\top\bs \Omega(\bs \beta_{0i}-\bs \beta_i) + O_p(1/\sqrt G)\|\bs \beta_{0i}-\bs \beta_i\|_2 + O_p(1/\sqrt G)\|\bs \beta_{0i}-\bs \beta_i\|_2^2 > 0}\\
    \geq & \PP{\lambda_{\text{min}}(\bs \Omega)\epsilon^2 - |O_p(1/\sqrt G)\epsilon| - |O_p(1/\sqrt G)\epsilon^2|  > 0}   \to  1,
\end{align*}
where $\lambda_{\text{min}}(\bs \Omega) > 0$ is a minimum eigenvalue of matrix $\bs \Omega$. Thus
$$\widehat {\bs \beta}_i^\star \overset{p}{\to} \bs \beta_{0i}.$$
 
 Finally, by the Continuous mapping theorem, we have 
$$\widehat{\bs p_i} = \frac{\widehat{\bs \beta}^\star_i}{\widehat{\bs \beta}_i^{\star \top} \bs 1}\overset{p}{\to}\frac{{\bs \beta}_i}{{\bs \beta}_i^\top \bs 1} = \bs p_i$$ 
when $G \to \infty$ for any target individual $i$.

\end{proof}

\subsubsection{Proof of Theorem~\ref{thm:clt}}
We need the following lemmas.



\begin{lemma}[Theorem 2.7 of Chen and Shao]\label{lem:chenzhao}
Let $\{X_i, i \in \mathcal{V}\}$ be random variables indexed by the vertices of a dependency graph and let $D$ be the maximum degree. Put $W = \sum_{i \in \mathcal{V}} X_i$. Assume that $\EE{W^2} = 1$, $\EE{X_i} = 0$ and $\EE{|X_i|^{p}} \leq \theta^{p}$ 
for $i \in \mathcal{V}$ and for some $\theta > 0$. Then
$$\sum_z \|\PP{W \leq z} - \Phi(z)\|\leq 75 D^{5(p - 1)}|V|\theta^p$$

\end{lemma}

\begin{lemma}\label{lem:pos_var}
   Under Assumptions~\ref{asp:homo_popu}-\ref{assmp:variance}, we have 
$    \bm\Sigma_i \overset{\Delta}{=} \lim_{G\to \infty}\Cov{\bs \phi(\bs \beta_i)}/G   \succ 0 $
for the target individual $i$. 
\end{lemma}

\begin{proof}[Proof of Lemma~\ref{lem:pos_var}]

First, notice that by definition of $\widehat{\bm U}^{\star \top}$ and $\widehat {\bm V}_g^\star$, and following model~\eqref{eq:final_bulk},
$$\widehat{\bm U}^{\star \top} \bm W (\widehat{\bm U}^\star - \bm U) - \sum_{g=1}^G w_g \widehat {\bm V}_g^\star = \sum_{g=1}^G w_g\left[\rr{\bar{\bs \epsilon}}_g(\rr{\bar{\bs \epsilon}}_g+\bs \mu_g)^\top - \text{Cov}_M\left(\rr{\bar{\bs \epsilon}}_g\right)\right] $$
On the other hand, given the definition of $\bs e_i$ in model~\eqref{eq:final_bulk}, we have
\begin{align*}
 \widehat{\bm U}^{\star \top} \bm W \bs \epsilon_i' - \bs H^\star\bs \beta_i & = 
\widehat{\bm U}^{\star \top} \bm W \bs e_i  -\left(\widehat{\bm U}^{\star \top} \bm W (\widehat{\bm U}^\star - \bm U) - \sum_{g=1}^G w_g \widehat {\bm V}_g^\star\right)\bs \beta_i   \\
& = 
\bs s_i  -\left(\widehat{\bm U}^{\star \top} \bm W (\widehat{\bm U}^\star - \bm U) - \sum_{g=1}^G w_g \widehat {\bm V}_g^\star\right)\bs \beta_i
\end{align*}
Thus using Lemma~\ref{lem:gap}, we have 
\begin{align*}
    \bm \phi(\bm \beta_i) & = \widehat{\bm U}^\top \bm W \bs \epsilon_i - \bs H\bs \beta_i \\
    & = \widehat{\bm U}^{\star \top} \bm W \bs \epsilon_i - \bs H^\star\bs \beta_i  + \frac{1}{M}\sum_{j=1}^M\frac{1}{\rr\gamma_j}(\widehat \gamma_j - \rr\gamma_j)\bs U^\top\bs W\bs U\bs \beta_i + O_p(1)\\
    &= \bs s_i  -\left(\widehat{\bm U}^{\star \top} \bm W (\widehat{\bm U}^\star - \bm U) - \sum_{g=1}^G w_g \widehat {\bm V}_g^\star\right)\bs \beta_i +\frac{G}{M}\sum_{j=1}^M\frac{1}{\rr\gamma_j}(\widehat \gamma_j - \rr\gamma_j)\bs \Omega\bs \beta_i + o_p(\sqrt G).
\end{align*}
Since 
$$\widehat{\gamma}_j-\gamma_j=\sum_{g = 1}^G\frac{\sum_{k = 1}^K\rr z_{jgk}}{KG} - \gamma_j = \sum_{g = 1}^G\frac{\sum_{k = 1}^K\rr \epsilon_{jgk}}{KG},$$
given the definition of $\rr{\bs H}$, we can rewrite $ \bm \phi(\bm \beta_i)$ as 
$$ \bm \phi(\bm \beta_i) = \bs s_i  -\rr{\bs H}\bs \beta_i  + o_p(\sqrt G).$$
As $\widehat {\bs U}^\star$ and $\rr{\bs H}$ only depends on the reference data, while $\bs e_i$ only depends on the target data, we have $\bs e_i \indep (\widehat {\bs U}^\star, \rr{\bs H})$. Also $\EE{\bs e_i} = \bs 0$, thus given that $\bs s_i = \widehat{\bm U}^{\star \top} \bm W \bs e_i$, we have
$$\Cov{\bs s_i, \rr{\bs H}\bs \beta_i} = \bs 0$$
Thus, under Assumption~\ref{assmp:variance}, we have
$$\bs \Sigma_i = \lim_{G\to \infty}\frac{\Cov{ \phi(\bm \beta_i) }}{G} =\lim_{G\to \infty}\frac{\Cov{\bs s_i - \rr{\bs H} \bs \beta_i}}{G} = \lim_{G\to \infty}\frac{\Cov{\bs s_i} + \Cov{\rr{\bs H} \bs \beta_i}}{G}\succ 0.$$
\end{proof}


\begin{lemma}\label{lem:score}
Under the assumptions of Theorem~\ref{thm:clt}, for each target individual $i$ the score function $\bs \phi(\bs \beta_i)$ satisfy
$$\frac{1}{\sqrt G} \bs \phi(\bs \beta_i) \overset{d}{\to} \mathcal{N}(\bs 0, \bs \Sigma_i).$$
\end{lemma}

\begin{proof}[Proof of Lemma \ref{lem:score}]
As shown in the proof of Lemma~\ref{lem:pos_var} 
\begin{align*}
    \frac{1}{G}\bm \phi(\bm \beta_i) 
    & = \frac{1}{G}\left(\widehat{\bm U}^{\star \top} \bm W \bs \epsilon_i' - \bs H^\star\bs \beta_i\right)  + \frac{1}{M}\sum_{j=1}^M\frac{1}{\rr\gamma_j}(\widehat \gamma_j - \rr\gamma_j)\bs \Omega\bs \beta_i + o_p(1/\sqrt G) \\
    &=\frac{1}{G} \sum_{g = 1}^G\bs \eta_g + o_p(\frac{1}{\sqrt G}).
\end{align*}
where 
$$\bs \eta_g \overset{\Delta}{=} w_g \epsilon_{ig}'\widehat {\bs \mu}_g^\star - w_g\left(\widehat{\bs \mu}_g^\star(\widehat{\bs \mu}_g^\star - \lambda_g\bs \mu_g)^\top - \widehat{\bs V}_g^\star\right)\bs \beta_i   + \sum_j\frac{1}{\rr\gamma_j}\frac{\sum_k \rr z_{jgk} - \sum_k \rr\gamma_j \mu_{gk}}{K}\bs \Omega\bs \beta_i.$$
Each $\EE{\bs \eta_g} = \bs 0$ and given Assumption~\ref{assmp:variance} and Lemma~\ref{lem:pos_var}, $\lim_{G \to \infty}\Var{\sum_{g = 1}^G\bs \eta_g}/G = \bs \Sigma_i \succ 0$.
Also, similar to our argument in the proof of Lemma~\ref{lem:con1}, under Assumption~\ref{assumption_consistency}bc, let $\bs \eta_g = (\eta_{g1}, \cdots, \eta_{gK})$, then
$\EE{\eta_{gk}^{2 + \delta/2}}$ is uniformly bounded across all genes $g$. 

Further under Assumption~\ref{assmp:dependence}, 
$\Cov{\sum_{g \in \mathcal{V}^c}\bs \eta_g} = o(G)$ as $|\mathcal{V}^c| = o(\sqrt G)$. Thus $\sum_{g \in \mathcal{V}^c}\bs \eta_g = o_p(\sqrt G)$ and
$$\frac{1}{G}\bs \phi(\bs \beta_i) = \frac{1}{G} \sum_{g \in \mathcal{V}}\bs \eta_g + o_p\left(\frac{1}{\sqrt G}\right).$$

Now let $\bm t \in \mathbb{R}^K$ be a non-random vector with $\|\bm t\|_2 = 1$. Then under Assumption~\ref{assmp:dependence}, $\{\bs 
\eta_g^\top \bs t, g \in \mathcal V\}$ forms a dependency graph with maximum degree $D = O(1)$. 
Additionally, we have
\[
\max_g\EE{(\bs \eta_g^\top \bs t)^{2 + \delta/2}} \leq c
\]
for some constant $c$.
Also, as $\Var{\sum_{g \in \mathcal V} \bs \eta_g^\top\bs t}/G  = \Var{\bs\phi(\bs \beta_i)\bs t}/G $ + o(1), we have 
$$\lim_{G\to \infty}\Var{\sum_{g \in \mathcal V} \bs \eta_g^\top\bs t}/G = \bm t^\top \bm\Sigma_{i} \bm t > 0.$$ Using Lemma~\ref{lem:chenzhao}, we have
$$\frac{1}{\sqrt G}\left(\sum_{g \in \mathcal V} \bs \eta_g^\top\bs t
\right) \overset{d}{\to} N(0, \bm t^\top \bm\Sigma_{i} \bm t).$$
Then, using the Cramer-wold theorem, we can obtain
$$\frac{1}{\sqrt G}\sum_{g \in \mathcal{V} } \bs \eta_g \overset{d}{\to} \mathcal{N}(\bs 0, \bs \Sigma_i).$$
which implies that
$$\frac{1}{\sqrt G} \bs \phi(\bs \beta_i) \overset{d}{\to} \mathcal{N}(\bs 0, \bs \Sigma_i).$$

\end{proof}

\begin{proof}[Proof of Theorem \ref{thm:clt}]

Notice that 
$$\sqrt G\left(\widehat {\bs \beta}_i - \bs \beta_i\right) = \widehat {\bs \Omega}^{-1} \bs \phi(\bs \beta_i)/\sqrt G $$

Then combining Theorem~\ref{thm:consistency} and Lemma~\ref{lem:score}, we have 
\begin{equation}\label{eq:CLT_beta}
    \sqrt G(\widehat {\bs \beta}_i - \bs \beta_i) \overset{d}{\to} \mathcal{N}(\bs 0, \bs \Omega^{-1}\bs \Sigma_i \bs \Omega^{-1})
\end{equation}
Next, notice that for either the truncated on the constrained estimator, $\widehat{\bs \beta}_i^\star \neq \widehat {\bm \beta}_i$ only when at least one $\widehat \beta_{ik} <0$. For a target individual $i$, if $p_{ik} > 0$ for any $k$, then $\beta_{ik}> 0$ for any $k$. Thus, for any $\epsilon > 0$,
$$\PP{\|\sqrt G(\widehat{\bs \beta}_i^\star - \widehat{\bs \beta}_i)\|_2 > \epsilon} \leq \PP{\widehat{\bs \beta}_i^\star \neq \widehat{\bs \beta}_i} \leq  \sum_{k = 1}^K\PP{\widehat{\beta}_{ik} < 0} \overset{G \to \infty}{\to} 0$$
where the last limit is due to the consistency of $\widehat \beta_{ik}$. This indicates that $\sqrt G(\widehat{\bs \beta}_i^\star - \widehat{\bs \beta}_i) \overset{p}{\to} \bs 0$, thus 
$$\sqrt G(\widehat {\bs \beta}_i^\star - \bs \beta_i) \overset{d}{\to} \mathcal{N}(\bs 0, \bs \Omega^{-1}\bs \Sigma_i \bs \Omega^{-1})$$

Finally, the cell type proportions $\widehat{\bs p}_i = \bs g(\widehat{\bs\beta}^\star_i)$ is the standardized $\widehat{\bs\beta}^\star_i$. Using the Delta method, we have 
\begin{equation}
    \sqrt{G}(\widehat{\bs p}_i - \bs p_i)\overset{d}{\to} N(\bs 0, \nabla\bs g( \bs\beta_i)\bs\Omega^{-1}\bs\Sigma_i\bs\Omega^{-1}\nabla\bs g(\bs\beta_i)^\top).
\end{equation}
\end{proof}

\begin{proof}[Proof of Corollary~ \ref{cor:CLT_softplus}]

Given the CLT of $\widehat{\bs\beta}_i$ in \eqref{eq:CLT_beta}, and the fact that 
$\widehat{\bs p}_i^{(a)} = \bs g(\widehat{\bs\beta}_i^{(a)}) = \bs g\circ\bs h(\widehat{\bs\beta}_i)$ is smooth function of $\widehat{\bs\beta}_i$ with the Jacobian matrix 
$ \nabla \bs g\circ\bs h(\bs \beta_i) = \nabla \bs g(\bs\beta_i^{(a)})\bs \Gamma.$
Thus, we complete the proof using the Delta method.

\end{proof}

\subsubsection{Asymptotic Normality of  $\mathbf{A}_N$}

\label{subsec:proof-of-multiple-individual-clt}

For simplicity, we use $\bs p_i^0$ and $\bs h^0(\cdot)$ to denote $\bs p_{i,1:(K-1)}$ and $\bs h(\cdot)_{1:(K-1)}$.
\begin{assumption}[Regularity Conditions for Asymptotic Normality]
\label{assmp:for-clt}
The following holds for the link function $\bs h$ and the estimating equation \eqref{eq:estimating-equation-for-A}:
\begin{enumerate}[label=\alph*.]
\item Covariates $\bs f_i$ satisfy
\begin{enumerate}[label=(\roman*)]
\item For any $\bs b,~\bs A$, it holds that
\[
\frac{1}{N}\sum_{i=1}^N\left[\left\{\bs p_i^0-\bs h^0(\bs b+\bs A^\top\bs f_i)\right\}\tilde{\bs f}_i^\top\right]\overset{p}{\to}\bs L(\bs A,\bs b),
\]
and $(\bs b_0,\bs A_0)$ is its unique root such that $\bs L(\bs A,\bs b)=\bs 0$.
\item $\bs D=\lim_{N\to\infty}\frac{1}{N}\sum_{i=1}^N\{(\tilde{\bs f}_i\tilde{\bs f}_i^\top)\otimes \bs D_i\}\succ\bs 0$, where $\bs D_i=\mathrm{Cov}(\bs p_i^0)$,
\item there exists a constant $C_1$ such that $\max_i\|\bs f_i\|_2\le C_1$
\end{enumerate}
\item The function $h^0$ is continuously differentiable. Its derivative $\dot{h}^0$ satisfies that 
$$\bs L_{\bs B_0}=\lim_{N\to\infty}\frac{1}{N}\sum_{i=1}^N\{(\tilde{\bs f}_i\tilde{\bs f}_i^\top)\otimes \dot{h}^0(\bs b_0+\bs A_0^\top\bs f_i)\}$$ is invertible.
\item The equation $\bs L_N(\bs A, \bs b; \bs P) = \bs 0$ has a unique solution for any $\bs P$.
\end{enumerate}
\end{assumption}

Specifically, the asymptotic normality of $\bs A_N$ is as follows:
\begin{theorem}
\label{lemma:An-clt}
    Under Assumptions~\ref{assmp:random_prop} and \ref{assmp:for-clt}, when $N\to \infty$, we have
    $$\sqrt N\vect{({\bs A}_N^\top - \bs A_0^\top)} \overset{d}{\to} \mathcal{N}\left(\bs 0,(\bs L_{\bs B_0}^{-1}\bs D \bs L_{\bs B_0}^{-\top})_{I_{\bs A}\times I_{\bs A}}\right),$$
where $I_{\bs A}=\{2,\dots,S+1\}$.
\end{theorem}

\begin{proof}
For simplicity, we define $\bs B=(\bs b,\bs A^\top)$, and $\bs B_N$ and $\bs B_0$ are defined correspondingly. Then $\bs L_N(\bs A, \bs b;\bs P)$ is rewritten as $\bs L_N(\bs B;\bs P)$, and $\bs L(\bs A,\bs b)=\bs L(\bs B)$.

First, we prove the consistency. By Assumption~\ref{assmp:for-clt}, $\bs L_N(\bs B;\bs P)\overset{p}{\to}\bs L(\bs B)$.
Since $\bs h$ is continuous, $\bs L_N(\bs B;\bs P)$ is continuous. 
Moreover, since both $\bs L_N(\bs B;\bs P)=\bs 0$ and $\bs L(\bs B)=\bs 0$ have unique roots, by Lemma 5.10 of \citet{van2000asymptotic}, $\bs B_N\overset{p}{\to}\bs B_0$ as $N\to\infty$. That is, $\bs B_N-\bs B_0 = o_p(1)$.

Then we do Taylor expansion at $\bs B_0$,
\begin{align*}
\bs 0=\vect\left\{\bs L_N(\bs B_N;\bs P)\right\}
=\vect\left\{\bs L_N(\bs B_0;\bs P)\right\}+\dot{\bs L}_N(\bs B_0)\vect\left(\bs B_N-\bs B_0\right)+O\left(\left\|\bs B_N-\bs B_0\right\|_F^2\right),
\end{align*}
where
\[
\dot{\bs L}_N(\bs B_0)=-\frac{1}{N}\sum_{i=1}^N\tilde{\bs f}_i\tilde{\bs f}_i^\top\otimes \dot{\bs h^0}(\bs B_0\tilde{\bs f}_i)
\]
and it is independent of $\bs P$.
By Assumption~\ref{assmp:for-clt}b, $\bs L_{\bs B_0}=-\lim_{N\to\infty}\dot{\bs L}_N(\bs B_0)$ exists and it is invertible.
Since
\[
\lim_{N\to\infty}\dot{\bs L}_N(\bs B_0)+o(1)=-\bs L_{\bs B_0},
\]
for sufficiently large $N$s,
\[
\sqrt{N}\vect\left(\bs B_N-\bs B_0\right)=-\left\{\dot{\bs L}_N(\bs B_0)+o_p(1)\right\}^{-1}\sqrt{N}\vect\left\{\bs L_N(\bs B_0;\bs P)\right\}.
\]
Now we only need to investigate the asymptotic normality of the following quantity:
\[
\vect\left\{\bs L_N(\bs B_0;\bs P)\right\}=\frac{1}{N}\sum_{i=1}^N\vect\left[\left\{\bs p_i^0-\bs h^0(\bs B_0 \tilde{\bs f}_i)\right\}\tilde{\bs f}_i^\top\right].
\]
For each term, the variance is
\begin{align*}
&\mathrm{Cov}\left(\vect\left[\left\{\bs p_i^0-\bs h^0(\bs B_0 \tilde{\bs f}_i)\right\}\tilde{\bs f}_i^\top\right]\right)
=\mathrm{Cov}\left\{\vect\left(\bs p_i^0\tilde{\bs f}_i^\top\right)\right\}
=\mathrm{Cov}\left\{(\tilde{\bs f}_i\otimes \bs I_K)\bs p_i^0\right\}\\
&=(\tilde{\bs f}_i\otimes \bs I_K)\bs D_i (\tilde{\bs f}_i^\top\otimes \bs I_K)
=\tilde{\bs f}_i\tilde{\bs f}_i^\top\otimes\bs D_i.
\end{align*}
Then the variance of $\sqrt{N}\vect\left\{\bs L_N(\bs B_0;\bs P)\right\}$ is
\[
\frac{1}{N}\sum_{i=1}^N\tilde{\bs f}_i\tilde{\bs f}_i^\top\otimes\bs D_i\to\bs D\succ0.
\]
For any $\bs t\in\mathbb{R}^{K(S+1)}$, we check the asymptotic normality of $\sqrt{N}\bs t^\top\vect\left\{\bs L_N(\bs B_0;\bs P)\right\}$. The Lyapunov condition is
\begin{align*}
&\lim_{N\to\infty}\frac{\sum_{i=1}^N\mathbb{E}\left(\bs t^\top\vect\left[\left\{\bs p_i^0-\bs h^0(\bs B_0 \tilde{\bs f}_i)\right\}\tilde{\bs f}_i^\top\right]\right)^{2+\delta}}{\left[\sum_{i=1}^N\bs t^\top\mathrm{Cov}\left(\vect\left[\left\{\bs p_i^0-\bs h^0(\bs B_0 \tilde{\bs f}_i)\right\}\tilde{\bs f}_i^\top\right]\right)\bs t\right]^{1+\frac{\delta}{2}}}\\
&\le\lim_{N\to\infty}\frac{\|\bs t\|_2^{2+\delta}\sum_{i=1}^N\mathbb{E}\left\|\vect\left[\left\{\bs p_i^0-\bs h^0(\bs B_0 \tilde{\bs f}_i)\right\}\tilde{\bs f}_i^\top\right]\right\|_2^{2+\delta}}{\left[\bs t^\top\left\{\sum_{i=1}^N\mathrm{Cov}\left(\vect\left[\left\{\bs p_i^0-\bs h^0(\bs B_0 \tilde{\bs f}_i)\right\}\tilde{\bs f}_i^\top\right]\right)\right\}\bs t\right]^{1+\frac{\delta}{2}}}\\
&=\lim_{N\to\infty}\frac{\|\bs t\|^{2+\delta}\sum_{i=1}^N\mathbb{E}\left\|\left(\tilde{\bs f}_i\otimes\bs I_K\right)\left(\bs p_i^0-\bs h^0(\bs B_0 \tilde{\bs f}_i)\right)\right\|_2^{2+\delta}}{\left(\bs t^\top\left[\sum_{i=1}^N\tilde{\bs f}_i\tilde{\bs f}_i^\top\otimes\bs D_i\right]\bs t\right)^{1+\frac{\delta}{2}}}\\
&\le\lim_{N\to\infty}\frac{1}{N^{\delta/2}}\frac{\frac{1}{N}
\sum_{i=1}^N\|\tilde{\bs f}_i\|_2^{2+\delta}\mathbb{E}\left\|\bs p_i^0-\bs h^0(\bs B_0\tilde{\bs f}_i)\right\|_2^{2+\delta}}{\sigma_{\mathrm{min}}\left(\frac{1}{N}\sum_{i=1}^N\tilde{\bs f}_i\tilde{\bs f}_i^\top\otimes\bs D_i\right)^{2+\delta}}
=0,
\end{align*}
where the last equality is because $\|\bs f_i\|_2$ are uniformly upper-bounded (see Assumption~\ref{assmp:for-clt}a) and the fact that $\|\bs p_i\|_1=1$.
Since $\bs t$ is arbitrary, with the Cram\'er-Wold theorem, it follows that
\begin{equation}
\label{eq:LN-clt}
\sqrt{n}\vect\left\{\bs L_N(\bs B_0;\bs P)\right\}
\overset{d}{\to}\mathcal{N}\left(\bs 0, \bs D\right),
\end{equation}
where $\bs D=\lim_{N\to\infty}\frac{1}{N}\sum_{i=1}^N(\tilde{\bs f}_i\tilde{\bs f}_i^\top)\otimes\bs D_i$. 
Hence
\begin{equation}
\label{eq:b-clt}
\sqrt{n}\vect\left(\bs B_N-\bs B_0\right)
\overset{d}{\to}\mathcal{N}\left(\bs 0, \bs L_{\bs B_0}^{-1}\bs D \bs L_{\bs B_0}^{-\top}\right).
\end{equation}
Then, by the definition of $\bs B$, we have
$$\sqrt N\vect{(\widehat {\bs A}^\top - \bs A^\top)} \overset{d}{\to} \mathcal{N}\left(\bs 0,(\bs L_{\bs B_0}^{-1}\bs D \bs L_{\bs B_0}^{-\top})_{I_{\bs A}\times I_{\bs A}}\right),$$
where $I_{\bs A}=\{2,\dots,S+1\}$ is the set of indices corresponding to $\bs A$.
\end{proof}

\subsubsection{Proof of Theorem~\ref{thm:two_group}}

 We investigate the difference $\widehat{\bs A}-\bs A_N$ in the following part. We start with several preliminary lemmas and ends with the proof of Theorem~\ref{thm:two_group}.
\begin{lemma}\label{lem:normalizing}
Let the normalization function be $\bs g(\bs z) = \frac{\bs z}{\bs z^\top \bs 1}$ where $\bs z$ is some non-negative k-dimensional vector. Then for any $\bs z$ and any $\bs z_0$ satisfying $\|\bs z_0\|_1 > \delta$ with some $\delta > 0$, there exists some $L_0 > 0$ satisfying
$$\|\bs g(\bs z) - \bs g(\bs z_0)\|_2 \leq L_0\|\bs z - \bs z_0\|_2$$
and some $L_1 >0$ satisfying that
$$\|\bs g(\bs z) - \bs g(\bs z_0) - \nabla \bs g(\bs z_0)(\bs z - \bs z_0)\|_2 \leq L_1\|\bs z - \bs z_0\|_2^2$$
\end{lemma}

\begin{proof}
By definition, we can calculate that 
$$\nabla \bs g(\bs z_0) = \frac{1}{\bs z_0^\top \bs 1}\left(\bs I - \frac{\bs z_0\bs 1^\top}{\bs z_0^\top \bs 1}\right)$$
where $\bs I$ is a $k \times k$ identity matrix. Then we have
\begin{equation}\label{eq:taylor}
    \bs g(\bs z) - \bs g(\bs z_0) - \nabla \bs g(\bs z_0)(\bs z - \bs z_0) =
    \frac{(\bs z - \bs z_0)^\top \bs 1}{\bs z_0^\top \bs 1}\left(\bs g(\bs z_0) - \bs g(\bs z)\right).
\end{equation}
In general, notice that for any two k-dimensional non-negative  vectors $\bs x$ and $\bs y$
\begin{align*}
    \|\bs g(\bs x) - \bs g(\bs y)\|_2^2 & = \frac{\sum_l(\sum_k y_k x_l - \sum_k x_k y_l)^2}{(\sum_k x_k)^2(\sum_k y_k)^2}\\
    & = \frac{\sum_l\big(\sum_k y_k (x_l - y_l) + (\sum_k y_k- \sum_k x_k) y_l\big)^2}{(\sum_k x_k)^2(\sum_k y_k)^2}\\
   & \leq 2\frac{(\sum_k y_k)^2\sum_l (x_l - y_l)^2 + (\sum_k y_k- \sum_k x_k)^2 \sum_ly_l^2}{(\sum_k x_k)^2(\sum_k y_k)^2}
\end{align*}
Notice that $\sum_l y_l^2 < (\sum_k y_k)^2$ as $\bs y$ is non-negative and $(\sum_k y_k- \sum_k x_k)^2\leq k\sum_k(y_k - x_k)^2$, we have

$$\|\bs g(\bs x) - \bs g(\bs y)\|_2^2 \leq \|\bs x - \bs y\|_2^2 \frac{2(1 + k)}{\|\bs x\|_1}$$

As $\|\bs z_0\|_1$ is bounded below, there exists some $L_0 > 0$ so that  $\|\bs g(\bs z_0) - \bs g(\bs z)\|_2\leq L_0 \|\bs z - \bs z_0\|_2$.

As $\|\bs z - \bs z_0\|_1 \leq \sqrt k \|\bs z - \bs z_0\|_2$, from \eqref{eq:taylor} we also have
\begin{equation*}
    \|\bs g(\bs z) - \bs g(\bs z_0) - \nabla \bs g(\bs z_0)(\bs z - \bs z_0)\|_2 \leq 
    \frac{\|\bs z - \bs z_0)\|_1}{\|\bs z_0\|_1}L_0\|\bs z - \bs z_0\|_2\leq \frac{L_0\sqrt k}{\delta}\|\bs z - \bs z_0\|_2^2
\end{equation*}
\end{proof}

\begin{lemma}\label{lem:phi_unif}
Under Assumptions~\ref{asp:homo_popu}-\ref{assmp:variance} and Assumptions~\ref{assmp:group},  if $N/G^2 \to 0$ when $G \to \infty$, we then have 
$$\frac{1}{N}\sum_{i = 1}^N\big\|\bs \phi(\bs \beta_i)\big\|_2 = O_p(\sqrt G),
\quad \frac{1}{N}\sum_{i = 1}^N\big\|\bs \phi(\bs \beta_i)\big\|_2^2 = O_p(G), \quad \max_{i}\big\|\bs \phi(\bs \beta_i)\big\|_2 = o_p(G)$$

\end{lemma}

\begin{proof}[Proof of Lemma~\ref{lem:phi_unif}]
Recall that 
$\bm \phi(\bm \beta_i)= \widehat{\bm U}^\top \bm W \bm \epsilon_i' - \bm H \bm{\beta}_i$,
so 
$$\big\|\bs \phi(\bs \beta_i)\big\|_2 \leq  \left\|\widehat{\bm U}^\top \bm W \bm \epsilon_i'\right\|_2 + \left\|\bm H\right\|_2\|\bm{\beta}_i\|_2$$
$$\big\|\bs \phi(\bs \beta_i)\big\|_2^2 \leq  2\left\|\widehat{\bm U}^\top \bm W \bm \epsilon_i'\right\|_2^2 + 2\left\|\bm H\right\|_2^2\|\bm{\beta}_i\|_2^2$$
Using Lemma~\ref{lem:con1} and Lemma~\ref{lem:gap}, we have $\|\bm H\|_2 = O_p\left(\sqrt G\right)$. Also, as $\bs \beta_i$ is always non-negative,  $\|\bm \beta_i\|_2^2 \leq \|\bs \beta_i\|_1^2= \gamma_i^2$ since by definition $\bs \beta_i = \gamma_i \bs p_i$.  Under Assumption~\ref{assmp:group}e,
\[
\frac{1}{N}\sum_{i = 1}^N\|\bs\beta_i\|_2 \leq \frac{1}{N}\sum_{i = 1}^N\gamma_i =  O_p(1),
\]
\[
\frac{1}{N}\sum_{i = 1}^N\|\bs\beta_i\|_2^2 \leq \frac{1}{N}\sum_{i = 1}^N\gamma_i^2 =  O_p(1),
\]
\[
\max_i \|\bs\beta_i\|_2 \leq \max_i \gamma_i = O_p(1).
\]
So $
\left\|\bm H\right\|_2\max_i\|\bm{\beta}_i\|_2 = o_p\left(G\right)$, 
$
\left\|\bm H\right\|_2\frac{1}{N}\sum_{i = 1}^N\|\bm{\beta}_i\|_2 = O_p\left( \sqrt G\right)$ and $\left\|\bm H\right\|_2^2\frac{1}{N}\sum_{i = 1}^N\|\bm{\beta}_i\|_2^2 = O_p\left(G\right)$.

Next, consider $\widehat{\bm U}^{\star \top} \bm W \bm \epsilon_i' = \sum_gw_g\epsilon_{ig}'\widehat{\bm \mu}_g^\star$ where $\widehat{\bm U}^
\star$ is defined as in Lemma~\ref{lem:con1}. 
Under Assumptions \ref{assumption_consistency}bc, 
$$\max_{i, g}\EE{\|w_g\epsilon_{ig}'\widehat{\bm \mu}_g^\star\|_2^4}\leq \max_{i, g}\left(w_g^4\EE{\epsilon_{ig}'^{4}}\EE{\|\widehat{\bm \mu}_g^\star\|_2^4}\right) \leq c$$ for some constant $c$. So all lower moments are also uniformly bounded. In addition, we have the inequality 
$$\max_i\EE{\left\|\widehat{\bm U}^{\star \top} \bm W \bm \epsilon_i'\right\|_2^4} \leq \max_i8\left(\EE{\left\|\sum_{g \in \mathcal{V}} w_g\epsilon_{ig}'\widehat{\bm \mu}_g^\star\right\|_2^4} + \EE{\left\|\sum_{g \in \mathcal{V}^c} w_g\epsilon_{ig}'\widehat{\bm \mu}_g^\star\right\|_2^4}\right).$$
Since $|\mathcal{V}^c| = o(\sqrt G)$, so we have $\max_i\EE{\left\|\sum_{g \in \mathcal{V}^c} w_g\epsilon'_{ig}\widehat{\bm \mu}_g^\star\right\|_2^4} = o(G^2)$. Also, under the sparse dependence structure in  Assumption~\ref{assmp:dependence} and the fact that $\bs \epsilon'_{i}$ is mutually independent from $\widehat{\bm U}^\star$ with $\EE{\bm\epsilon'_{i}} = \bm 0$, we have
\begin{align*}
   & \EE{\left\|\sum_{g \in \mathcal{V}}  w_g\epsilon'_{ig}\widehat{\bm \mu}_g^\star\right\|_2^4} \\
   = & 
\sum_{g \in \mathcal{V}} \EE{w_g^4\epsilon_{ig}'^4}\EE{\left\| \widehat{\bm \mu}_g^\star\right\|_2^4} + 
4\sum_{g_1 \neq g_2 \in \mathcal{V}} \EE{w_{g_1}^3w_{g_2}\epsilon_{ig_1}'^3\epsilon'_{ig_2}}\EE{\left\| \widehat{\bm \mu}_{g_1}^\star\right\|_2^2\widehat{\bm \mu}_{g_1}^{\star \top}\widehat{\bm \mu}_{g_2}^\star} 
\\ & +  \sum_{g_1 \neq g_2 \in \mathcal{V}} \EE{w_{g_1}^2w_{g_2}^2\epsilon_{ig_1}'^2\epsilon_{ig_2}'^2}
\EE{\left\| \widehat{\bm \mu}_{g_1}^\star\right\|_2^2\left\| \widehat{\bm \mu}_{g_2}^\star\right\|_2^2 + 2(\widehat{\bm \mu}_{g_1}^{\star \top}\widehat{\bm \mu}_{g_2}^\star)^2}  \\
& + 
2\sum_{g_1 \neq g_2 \neq g_3 \in \mathcal{V}} \EE{w_{g_1}^2w_{g_2}w_{g_3}\epsilon_{ig_1}'^2\epsilon'_{ig_2}\epsilon'_{ig_3}}\EE{\left\| \widehat{\bm \mu}_{g_1}^\star\right\|_2^2\widehat{\bm \mu}_{g_2}^{\star \top}\widehat{\bm \mu}_{g_3}^\star + 2\widehat{\bm \mu}_{g_1}^{\star \top}\widehat{\bm \mu}_{g_2}^\star\widehat{\bm \mu}_{g_1}^{\star \top}\widehat{\bm \mu}_{g_3}^\star}
\\ & +
\sum_{g_1 \neq g_2 \neq g_3\neq g_4 \in \mathcal{V}} \EE{w_{g_1}w_{g_2}w_{g_3}w_{g_4}\epsilon'_{ig_1}\epsilon'_{ig_2}\epsilon'_{ig_3}\epsilon'_{ig_4}}\EE{\widehat{\bm \mu}_{g_1}^{\star \top}\widehat{\bm \mu}_{g_2}^\star\widehat{\bm \mu}_{g_3}^{\star \top}\widehat{\bm \mu}_{g_4}^\star} \\
& = O(G) + O(G) + O(G^2) + O(G^2) + O(G^2) = O(G^2)
\end{align*}
The last term has an order of $O(G^2)$ as $\EE{\epsilon'_{ig_1}\epsilon'_{ig_2}\epsilon'_{ig_3}\epsilon'_{ig_4}}\neq 0$ only when every node has an edge. For any selected node $g_1$, there is at least one other node $g_2$ that has an edge with $g_1$ (so there are $O(1)$ choices of $g_2$), and the other two nodes are either connected to $g_1$ and $g_2$ (at most $O(1)$ choices), or are connected with each other (at most $O(G)$ choices).  
Similarly, we can also show 
$\max_i\EE{\left\|\widehat{\bm U}^{\star \top} \bm W \bm \epsilon'_i\right\|_2^4} = O(G^2)$. At the same time,  we can also obtain 
$\max_i\EE{\left\|\widehat{\bm U}^{\star \top} \bm W \bm \epsilon'_i\right\|_2^2} = O(G)$ and $\max_i\EE{\left\|\widehat{\bm U}^{\star \top} \bm W \bm \epsilon'_i\right\|_2} = O(\sqrt G)$.
Then for any $\epsilon > 0$,
$$\PP{\frac{1}{G}\max_i\left\|\widehat{\bm U}^{\star \top} \bm W \bm \epsilon'_i\right\|_2 > \epsilon}
\leq \sum_{i = 1}^N \frac{\EE{\left\|\widehat{\bm U}^{\star \top} \bm W \bm \epsilon'_i\right\|_2^4}}{G^4\epsilon^4} = O\left(\frac{N}{G^2}\right) \to 0$$
and for any $\Delta > 0$ and any $G$, there is a constant $\tilde C$ for when $N$ is sufficiently large
$$\PP{\frac{1}{\sqrt GN}\sum_{i = 1}^N\left\|\widehat{\bm U}^{\star \top} \bm W \bm \epsilon'_i\right\|_2 > \Delta }
\leq \frac{\EE{\sum_{i}\left\|\widehat{\bm U}^{\star \top} \bm W \bm \epsilon'_i\right\|_2 }}{ \Delta N\sqrt G}\leq \frac{\max_i\EE{\left\|\widehat{\bm U}^{\star \top} \bm W \bm \epsilon'_i\right\|_2 }}{ \Delta \sqrt G} \leq \frac{\tilde C}{\Delta}.$$
$$\PP{\frac{1}{GN}\sum_{i = 1}^N\left\|\widehat{\bm U}^{\star \top} \bm W \bm \epsilon'_i\right\|_2^2 > \Delta }
\leq \frac{\EE{\sum_{i}\left\|\widehat{\bm U}^{\star \top} \bm W \bm \epsilon'_i\right\|_2^2 }}{ \Delta NG}\leq \frac{\max_i\EE{\left\|\widehat{\bm U}^{\star \top} \bm W \bm \epsilon'_i\right\|_2^2 }}{ \Delta G}  \leq \frac{\tilde C}{\Delta}.$$
Thus, we have $\max_i\big\|\widehat{\bm U}^{\star \top} \bm W \bm \epsilon'_i\big\|_2 = o_p(G)$, $\sum_{i = 1}^N\big\|\widehat{\bm U}^{\star \top} \bm W \bm \epsilon'_i\big\|_2/N = O_p(\sqrt G)$ and the relationship $\sum_{i = 1}^N\big\|\widehat{\bm U}^{\star \top} \bm W \bm \epsilon'_i\big\|_2^2/N = O_p(G)$.

Finally, consider the term $\widehat{\bm U}^\top \bm W \bm \epsilon'_i - \widehat{\bm U}^{\star \top} \bm W \bm \epsilon'_i$. 
Notice that, 
\begin{align*}
  \widehat{\bm U}^\top \bm W \bm \epsilon'_i - \widehat{\bm U}^{\star \top} \bm W \bm \epsilon'_i 
   =  \frac{1}{M}\sum_{j = 1}^M\left(\frac{\rr\gamma_{j}}{\widehat\gamma_{j}}- 1\right)\sum_{g = 1}^G w_g \epsilon'_{ig}\frac{\rr{\bm z}_{jg}}{\rr\gamma_{j}}.
\end{align*}
Similar to our previous argument, we also have 
\begin{equation*}
\max_i\left\|\sum_{g = 1}^G w_g \epsilon'_{ig}\frac{\rr{\bm z}_{jg}}{\rr\gamma_{j}}\right\|_2 = o_p(G)
\end{equation*}
and 
\begin{equation*}
    \frac{1}{N}\sum_{i = 1}^N\left\|\sum_{g = 1}^G w_g \epsilon'_{ig}\frac{{\bm z}_{jg}}{\rr\gamma_{j}}\right\|_2 = O_p(\sqrt G), \quad\frac{1}{N}\sum_{i = 1}^N\left\|\sum_{g = 1}^G w_g \epsilon'_{ig}\frac{\rr{\bm z}_{jg}}{\rr\gamma_{j}}\right\|_2^2 = O_p(G).
\end{equation*}
As $\widehat \gamma_j - \rr\gamma_j = O_p(1/\sqrt G)$ and $M$ is fixed when $G \to\infty$, we obtain 
$$\max_i\left\|\widehat{\bm U}^\top \bm W \bm \epsilon'_i - \widehat{\bm U}^{\star \top} \bm W \bm \epsilon'_i\right\|_2 = o_p(\sqrt G),$$
$$\frac{1}{N}\sum_{i = 1}^N\left\|\widehat{\bm U}^\top \bm W \bm \epsilon'_i - \widehat{\bm U}^{\star \top} \bm W \bm \epsilon'_i\right\|_2 = O_p(1), \quad
\frac{1}{N}\sum_{i = 1}^N\left\|\widehat{\bm U}^\top \bm W \bm \epsilon'_i - \widehat{\bm U}^{\star \top} \bm W \bm \epsilon'_i\right\|_2^2 = O_p(1)
.$$
Combining all above, we get $\sum_{i = 1}^N\big\|\widehat{\bm U}^\top \bm W \bm \epsilon'_i\big\|_2/N = O_p(\sqrt G)$, $\sum_{i = 1}^N\big\|\widehat{\bm U}^\top \bm W \bm \epsilon'_i\big\|_2^2/N = O_p(G)$ and $\max_i\big\|\widehat{\bm U}^\top \bm W \bm \epsilon'_i\big\|_2 = o_p(G)$, 
and we prove the lemma.
\end{proof}

\begin{proof}[Proof of Theorem~\ref{thm:two_group}]
We start with analyzing the difference
\[
\frac{1}{N}\sum_{i = 1}^N\big\|\bs{\widehat p}_i - \bs p_i\big\|_2
\]
under the condition $N/G^2\to0$ as $G\to\infty$.
First, recall that $\widehat {\bs \beta}_i - \bs \beta_i = \widehat{\bs \Omega}^{-1}\bs\phi(\bs \beta_i)/G$. Also, notice that from Theorem~\ref{thm:consistency}, $\widehat{\bs \Omega}^{-1} \overset{p}{\to} \bs \Omega^{-1}$, indicating $\|\widehat{\bs \Omega}^{-1}\|_2 = O_p(1)$. Thus, using Lemma~\ref{lem:phi_unif} we have
$$\frac{1}{N}\sum_{i = 1}^N\big\|\bs{\widehat \beta}_i - \bs \beta_i\big\|_2 \leq \|\widehat{\bs \Omega}^{-1}\|_2\frac{1}{N}\sum_{i = 1}^N\big\|\bs \phi(\bs \beta_i)\big\|_2/G = O_p(G^{-1/2})$$
and 
$$\max_i\big\|\bs{\widehat \beta}_i - \bs \beta_i\big\|_2 \leq \|\widehat{\bs \Omega}^{-1}\|_2\max_i\big\|\bs \phi(\bs \beta_i)\big\|_2/G = o_p(1).$$
Then, using Lemma~\ref{lem:normalizing}, we also have
$$\frac{1}{N}\sum_{i = 1}^N\big\|\bs g(\bs{\widehat \beta}_i) - \bs g(\bs \beta_i)\big\|_2 \leq \frac{L}{N}\sum_{i = 1}^N\big\|\bs{\widehat \beta}_i - \bs \beta_i\big\|_2  = O_p(G^{-1/2})$$
as Assumption~\ref{assmp:group}e guarantees that $\min_i\|\bs \beta_i\|_1 = \min_i \gamma_i\geq C_3$. In addition, for any $\epsilon > 0$ and truncated estimator $\bs{\widehat \beta}_i^\star$,
\begin{align*}
& \mathbb{P}\left(\frac{\sqrt{G}}{N}\sum_{i = 1}^N\big\|\bs g(\bs{\widehat \beta}_i^\star) - \bs g(\widehat{\bs \beta}_i)\big\|_2 > \epsilon\right) \\
&\leq \mathbb{P}\left(\frac{\sqrt{G}}{N}\sum_{i = 1}^N\big\|\bs g(\bs{\widehat \beta}_i^\star) - \bs g(\widehat{\bs \beta}_i)\big\|_2 \neq 0\right) \\
&=\mathbb{P}\left(\cup_{i = 1}^N \{\bs{\widehat \beta}_i^\star \neq \bs{\widehat \beta}_i\}\right) = \mathbb{P}\left(\cup_{i = 1}^N \cup_{k = 1}^K\{{\widehat \beta}_{ik} < 0\}\right)\\
&\leq \mathbb{P}\left(\max_i\|\bs{\widehat \beta}_i - \bs{ \beta}_i\|_2 > C_2\}\right) \overset{G \to \infty}{\to} 0
\end{align*}
So $\sum_{i = 1}^N\big\|\bs g(\bs{\widehat \beta}_i^\star) - \bs g(\widehat{\bs \beta}_i)\big \|_2/N = o_p(G^{-1/2})$ and thus,
\begin{equation}
\label{eq:-est-p-true-p-diff}
\frac{1}{N}\sum_{i = 1}^N\big\|\bs{\widehat p}_i - \bs p_i\big\|_2
=\frac{1}{N}\sum_{i = 1}^N\big\|\bs g(\bs{\widehat \beta}_i^\star) - \bs g(\bs \beta_i)\big\|_2  
= O_p(G^{-1/2}).
\end{equation}

Then, we investigate the consistency of $\widehat{\bs B}$. For any $\bs B$,
\begin{align*}
    \bs L_N(\bs B;\widehat{\bs P}) = \bs L_N(\bs B;\bs P)+\frac{1}{N}\sum_{i=1}^N\left\{\left(\widehat{\bs p^0_i}-\bs p^0_i\right)\tilde{\bs f}_i^\top\right\}.
\end{align*}
Since $\|\tilde{\bs f}_i\|_2$ are uniformly bounded above by Assumption~\ref{assmp:group}a and \eqref{eq:-est-p-true-p-diff} holds, further with Assumption~\ref{assmp:group}c,
\[
\bs L_N(\bs B;\widehat{\bs P}) = \bs L(\bs B)+o_p(1)+O_p(1/\sqrt{G}).
\]
Since $h$ is continuous and both $\bs L_N(\bs B;\widehat{\bs P})=\bs 0$ and $\bs L(\bs B)=\bs 0$ have unique roots by Assumption~\ref{assmp:group}a-b, 
\[
\widehat{\bs B}\overset{p}{\to}\bs B_0
\]
following Lemma 5.10 of \citet{van2000asymptotic}. With this, we do Taylor expansion of $\bs L_N(\bs B;\widehat{\bs P})$ at $\bs B_0$,
\begin{align*}
&0=\vect\left\{\bs L_N(\widehat{\bs B};\widehat{\bs P})\right\}
=\vect\left\{\bs L_N(\bs B_0;\widehat{\bs P})\right\}+\dot{\bs L}_N(\bs B_0)\vect\left(\widehat{\bs B}-\bs B_0\right)+O(\|\widehat{\bs B}-\bs B_0\|_F^2),
\end{align*}
where
\[
\dot{\bs L}_N(\bs B_0)=-\frac{1}{N}\sum_{i=1}^N\left\{(\tilde{\bs f}_i\tilde{\bs f}_i^\top)\otimes \dot{\bs h^0}(\bs B_0\tilde{\bs f}_i)\right\}
\]
and it is free of $\bs P$.
By Assumption~\ref{assmp:group}d, $\bs L_{\bs B_0}=\lim_{N\to\infty}\dot{\bs L}_N(\bs B_0)$ exists and it is invertible.
Since 
\begin{align*}
\lim_{N\to\infty}\dot{\bs L}_N(\bs B_0)+o_p(1)=-\bs L_{\bs B_0},
\quad\widehat{\bs B}\overset{p}{\to}\bs B_0,
\end{align*}
we have
\[
\vect\left(\widehat{\bs B}-\bs B_0\right)
=-\left\{\dot{\bs L}_N(\bs B_0)+o_p(1)\right\}^{-1}\vect\left\{\bs L_N(\bs B_0;\bs P)+L_N(\bs B_0;\widehat{\bs P})-L_N(\bs B_0;\bs P)\right\}.
\]
For sufficiently large $N$s, with equation \eqref{eq:LN-clt} and the fact that
\begin{align*}
\lim_{N\to\infty}\dot{\bs L}_N(\bs B_0)+o_p(1)&=-\bs L_{\bs B_0},
\quad
L_N(\bs B_0;\widehat{\bs P})-L_N(\bs B_0;\bs P)=O_p(1/\sqrt{G}),
\end{align*}
we have
\[
\vect\left(\widehat{\bs B}-\bs B_0\right)
=\left\{\bs L_{\bs B_0}+o_p(1)\right\}^{-1}\vect\left\{\bs L_N(\bs B_0;\bs P)\right\}+O_p(1/\sqrt{G}).
\]
Thus, it holds that
\[
\vect\left(\widehat{\bs B}-\bs B_0\right)=O_p(1/\sqrt{N}+1/\sqrt{G})
\]
We do Taylor expansion again for both $\bs L_N(\bs B;\bs P)$ and $\bs L_N(\bs B;\widehat{\bs P})$ at $\bs B_0$,
\begin{align*}
&0=\vect\left\{\bs L_N(\bs B_N;\bs P)\right\}
=\vect\left\{\bs L_N(\bs B_0;\bs P)\right\}+\dot{\bs L}_N(\bs B_0)\vect\left(\bs B_N-\bs B_0\right)+O\left(\left\|\bs B_N-\bs B_0\right\|_F^2\right),\\
&0=\vect\left\{\bs L_N(\widehat{\bs B};\widehat{\bs P})\right\}
=\vect\left\{\bs L_N(\bs B_0;\widehat{\bs P})\right\}+\dot{\bs L}_N(\bs B_0)\vect\left(\widehat{\bs B}-\bs B_0\right)+O(\|\widehat{\bs B}-\bs B_0\|_F^2)
\end{align*}
Reorganizing these two equations we get
\begin{equation}
\label{eq:reorg-taylor-b-diff}
\begin{split}
&\vect\left\{\frac{1}{N}\sum_{i=1}^N\left(\widehat{\bs p^0_i}-\bs p_i^0\right)\tilde{\bs f}_i^\top\right\}
=\dot{\bs L}_N(\bs B_0)\vect\left(\bs B_N-\widehat{\bs B}\right)+O(\|\bs B_N-\bs B_0\|^2_F)+O(\|\widehat{\bs B}-\bs B_0\|^2_F)\\
&\Rightarrow ~\vect\left(\bs B_N-\widehat{\bs B}\right)
=\left(\dot{\bs L}_N(\bs B_0)\right)^{-1}\vect\left\{\frac{1}{N}\sum_{i=1}^N\left(\widehat{\bs p^0_i}-\bs p_i^0\right)\tilde{\bs f}_i^\top\right\}+O_p\left(\frac{1}{N}+\frac{1}{G}\right),    
\end{split}
\end{equation}
where the last equation is because by Assumption~\ref{assmp:group}d, $\bs L_{\bs B_0}=\lim_{N\to\infty}\dot{\bs L}_N(\bs B_0)$ exists and it is invertible.

\textbf{Part I: when $\bm{N/G\to0}$.}
Since by the definition of Frobenius norm,
\[
 \|\widehat{\bs A} - \bs A_N\|_F\le \|\widehat{\bs B} - \bs B_N\|_F,
\]
it suffices to bound the difference between $\bs B_N$ and $\widehat{\bs B}$.

Since by Assumption~\ref{assmp:group}d,
\[
\left\|\left(\bs L_{\bs B_0}\right)^{-1}\right\|_2
\le 1/C_2,
\]
for sufficiently large $N$s, it holds that
\begin{equation}
\label{eq:b-diff-wrt-G}
\begin{split}
&\|\widehat{\bs A} - \bs A_N\|_F
\le\|\widehat{\bs B} - \bs B_N\|_F \\ 
&\le \left\|\left\{-\bs L_{\bs B_0}+o(1)\right\}^{-1}\right\|_2\left\|\frac{1}{N}\sum_{i=1}^N\left(\widehat{\bs p^0_i}-\bs p_i^0\right)\tilde{\bs f}_i^\top\right\|_F+O_p(1/N)+O_p(1/G) \\
&\le \frac{1}{N C_2}\left\|\sum_{i=1}^N\left({\widehat{\bs p_i^0}}-{\bs p_i^0}\right)\tilde{\bs f}_i^\top\right\|_F+O_p(1/N)+O_p(1/G)\\
&\leq \frac{1}{N C_2}\sum_{i=1}^N\left\|\tilde{\bs f}_i\right\|_2\left\|\widehat {\bs p^0_i} - \bs p_i^0\right\|_2 + O_p(1/N)+O_p(1/G)\\
& \leq \frac{\sqrt{C_1+1}}{C_2}\times\frac{1}{N}\sum_{i=1}^N\left\|\widehat{\bs p_i^0} - {\bs p}_i\right\|_2+O_p(1/N)+O_p(1/G)
= O_p(\frac{1}{\sqrt{G}}),    
\end{split}
\end{equation}
where the last inequality is by Assumption~\ref{assmp:group}c.
That is, $\|\widehat{\bs A} - \bs A_N\|_F = O_p(1/\sqrt G) = o_p(1/\sqrt N)$ when $N/G \to 0$.

\textbf{Part II: when $\bm{N/G^2\to0}$ and the global null holds.}
We intend to prove for the conclusion that $\|\widehat{\bs A} - \bs A_N\|_F = O_p(1/G)+O_p(1/N)$ when the global null holds. 

By \eqref{eq:reorg-taylor-b-diff}, when $\bs A_0=\bs 0$,
\begin{align*}
&\vect\left(\bs B_N-\widehat{\bs B}\right)
=\left\{\dot{\bs L}_N(\bs B_0)\right\}^{-1}\vect\left\{\frac{1}{N}\sum_{i=1}^N\left(\widehat{\bs p^0_i}-\bs p_i^0\right)\tilde{\bs f}_i^\top\right\}+O_p(1/N)+ O_p(1/G)\\
&=\left\{\frac{1}{N}\sum_{i=1}^N \tilde{\bs f}_i\tilde{\bs f}_i^\top\otimes\dot{\bs h^0}(\bs b_0)\right\}^{-1}\vect\left\{\frac{1}{N}\sum_{i=1}^N\left(\widehat{\bs p^0_i}-\bs p_i^0\right)\tilde{\bs f}_i^\top\right\}+O_p(1/N)+ O_p(1/G).
\end{align*}
Moreover, by $\sum_{i=1}^N\bs f_i=0$,
\begin{align*}
&\vect\left(\bs B_N-\widehat{\bs B}\right)
=\begin{bmatrix}
\bs b_N-\widehat{\bs b}\\
\vect\left(\bs A_N^\top-\widehat{\bs A}^\top\right)
\end{bmatrix},\\
&
\left\{\frac{1}{N}\sum_{i=1}^N \tilde{\bs f}_i\tilde{\bs f}_i^\top\otimes\dot{\bs h^0}(\bs b_0)\right\}^{-1}
\begin{bmatrix}
    \frac{1}{N}\sum_{i=1}^N(\widehat{\bs p^0_i}-\bs p_i^0)\\
    \vect\left\{\frac{1}{N}\sum_{i=1}^N\left(\widehat{\bs p^0_i}-\bs p_i^0\right)\tilde{\bs f}_i^\top\right\}
\end{bmatrix}\\
&=
\begin{bmatrix}
 \dot{\bs h^0}(\bs b_0)^{-1}  & \bs 0 \\
  \bs 0 & \left\{\frac{1}{N}\sum_{i=1}^N \bs f_i\bs f_i^\top\otimes\dot{\bs h^0}(\bs b_0)\right\}^{-1} 
\end{bmatrix}
\begin{bmatrix}
    \frac{1}{N}\sum_{i=1}^N(\widehat{\bs p^0_i}-\bs p_i^0)\\
    \vect\left\{\frac{1}{N}\sum_{i=1}^N\left(\widehat{\bs p^0_i}-\bs p_i^0\right)\bs f_i^\top\right\}
\end{bmatrix}\\
&=\begin{bmatrix}
    \dot{\bs h^0}(\bs b_0)^{-1}\frac{1}{N}\sum_{i=1}^N(\widehat{\bs p^0_i}-\bs p_i^0)\\
    \left\{\frac{1}{N}\sum_{i=1}^N \bs f_i\bs f_i^\top\otimes\dot{\bs h^0}(\bs b_0)\right\}^{-1} \vect\left\{\frac{1}{N}\sum_{i=1}^N\left(\widehat{\bs p^0_i}-\bs p_i^0\right)\bs f_i^\top\right\}
\end{bmatrix}.
\end{align*}
Since we care about the difference $\bs A_N-\widehat{\bs A}$, in the following part, we consider
\begin{equation}
\label{eq:a-diff-first-term}
\begin{split}
\vect\left(\bs A_N^\top-\widehat{\bs A}^\top\right)
=&\left\{\frac{1}{N}\sum_{i=1}^N \bs f_i\bs f_i^\top\otimes\dot{\bs h^0}(\bs b_0)\right\}^{-1} \vect\left\{\frac{1}{N}\sum_{i=1}^N\left(\widehat{\bs p^0_i}-\bs p_i^0\right)\bs f_i^\top\right\}\\
&\quad+O_p(1/G)+O_p(1/N).      
\end{split}
\end{equation}
By Assumption~\ref{assmp:group}d,
\[
\left\|\left(\frac{1}{N}\sum_{i=1}^N \bs f_i \bs f_i^\top\otimes\dot{\bs h^0}(\bs b_0)\right)^{-1}\right\|_2
=O(1).
\]
Therefore, we ignore this part in the following analysis and only focus on
\begin{equation}
\label{eq:simplified-goal}
\frac{1}{N}\sum_{i=1}^N\left(\widehat{\bs p}_i-\bs p_i\right)\bs f_i^\top
=\frac{1}{N}\sum_{i=1}^N\left\{\bs g(\widehat{\bs \beta}^\star_i)-\bs g(\bs \beta_i)\right\}\bs f_i^\top.
\end{equation}
Applying Lemma~\ref{lem:normalizing} and Assumption~\ref{assmp:group}c,
\begin{equation}\label{eq:decomp}
\frac{1}{N}\sum_{i=1}^N\bs f_i\left\{\bs g(\widehat{\bs \beta}^\star_i)-\bs g(\bs \beta_i)\right\}^\top
=\frac{1}{N}\sum_{i=1}^N \bs f_i\left(\widehat{\bs \beta}^\star_i-\bs \beta_i\right)^\top\nabla\bs g(\bs\beta_i)^\top
+O\left(\frac{1}{N}\sum_{i=1}^N ||\widehat{\bs\beta}_i^\star-\bs\beta_i||_2^2\right)
\end{equation}
Notice that from Theorem~\ref{thm:consistency}, $\widehat{\bs \Omega}^{-1} \overset{p}{\to} \bs \Omega^{-1}$, indicating $\|\widehat{\bs \Omega}^{-1}\|_2 = O_p(1)$. Using this and Lemma~\ref{lem:phi_unif}, we have
$$\frac{1}{N}\sum_{i = 1}^N\big\|\bs{\widehat \beta}_i - \bs \beta_i\big\|_2^2 \leq \|\widehat{\bs \Omega}^{-1}\|_2\frac{1}{N}\sum_{i=1}^N\big\|\bs \phi(\bs \beta_i)\big\|_2^2/G^2 = O_p(1/G)$$
and 
$$\max_i\big\|\bs{\widehat \beta}_i - \bs \beta_i\big\|_2 \leq \|\widehat{\bs \Omega}^{-1}\|_2\max_i\big\|\bs \phi(\bs \beta_i)\big\|_2/G = o_p(1).$$
Also, same as in the proof of Theorem~\ref{thm:two_group}, for any $\epsilon > 0$,
\begin{align*}
     \PP{\frac{G}{N}\sum_{i = 1}^N\big\|\bs{\widehat \beta}_i^\star - \widehat{\bs \beta}_i\big\|_2^2 > \epsilon} 
    \leq 
\PP{\cup_{i = 1}^N \{\bs{\widehat \beta}_i^\star \neq \bs{\widehat \beta}_i\}}
\leq \PP{\max_i\|\bs{\widehat \beta}_i - \bs{ \beta}_i\|_2 > \delta\}} \overset{G \to \infty}{\to} 0
\end{align*}
which indicates that $\frac{1}{N}\sum_{i=1}^N ||\widehat{\bs\beta}_i^\star-\widehat{\bs\beta}_i||_2^2 = o_p(1/G)$. Thus, $\frac{1}{N}\sum_{i=1}^N ||\widehat{\bs\beta}_i^\star-\bs\beta_i||_2^2 = O_p(1/G)$.

Next, notice that
{\small
\[
\frac{1}{N}\sum_{i=1}^N \bs f_i\left(\widehat{\bs \beta}^\star_i-\bs \beta_i\right)^\top\nabla\bs g(\bs\beta_i)^\top
=\frac{1}{N}\sum_{i=1}^N \bs f_i\left(\widehat{\bs \beta}_i-\bs \beta_i\right)^\top\nabla\bs g(\bs\beta_i)^\top
+\frac{1}{N}\sum_{i=1}^N \bs f_i\left(\widehat{\bs \beta}^\star_i-\widehat{\bs \beta}_i\right)^\top\nabla\bs g(\bs\beta_i)^\top
\]
}
As $\|\nabla\bs g(\bs\beta_i)\|_2$s are uniformly bounded above. With Assumption~\ref{assmp:group}a, we have
\[
\frac{1}{N}\sum_{i=1}^N \bs f_i(\bs{\widehat \beta}_i - \widehat{\bs \beta}_i^{\star})^\top\nabla\bs g(\bs\beta_i)^\top
= O\left(\frac{1}{N}\sum_{i=1}^N\|\bs{\widehat \beta}_i - \widehat{\bs \beta}_i^{\star}\|_2\right) = o_p(1/G).
\]
Based on the above results, we simplify \eqref{eq:decomp} and obtain
\[
\frac{1}{N}\sum_{i=1}^N \bs f_i\left\{\bs g(\widehat{\bs \beta}^\star_i)-\bs g(\bs \beta_i)\right\}^\top 
= \frac{1}{N}\sum_{i=1}^N \bs f_i(\bs{\widehat \beta}_i - \bs \beta_i)^\top\nabla\bs g(\bs\beta_i)^\top + O_p(1/G).
\]
Notice that $\widehat {\bs \beta}_i - \bs \beta_i = \widehat{\bs \Omega}^{-1}\bs\phi(\bs \beta_i)/G$, so 
\begin{align*}
& \frac{1}{N}\sum_{i=1}^N\nabla\bs g(\bs\beta_i)(\bs{\widehat \beta}_i - \bs \beta_i) \bs f_i^\top
= \frac{1}{NG}\sum_{i=1}^N\nabla\bs g(\bs\beta_i)\widehat{\bs \Omega}^{-1}\bs\phi(\bs \beta_i)\bs f_i^\top \\
&= \frac{1}{NG}\sum_{i=1}^N\nabla\bs g(\bs\beta_i){\bs \Omega}^{-1}\bs\phi(\bs \beta_i)\bs f_i^\top + \frac{1}{NG}\sum_{i=1}^N\nabla\bs g(\bs\beta_i)(\widehat{\bs \Omega}^{-1} - {\bs \Omega}^{-1}) \bs\phi(\bs \beta_i)\bs f_i^\top
\end{align*}
As both $\|\nabla\bs g(\bs\beta_i)\|_2$ and $\|\bs f_i\|_2$ are uniformly bounded across $i$,
we have
\begin{align*}
&\left\|\frac{1}{NG}\sum_{i=1}^N\nabla\bs g(\bs\beta_i)(\widehat{\bs \Omega}^{-1} - {\bs \Omega}^{-1}) \bs\phi(\bs \beta_i)\bs f_i^\top\right\|_2\\
&= O\left(\frac{1}{G}\|\widehat{\bs \Omega}^{-1} - {\bs \Omega}^{-1}\|_2 \frac{1}{N}\sum_{i=1}^N\|\bs\phi(\bs \beta_i)\|_2 \right)\\ 
&= O_p(1/G),
\end{align*}
where the last equality is based on Lemma~\ref{lem:con1}-\ref{lem:gap} and \ref{lem:phi_unif}.
So finally, we simply \eqref{eq:simplified-goal} to 
\begin{equation}\label{eq:decomp_final}
\frac{1}{N}\sum_{i=1}^N\nabla\bs g(\bs\beta_i)(\bs{\widehat \beta}_i - \bs \beta_i) \bs f_i^\top
= \frac{1}{NG}\sum_{i=1}^N\nabla\bs g(\bs\beta_i){\bs \Omega}^{-1}\bs\phi(\bs \beta_i)\bs f_i^\top + O_p(1/G).
\end{equation}
Accordingly, we only need to focus on proving that when $\bs A_0=\bs 0$
\[
\frac{1}{NG}\sum_{i=1}^N\nabla\bs g(\bs\beta_i)\bs\Omega^{-1}\bs\phi(\bs \beta_i)\bs f_i^\top = O_p(1/\sqrt{NG}) + O_p(1/G).
\]

As $\bs\phi(\bs \beta_i) = \widehat{\bm U}^\top \bm W \bm \epsilon_i' - \bm H \bm{\beta}_i$, we have 
\begin{equation}\label{eq:two_terms}
\sum_{i=1}^N\nabla\bs g(\bs\beta_i)\bs\Omega^{-1}\bs\phi(\bs \beta_i)\bs f_i^\top =  \sum_{i=1}^N\nabla\bs g(\bs\beta_i){\bs \Omega}^{-1}\widehat{\bm U}^\top \bm W \bm \epsilon_i'\bs f_i^\top - \sum_{i=1}^N\nabla\bs g(\bs\beta_i){\bs \Omega}^{-1}\bm H \bm{\beta}_i\bs f_i^\top
\end{equation}
We prove for each of the two terms. For the first term, first notice that in the proof of Lemma~\ref{lem:phi_unif}, we have already shown that 
$$\frac{1}{N}\sum_{i = 1}^N\left\|\widehat{\bm U}^{ T} \bm W \bm \epsilon_i' - \widehat{\bm U}^{\star \top} \bm W \bm \epsilon_i'\right\|_2 = O_p(1).$$
So given that $\|\nabla\bs g(\bs\beta_i)\|_2$s and $\|\tilde{\bs f}_i\|_2$s are uniformly bounded across $i$,
$$\sum_{i=1}^N\nabla\bs g(\bs\beta_i)\bs\Omega^{-1}\bs\phi(\bs \beta_i)\bs f_i^\top =  \sum_{i=1}^N\nabla\bs g(\bs\beta_i){\bs \Omega}^{-1}\widehat{\bm U}^{\star \top} \bm W \bm \epsilon_i'\bs f_i^\top + O_p(N).$$
Because $\widehat{\bm U}^{\star}$, $\bs g(\bs \beta_i)$ and $\bs \epsilon_i'$ are mutually independent based on Assumption~\ref{asp:homo_popu} and \ref{assmp:random_prop}, and $\EE{\bs \epsilon_i'} = \bs 0$ for each $i$, we have for any $i_1\neq i_2$,
\[
\Cov{\nabla\bs g(\bs\beta_{i_1}){\bs \Omega}^{-1}\widehat{\bm U}^{\star \top} \bm W \bm \epsilon_{i_1}', \nabla\bs g(\bs\beta_{i_2}){\bs \Omega}^{-1}\widehat{\bm U}^{\star \top} \bm W \bm \epsilon_{i_2}'} = \bs 0.
\]
In the proof of Lemma~\ref{lem:phi_unif}, we have also shown that
$\max_i\mathbb{E}\left\|\widehat{\bm U}^{\star \top} \bm W \bm \epsilon_i'\right\|_2^2 = O(G)$. So,
\[
\Var{\sum_{i=1}^N\nabla\bs g(\bs\beta_i){\bs \Omega}^{-1}\widehat{\bm U}^{\star \top} \bm W \bm \epsilon_i'\bs f_i^\top} = \sum_{i=1}^N \Var{\nabla\bs g(\bs\beta_i){\bs \Omega}^{-1}\widehat{\bm U}^{\star \top} \bm W \bm \epsilon_i'\bs f_i^\top} = O(NG).
\]
This indicates that $\sum_{i=1}^N\nabla\bs g(\bs\beta_i){\bs \Omega}^{-1}\widehat{\bm U}^{\star \top} \bm W \bm \epsilon_i'\bs f_i^\top = O_p(\sqrt{NG})$, so the first term of \eqref{eq:two_terms}
\begin{equation}\label{eq:first_term}
    \sum_{i=1}^N\nabla\bs g(\bs\beta_i){\bs \Omega}^{-1}\widehat{\bm U}^\top \bm W \bm \epsilon_i'\bs f_i^\top = O_p(\sqrt{NG}) + O_p(N).
\end{equation}
For the second term of \eqref{eq:two_terms}, note that $\nabla\bs g(\bs\beta_i) = \frac{1}{\bs\beta_i^\top \bs 1}\left(\bs I - \frac{\bs\beta_i\bs 1^\top}{\bs\beta_i^\top \bs 1}\right)$, so
\[
\nabla\bs g(\bs\beta_i){\bs \Omega}^{-1}\bm H \bm{\beta}_i\bs f_i^\top 
= (\bs I - \bs p_i \bs 1^\top) {\bs \Omega}^{-1}\bm H \bm{p}_i\bs f_i^\top \in \mathbb R^{K \times (S+1)}.
\]
Then, by the property of Kronecker product,
\[
\vect\left(\sum_{i=1}^N\nabla\bs g(\bs\beta_i){\bs \Omega}^{-1}\bm H \bm{\beta}_i\bs f_i^\top \right)
=\sum_{i=1}^N(\bs f_i\bs p_i^\top)\otimes(\bs I - \bs p_i \bs 1^\top)\vect\left({\bs \Omega}^{-1}\bm H\right) 
\]
Since $\bs H=O_p(\sqrt{G})$, we only need to check the order of $\sum_{i=1}^N(\bs f_i\bs p_i^\top)\otimes(\bs I - \bs p_i \bs 1^\top)$. The Frobenius norm of each term is
\begin{align*}
&\left\|\bs f_i\bs p_i^\top\otimes(\bs I - \bs p_i \bs 1^\top)\right\|_F^2
=\mathrm{tr}\left[\left\{(\bs p_i\bs f_i^\top)\otimes(\bs I - \bs 1\bs p_i ^\top)\right\}\left\{(\bs f_i\bs p_i^\top)\otimes(\bs I - \bs p_i \bs 1^\top)\right\}\right]\\
&=\mathrm{tr}\left\{\|\bs f_i\|_2^2(\bs p_i\bs p_i^\top)\otimes(\bs I - \bs 1\bs p_i ^\top)(\bs I - \bs p_i \bs 1^\top)\right\}
=\|\bs f_i\|_2^2\mathrm{tr}\left(\bs p_i\bs p_i^\top\right)\mathrm{tr}\left\{(\bs I - \bs 1\bs p_i ^\top)(\bs I - \bs p_i \bs 1^\top)\right\}\\
&=\|\bs f_i\|_2^2\|\bs p_i\|_2^2\left(K-2+\|\bs p_i\|_2^2K\right)
\le2\|\bs f_i\|_2^2(K-1)\le 2 C_1^2(K-1),
\end{align*}
where the inequality uses the fact that $\|\bs p_i\|_2^2\le\|\bs p_i\|_1=1$ and $\max_{i}\|\bs f_i\|_2\le C_1$ by Assumption~\ref{assmp:group}c. Since $\bs p_i$s are mutually independent,
\begin{align*}
&\mathrm{Var}\left[\sum_{i=1}^N\vect \left\{(\bs f_i\bs p_i^\top)\otimes(\bs I - \bs p_i \bs 1^\top)\right\}\right]
=\sum_{i=1}^N\mathrm{Var}\left[\vect\left\{(\bs f_i\bs p_i^\top)\otimes(\bs I - \bs p_i \bs 1^\top)\right\}\right]\\
&\le\sum_{i=1}^N\mathbb{E}\left\|(\bs f_i\bs p_i^\top)\otimes(\bs I - \bs p_i \bs 1^\top)\right\|_F^2
=O(N).
\end{align*}
Since under the global null, $\bs p_i$s share the same mean and covariance matrix, with centered features such that $\sum_{i=1}^N\bs f_i=0$,
\[
\mathbb{E}\left\{\sum_{i=1}^N (\bs f_i\bs p_i^\top)\otimes(\bs I-\bs p_i\bs 1^\top)\right\}
=\sum_{i=1}^N\left\{\bs f_i\mathbb{E}(\bs p_i)^\top\right\}\otimes\bs I-\sum_{i=1}^N\bs f_i\mathbb{E}\left\{\bs p_i^\top\otimes(\bs p_i\bs 1^\top)\right\}=\bs 0.
\]
Then by the Chebyshev inequality, it holds that
\[
\sum_{i=1}^N(\bs f_i\bs p_i^\top)\otimes(\bs I-\bs p_i\bs 1^\top)=O_p(\sqrt{N}).
\]
Therefore,
\begin{equation}\label{eq:second_term}
    \sum_{i=1}^N\nabla\bs g(\bs\beta_i){\bs \Omega}^{-1}\bm H \bm{\beta}_i\bs f_i^\top = O_p(\sqrt{NG})
\end{equation}
Combining \eqref{eq:a-diff-first-term}, \eqref{eq:decomp_final}, \eqref{eq:two_terms}, \eqref{eq:first_term} and \eqref{eq:second_term}, we have
$$\widehat{\bs A} - \bs A_N = O_p\left(\frac{1}{\sqrt{NG}}\right) + O_p\left(\frac{1}{G}\right)+O_p\left(\frac{1}{N}\right) = o_p\left(\frac{1}{\sqrt N}\right).$$
when $N/G^2 \to 0$.
\end{proof}

\begin{remark}
The condition that $\bs A_0 = \bs 0$ is only used to bound the second term \eqref{eq:second_term}. So in the general case when the null $\bs A_0 = \bs 0$ does not hold, we have 
$$\widehat{\bs A} - \bs A_0 = - \frac{1}{NG}\sum_{i=1}^N(\bs I - \bs p_i\bs 1^\top){\bs \Omega}^{-1}\bm H \bm{p}_i\bs f_i^\top + {\bs A}_N - \bs A_0 + O_p(1/G)+O_p(1/N)$$
where 
$$-\frac{1}{NG}\sum_{i=1}^N(\bs I - \bs p_i\bs 1^\top){\bs \Omega}^{-1}\bm H \bm{p}_i\bs f_i^\top = O_p(1/{\sqrt G}).$$
We can still establish the asymptotic normality of $\widehat{\bs A} - \bs A_0$ using our previous proof techniques, although estimating its asymptotic variance in practice will be very challenging.
\end{remark}

\end{document}